\pdfoutput=1
\documentclass[letterpaper,11pt]{article}
\RequirePackage{silence}
\RequirePackage{float}
\usepackage[ascii]{inputenc}
\usepackage[bookmarksnumbered,hypertexnames=false,colorlinks,colorlinks=true,linkcolor=blue,urlcolor=black,citecolor=blue,anchorcolor=green]{hyperref}
\usepackage{bbm,braket,microtype,mathrsfs,amsmath,amssymb,color,amsthm,mathtools,fullpage,graphicx,enumitem,ae,aecompl,bm,thm-restate}
\usepackage[T1]{fontenc}
\usepackage[USenglish]{babel}
\usepackage[capitalize]{cleveref}
\usepackage{silence}
\usepackage{newpxtext}
\usepackage[euler-digits,euler-hat-accent]{eulervm}
\allowdisplaybreaks
\crefname{algo}{Algorithm}{Algorithms}\crefname{fig}{Figure}{Figures}
\newtheorem{thm}{Theorem}[section]\crefname{thm}{Theorem}{Theorems}
\newtheorem{lem}[thm]{Lemma}\crefname{lem}{Lemma}{Lemmas}
\newtheorem{prp}[thm]{Proposition}\crefname{prp}{Proposition}{Propositions}
\newtheorem{rem}[thm]{Remark}\crefname{rem}{Remark}{Remarks}
\newtheorem{cor}[thm]{Corollary}\crefname{cor}{Corollary}{Corollaries}
\newtheorem{dfn}[thm]{Definition}\crefname{dfn}{Definition}{Definitions}
\newtheorem{exa}[thm]{Example}\crefname{exa}{Example}{Examples}
\newtheorem*{thm*}{Theorem}\crefname{thm}{Theorem}{Theorems}
\newtheorem*{lem*}{Lemma}\crefname{lem}{Lemma}{Lemmas}
\numberwithin{equation}{section}
\newcommand{\CC}{\mathbb C}
\newcommand{\RR}{\mathbb R}
\newcommand{\ZZ}{\mathbb Z}
\newcommand{\NN}{\mathbb N}
\newcommand{\EE}{\mathbb E}
\newcommand{\FF}{\mathbb F}
\newcommand{\ot}{\otimes}
\newcommand{\op}{\oplus}
\newcommand{\eps}{\varepsilon}
\newcommand{\id}{\mathbbm 1}

\renewcommand{\vec}{\mathbf}

\DeclareMathOperator{\Sym}{Sym}
\DeclareMathOperator{\Alt}{Alt}
\DeclareMathOperator{\rank}{rank}
\DeclareMathOperator{\ran}{ran}

\DeclareMathOperator{\vecmap}{vec}
\DeclareMathOperator{\tr}{tr}
\DeclareMathOperator{\sn}{sn}
\DeclareMathOperator{\CSS}{CSS}
\DeclareMathOperator{\Span}{span}
\DeclareMathOperator{\Cliff}{Cliff}
\DeclareMathOperator{\Sp}{Sp}
\DeclareMathOperator{\Stab}{Stab}
\DeclareMathOperator{\Iso}{Iso}
\floatstyle{boxed}\newfloat{algo}{htbp}{alg}\floatname{algo}{Algorithm~}

\begin{document}

%-----------------------------------------------------------------------------
\title{Schur-Weyl duality for the Clifford group with applications:\\
Property testing, a robust Hudson theorem, and de Finetti representations}
\author{ \hspace{0.35cm}David Gross%
\thanks{Institute for Theoretical Physics, University of Cologne, and
Centre for Engineered Quantum Systems, School of Physics, The University of Sydney,
\href{mailto:david.gross@thp.uni-koeln.de}{\texttt{david.gross@thp.uni-koeln.de}}.
\;$^\dagger$Stanford Institute for Theoretical Physics, Stanford University, and Kavli Institute for Theoretical Physics, \href{mailto:nezami@stanford.edu}{\texttt{nezami@stanford.edu}}.
\;$^\ddagger$Korteweg-de Vries Institute for Mathematics, Institute of Theoretical Physics, and Institute for Logic, Language and Computation, QuSoft, University of Amsterdam, \href{mailto:m.walter@uva.nl}{\texttt{m.walter@uva.nl}}.
}~
\qquad\,\, Sepehr Nezami$^\dagger$\qquad\, Michael Walter$^{\ddagger,\dagger}$}
\date{}
\maketitle
\thispagestyle{empty}
\begin{abstract}
Schur-Weyl duality is a ubiquitous tool in quantum information.
At its heart is the statement that the space of operators that commute with the $t$-fold tensor powers $U^{\otimes t}$ of all unitaries~$U\in U(d)$ is spanned by the permutations of the $t$ tensor factors.
In this work, we describe a similar duality theory for tensor powers of \emph{Clifford unitaries}.
The Clifford group is a central object in many subfields of quantum information, most prominently in the theory of fault-tolerance.
The duality theory has a simple and clean description in terms of finite geometries.
We demonstrate its effectiveness in several applications:
\begin{itemize}
\item We resolve an open problem in \emph{quantum property testing} by showing that ``stabilizerness'' is efficiently testable:
There is a protocol that, given access to six copies of an unknown state, can determine whether it is a stabilizer state, or whether it is far away from the set of stabilizer states.
We give a related
membership test for the Clifford group.
\item We find that tensor powers of stabilizer states have an increased symmetry group.
Conversely, we provide corresponding \emph{de Finetti theorems}, showing that the reductions of arbitrary states with this symmetry are well-approximated by mixtures of stabilizer tensor powers (in some cases, exponentially well).
\item
We show that the distance of a pure state to the set of stabilizers can be lower-bounded in terms of the sum-negativity of its Wigner function.
This gives a new quantitative meaning to the sum-negativity (and the related \emph{mana}) -- a measure relevant to fault-tolerant quantum computation.
The result constitutes a \emph{robust} generalization of the \emph{discrete Hudson theorem}.
\item We show that complex projective designs of arbitrary order can be obtained from a finite number (independent of the number of qudits) of Clifford orbits.
To prove this result, we give explicit formulas for arbitrary moments of random stabilizer states.
\end{itemize}
\end{abstract}
\clearpage
\tableofcontents
\clearpage
\setcounter{page}{1}
%-----------------------------------------------------------------------------

%=============================================================================
\section{Introduction}\label{sec:intro}
%=============================================================================

%-----------------------------------------------------------------------------
\subsection{Background}
%-----------------------------------------------------------------------------

\paragraph{Schur-Weyl duality.}
To motivate the symmetry this work is based on, we start by considering two types of problems that have frequently appeared in quantum information theory.
First, assume that we have access to $t$ copies $\rho^{\otimes t}$ of an unknown quantum state $\rho$ on $\CC^d$, and that we are interested in some property of $\rho$'s eigenvalues (for example its entropy).
Clearly, then, the problem has a $U^{\otimes t}$-symmetry in the sense that the inputs $\rho^{\otimes t}$ and
\begin{equation*}
% \rho^{\otimes t}
% \text{ and }
  U^{\otimes t} (\rho^{\otimes t}) {U^\dagger}^{\otimes t}
\end{equation*}
represent equivalent properties.
It thus makes sense to design a procedure that shares the $U^{\otimes t}$-symmetry, and indeed the resulting procedure has been shown to be optimal for estimating the eigenvalues~\cite{keyl2001estimating,hayashi2002quantum,christandl2006spectra,christandl2007nonzero,odonnell2015quantum}.
Moreover, consider \emph{quantum state tomography}, the task of estimating the entire quantum state $\rho$.
Essentially optimal estimators can be constructed by first estimating the eigenvalues and then the eigenbasis~\cite{odonnell2016efficient,odonnell2017efficient,haah2017sample}, crucially using the structure of $U^{\ot t}$ in each step.
There are many further problems in quantum information where this symmetry can be exploited -- for example in quantum Shannon theory, where optimal rates mostly depend only on the eigenvalues of the quantum state~\cite{hayashi2002quantum,harrow2005applications}.

Second, studying the properties of a Haar-random state vector $\ket\psi$ has proven to be extremely fruitful~\cite{hayden2006aspects,hastings2008counterexample}.
Instead of working with the full distribution, it is often sufficient to exploit information about the statistical \emph{moments} of the random matrix $\ket\psi\bra\psi$.
The $t$-th moment is described by the expected value of the $t$-th tensor power of the random matrix:
\begin{equation}\label{eq:t-th moment Haar}
  M_t = \EE_{\psi \text{ Haar}} [ (\ket\psi\bra\psi)^{\otimes t} ].
\end{equation}
Again, $M_t$ is invariant under conjugation by $U^{\otimes t}$, $U\in U(\CC^d)$.

The importance of Schur-Weyl duality in quantum information stems from the fact that it allows one to characterize the set of $U^{\otimes t}$-invariant operators on $(\CC^d)^{\otimes t}$.
Indeed, it implies that any such operator can be expressed as the linear combination of matrices $r_\pi$, $\pi\in S_t$ that act by permuting the tensor factors:
\begin{equation}\label{eq:perm_factors}
  r_\pi \, (\ket{\psi_1} \otimes \dots \otimes \ket{\psi_t}) = \ket{\psi_{\pi_1}} \otimes  \dots \otimes \ket{\psi_{\pi_t}}.
\end{equation}

\paragraph{Clifford group and stabilizer states.}
Arguably, the subgroup of the full unitary group that is most important to quantum information is the Clifford group.
The Clifford group and the closely related concept of stabilizer states and stabilizer codes feature centrally in fault-tolerant quantum computing, quantum coding in general, randomized benchmarking, measurement-based quantum computing, and many other subfields of quantum information.

To introduce the Clifford group, we first recall the definition of the set of \emph{Pauli operators}.
For a \emph{qudit} ($d$-dimensional system), they are defined by their action on a some basis $\{\ket{q}\}_{q=0}^{d-1}$ via
\begin{align*}
  X \ket q = \ket{q+1}, \quad Z \ket q = e^{2\pi i q/d} \ket q.% \quad Y = iXZ.
\end{align*}
For $n$ qudits, the Pauli group is defined as the finite group generated by the Pauli operators on each qudit.
The \emph{Clifford group} now is the natural symmetry group of the Pauli group.
That is, a unitary $U$ is Clifford if, for any Pauli operator $P$, $UPU^\dagger$ is again in the Pauli group.
Ignoring overall phases, the Clifford group is a finite group, which is intimately connected to the \emph{metaplectic representation} of the discrete \emph{symplectic group} (see, e.g.,~\cite{gross2006hudson}).
Closely related to the Clifford group is the set of \emph{stabilizer states}.
These are the states that can be obtained by acting on a basis vector $\ket{0\dots0}$ by arbitrary Clifford unitaries.

As before, there are many natural problems that are invariant under $U^{\otimes t}$, for $U$ a Clifford unitary.
Two examples we will discuss are: (1) Given access to $\psi^{\otimes t}$, decide whether $\psi$ is a stabilizer state; (2) What are the $t$-th moments of a random stabilizer state $\psi$?

\paragraph{Randomized constructions.}
Another motivation arises from randomized constructions. Unitaries and states drawn from the Haar measure appear in many situations, including in quantum cryptography, coding, and data hiding~\cite{hayden2006aspects}.
While randomized constructions are often near-optimal and frequently out-perform all known deterministic constructions, they have the drawback that generic quantum states cannot be efficiently prepared.

This contrasts with random Clifford unitaries and random stabilizer states, both of which can be efficiently realized (they require at most $O(n^2)$ gates to implement in a quantum circuit)~\cite{aaronson2004improved}.
They have therefore repeatedly been suggested as ``drop-in replacements'' for their Haar-measure analogues.
Examples include randomized benchmarking~\cite{magesan2011scalable,helsen2017multiqubit}, low-rank recovery~\cite{kueng2016low}, and tensor networks in the context of holography~\cite{hayden2016holographic,nezami2016multipartite}.
All these applications require information about the moments (in the sense of \cref{eq:t-th moment Haar}) of random stabilizer states, which they all obtain from representation-theoretic data.
To date, this representation theory and the associated stabilizer moments are understood only up to order $t=4$~\cite{zhu2016clifford,helsen2016representations,nezami2016multipartite}.
This contrasts with the Haar-random case, where Schur-Weyl duality gives this information for arbitrary orders $t$.
Making analogous techniques available for the Clifford case was one important motivation for this work.
Higher moments will generally lead to tighter performance bounds in randomized constructions, and are strictly required for some applications, like the \emph{stabilizer testing problem} resolved here.

%-----------------------------------------------------------------------------
\subsection{Schur-Weyl duality for the Clifford group}\label{subsec:schur weyl intro}
%-----------------------------------------------------------------------------
We start with an explicit description of the commutant of tensor powers of Clifford unitaries.
While such a description has not yet appeared in the quantum information literature, we emphasize that some of the key results can already be deduced from work by Nebe, Rains, Sloane and colleagues on invariants of self-dual codes (see the excellent monograph~\cite{nebe2006self}).
Also, in representation theory, there is a separate stream of closely related work regarding the structure of the oscillator representation and attempts to develop a Howe duality theory over finite fields, which is still an open problem (see, e.g.,~\cite{howe1973invariant,gurevich2016small} and references therein).
We discovered the approach presented below independently, starting from our results in~\cite[App.~C]{nezami2016multipartite} for third tensor powers.
Our proofs differ fundamentally from the preceding works in that they rely on the phase space formalism of finite-dimensional quantum mechanics, which offers additional insight.

To construct the commutant, start with the permutations $r_\pi$ on $(\CC^d)^{\otimes t}$ of Eq.~(\ref{eq:perm_factors}).
We assume for now that $\CC^d$ is the Hilbert space of a single qudit with ``computational basis'' $\{\ket{x}\}_{x\in\ZZ_d}$ labeled by elements in $\ZZ_d = \ZZ/d \ZZ$ (this is anyway required for defining the Pauli and the Clifford group).
Basis elements $\ket{\vec x}=\ket{x_1}\otimes\dots\otimes \ket{x_t}$ of $(\CC^d)^{\otimes t}$ are then labeled by vectors $\vec x \in \ZZ_d^t$.
In this language:
\begin{equation}\label{eq:pi_T}
  r_\pi
  =
  \sum_{\vec y\in\ZZ_d^t} \ket{\pi(\vec y)}\bra{\vec y}
  =
  \sum_{(\vec x, \vec y)\in T_\pi} \ket{\vec x}\bra{\vec y},
\end{equation}
where $T_\pi=\{ (\pi(\vec y), \vec y) : y\in\ZZ_d^t\}$ and $\pi$ permutes the components of $\vec y$.
Because the Clifford group is a subgroup of the unitaries, the commutant is in general strictly larger.
We thus have to add further operators to the $r_\pi$'s in order to find a complete set.

The central message of this section is that, surprisingly, a minor modification of (\ref{eq:pi_T}) suffices!
Indeed, for any subspace $T$ of $\ZZ_d^t \op \ZZ_d^t$ define
\begin{align*}
  r(T) = \sum_{(\vec x,\vec y)\in T} \ket{\vec x}\bra{\vec y}.
\end{align*}
We also consider the $n$-fold tensor power $R(T)\coloneqq r(T)^{\ot n}$, which is an operator on
$
((\CC^d)^{\ot t})^{\ot n} \cong (\CC^d)^{\ot t n} \cong ((\CC^d)^{\ot n})^{\ot t}
$.
% \end{align*}

We now single out subspaces that satisfy certain geometric properties.
Reflecting a well-known difference between even and odd dimensions in the stabilizer formalism, we define $D=d$ if $d$ is odd, and $D=2d$ if $d$ is even.

% [Lagrangian, stochastic, $\Sigma_{t,t}$]
\begin{restatable*}[$\Sigma_{t,t}$]{dfn}{restateDfnLagrangianStochastic}\label{dfn:lagrangian stochastic}
Consider the quadratic form $\mathfrak q\colon \ZZ_d^{2t} \to \ZZ_D$ defined by $\mathfrak q(\vec x,\vec y) \coloneqq \vec x\cdot\vec x - \vec y\cdot\vec y$.%
\footnote{Note that for $x\in\ZZ_d$, $x^2$ is well-defined modulo~$D$.}
We denote by $\Sigma_{t,t}(d)$ the set of subspaces $T\subseteq\ZZ_d^{2t}$ satisfying the following properties:
\begin{enumerate}
\item
$T$ is \emph{totally $\mathfrak q$-isotropic}: i.e.,
$\vec x\cdot\vec x = \vec y\cdot\vec y \pmod D$ for all $(\vec x,\vec y)\in T$.
\item $T$ has dimension~$t$ (the maximal possible dimension).
\item $T$ is \emph{stochastic}: $\vec 1_{2t} = (1,\dots,1) \in T$.
\end{enumerate}
We will summarize the first two conditions by saying that $T$ is \emph{Lagrangian}.
Thus, we will call $\Sigma_{t,t}(d)$ the set of \emph{stochastic Lagrangian subspaces}.
\end{restatable*}

Our first main result is the following theorem, which states that the operators $R(T)$ obtained from these subspaces are a basis of the commutant:

\begin{restatable*}[Commutant of Clifford tensor powers]{thm}{restateThmCommutant}\label{thm:commutant}
  Let $d$ be a prime and $n\geq t-1$.
  Then the operators $R(T)=r(T)^{\ot n}$ for $T \in \Sigma_{t,t}(d)$ are $\prod_{k=0}^{t-2} (d^k+1)$ many linearly independent operators that span the commutant of the $t$-th tensor power action of the Clifford group for $n$ qudits.
\end{restatable*}
\begin{proof}[Proof sketch]
We use the phase space formalism of finite-dimensional quantum mechanics developed in~\cite{wootters1987wigner,appleby2005sic,gross2006hudson,gross2008quantum,beaudrap2013linearized}.
In particular, Clifford unitaries have a simple description on phase space: they act by affine symplectic transformations.

We use this structure to give a concise proof that the operators $R(T)$ commute with $U^{\ot t}$ for any Clifford unitary.
The linear independence is not hard, so it remains to argue that the number of subspaces equals the dimension of the commutant.
We show this by a careful counting argument.
We first compute the number of stochastic Lagrangian subspaces.
Employing the fundamental Witt's theorem, we find recursive relations for the dimension of commutant of the Clifford group.
We solve this recursion using Gaussian binomial identities (the result is a generalization of
\cite[(8)--(10)]{zhu2015multiqubit}) and find that the cardinalities match, concluding the proof.
\end{proof}

There is a rich structure associated with the objects appearing in this theorem:
It is easy to see that the spaces $T_\pi=\{(\pi(\vec y),\vec y)\}$ that give rise to the commutant of $U(d)$ appear as special cases above.
For general $d$ and $t$, not all $R(T)$'s are invertible.
In particular, for some $T$'s, $R(T)$ is proportional to the projection onto a stabilizer code.
This way, one can e.g., recover the code that has been used to describe the irreps contained in the 4th tensor power of the Clifford group in~\cite{zhu2016clifford}. % (cf.~\cite{helsen2016representations}).
The set of invertible $R(T)$'s are associated with spaces $T$ of the form $(A\vec y, \vec y)$, for $A$ that are elements of a certain ``stochastic orthogonal'' group $O_t(d)$.
This group is of interest to the formulation of modular Howe duality~\cite{gurevich2016small}, and underlies several of our applications below.

Remarkably, the size of the commutant \emph{stabilizes} as soon as $n\geq t-1$.
That is, just like the symmetric group in Schur-Weyl duality, the set that parametrizes the commutant of the Clifford tensor powers is independent of the number $n$ of qu$d$its.
The fact that the operators~$R(T)=r(T)^{\ot n}$ are themselves tensor powers facilitates possible physical implementations.
This, once more, generalizes a property of the symmetric group in Schur-Weyl duality.
% Mention ``swap test''?

To find novel applications of this theory, it is helpful to identify a set of non-trivial $T$'s that afford an intuitive interpretation.
Several of our multi-qubit results presented below are based on spaces with elements~$(\bar\pi \vec y, \vec y)$, where $\bar\pi$ is what we refer to as an \emph{anti-permutation}.
An anti-permutation is simply the binary complement of a permutation matrix.
Formally, $\bar\pi = \vec 1_t \vec 1_t^T - \pi$, where $\vec 1_t=(1,\dots, 1)$ contains $t$ ones, and $\pi\in S_t$.
% For odd $d$ and $(d,t)=1$, the expression is~$\bar\pi=\frac 2t \vec 1 \vec 1^T - \pi$.
Its operator representation is particularly straightforward.
The $n$-qubit anti-identity, e.g., acts by
\begin{equation}\label{eq:anti identity}
  R(\bar\id) = 2^{-n} \left( I^{\otimes t} + X^{\otimes t} + Y^{\otimes t} + Z^{\otimes t} \right)^{\otimes n},
\end{equation}
which greatly facilitates the analysis (cf.\ \cref{eq:anti identity six,dfn:anti-permutation}).

%-----------------------------------------------------------------------------
\subsection{Quantum property testing: stabilizer testing}\label{subsec:stab test intro}
%-----------------------------------------------------------------------------
The theory of \emph{quantum property testing} asks which properties of a ``black box'' many-body quantum system can be learned efficiently---in particular without having to resort to costly full tomography~\cite{buhrman2003quantum,buhrman2008quantum,montanaro2013survey}.
A prototypical example of a testable property is \emph{purity}.
Indeed, given access to two copies $\rho\otimes\rho$ of an unknown quantum state $\rho$, the so-called \emph{swap test} provides for a simple protocol that accepts with certainty if $\rho=\ket\psi\bra\psi$ is pure, and rejects with probability~$\Theta(1/\eps^2)$ if $\rho$ is $\eps$-far away from the set of pure states in trace distance.
The test is \emph{perfectly complete} in the sense that it has a type-I error rate of zero (pure states are accepted with certainty); it requires a number of copies (two) that is independent of the dimension. It is also \emph{transversal} in the sense that if $\rho$ acts on $n$ qubits, all operations are required to be coherent only across the two copies, and factorize w.r.t.\ the $n$ qubits.

An open problem in this theory was whether \emph{stabilizerness} and \emph{Cliffordness} are testable properties of, respectively, states and unitaries~\cite{montanaro2013survey}.
Both properties are clearly Clifford-invariant---so by the arguments presented in the introduction, it makes sense to search for tests in the commutant of the Clifford group.
It is known that 2nd and 3rd moments of random stabilizer states are identical to the moments of Haar-random states~\cite{zhu2015multiqubit,kueng2015qubit,webb2016clifford}.
This implies that three copies of a state are not sufficient to test for stabilizerness, and the results of~\cite{zhu2016clifford} can be used to show that four copies are also insufficient for a dimension-independent theory.

\paragraph{Prior work.}
Prior to our results, the best known algorithms for stabilizer testing required a number of copies that scaled linearly with $n$, the number of qubits.
Indeed, these algorithms proceeded by attempting to \emph{identify} the stabilizer state, which necessarily requires $\Omega(n)$ copies by the Holevo bound~\cite{aaronson2008identifying,montanaro2017learning,zhao2016fast,krovi2008efficient}.
However, the existence of tests that require only a constant number of copies has been an important open question~\cite{montanaro2013survey}.
We note that the stabilizer testing problem asks whether a given state is \emph{any} stabilizer state -- which is distinct from the problem of verifying whether it equals some \emph{fixed} stabilizer state~\cite{hayashi2015verifiable}.

\paragraph{Our results.}
We show that for $n$ qudits $O(1)$ copies suffice to give an \emph{efficient}, \emph{perfectly complete}, \emph{dimension-independent}, and \emph{transversal} test.
For example, for qubits ($d=2$) our test requires only $6$ copies of the state to achieve a power independent of $n$ (\cref{algo:qubits}).
It requires coherent operations on only two qubits at a time, which means in particular that it can be implemented given a source that creates two copies of a fixed state at a time (\cref{fig:circuit qubits}).

First, we consider the problem for qubits.
Here our protocol affords an intuitive description using a new primitive which we call \emph{Bell difference sampling}.
Then we proceed to the general case and discuss the connection to the commutant of the Clifford group described in \cref{subsec:schur weyl intro}.

\subsubsection*{Qubits: Bell difference sampling}

We start with an intuitive motivation of the test.
Let $\ket\psi\bra\psi=2^{-n/2} \sum_{\vec a} c_{\vec a} W_{\vec a}$ be its expansion w.r.t.\ the Weyl operators $W_{\vec a}$ (which are just the Pauli operators labeled in the usual way by bitstrings $\vec a\in\ZZ_2^{2n}$, cf.\ \cref{sec:preliminaries}).
Now measure two copies of $\psi$ in the \emph{Bell basis} $\ket{W_{\vec x}}$ defined by applying the Weyl operators to the maximally entangled state, i.e., $\ket{W_{\vec x}} = (W_{\vec x} \ot I) \ket{\Phi^+}$ where $\ket{\Phi^+} = 2^{-n/2}\sum_{\vec q} \ket{\vec q,\vec q}$.
If $\psi$ is real in the computational basis then it is not hard to see that the measurement outcome is distributed according to the probability distribution $p_\psi(\vec a) = \lvert c_{\vec a}\rvert^2$.
This is known as Bell sampling~\cite{montanaro2017learning,zhao2016fast}.
Now stabilizer states are distinguished by the fact that they are eigenvectors of all Weyl operators $W_{\vec a}$ for which $\lvert c_{\vec a}\rvert^2\neq 0$ (these are its \emph{stabilizer group}).
This suggests using Bell sampling to obtain some $\vec a$, then measuring $W_{\vec a}$ twice on two fresh copies, and accepting $\psi$ as a stabilizer if the same eigenvalue is obtained twice.

\begin{algo}[t]
\raggedright
\textbf{Input:} Six copies of an unknown multi-qubit quantum state ($\psi^{\ot 6}$). \\[.2cm]
%\textbf{Algorithm:}
\begin{enumerate}
\item Perform \emph{Bell difference sampling}: That is, Bell sample twice (on two independent copies of $\psi^{\ot 2}$), with outcomes $\vec x$, $\vec y$, and set $\vec a = \vec x - \vec y$. % \pmod 2$.
\item Measure the Weyl operator $W_{\vec a}$ twice (on two independent copies of $\psi$). Accept iff the outcomes agree.
\end{enumerate}
\caption{Algorithm for testing whether an unknown multi-qubit state is a stabilizer state.}\label{algo:qubits}
\end{algo}

While we show that this works for real state vectors, Bell sampling unfortunately does not extend to complex state vectors.
To overcome this challenge, we introduce a new primitive:

% Formally, this amounts to the following projective measurement on $\mathcal H_n \ot \mathcal H_n$:
% \begin{align*}
%   \Pi_{\vec a}
% = \sum_{\vec x} \ket{W_{\vec x}}\bra{W_{\vec x}} \ot \ket{W_{\vec x+\vec a}}\bra{W_{\vec x+\vec a}},
% \end{align*}
% which we call \emph{Bell difference sampling}.

\begin{restatable*}[Bell difference sampling]{dfn}{restateDfnBellDifferenceSampling}\label{def:bell difference sampling}
We define \emph{Bell difference sampling} as performing Bell sampling twice and subtracting (adding) the results from each other (modulo two).
In other words, it is the projective measurement on four copies of a state, $\psi^{\ot4}\in((\CC^2)^{\ot n})^{\ot4}$, given by
\begin{align*}
  \Pi_{\vec a}
= \sum_{\vec x} \ket{W_{\vec x}}\bra{W_{\vec x}} \ot \ket{W_{\vec x+\vec a}}\bra{W_{\vec x+\vec a}}.
\end{align*}
\end{restatable*}

For stabilizer states (whether real or complex) it is easy to see that Bell difference sampling will always sample an element $\vec a$ corresponding to a Weyl operator $W_{\vec a}$ in its stabilizer group.
What is rather less obvious is that, even for arbitrary quantum states, Bell difference sampling still has a useful interpretation.
The following theorem shows that this is indeed the case: it amounts to sampling from the probability distribution $p_\psi$ twice and taking the difference.

\begin{restatable*}[Bell difference sampling]{thm}{restateThmBellDifferenceSamplingQubits}\label{thm:bell difference sampling qubits}
  Let $\psi$ be an arbitrary pure state of $n$ qubits.
  Then:
  \begin{align*}
    \tr\left[\Pi_{\vec a} \psi^{\ot 4}\right] = \sum_{\vec x} p_\psi(\vec x) p_\psi(\vec x+\vec a).
  \end{align*}
  If $\psi$ is a stabilizer state, say $\ket S\bra S$, then this is equal to $p_S(\vec a)$ from \cref{eq:characteristic distribution stabilizer}.
\end{restatable*}

Using Bell difference sampling as a primitive, we obtain the natural \cref{algo:qubits}.

\begin{restatable*}[Stabilizer testing for qubits]{thm}{restateThmMainQubits}\label{thm:main qubits}
  Let $\psi$ be a pure state of $n$ qubits.
  If $\psi$ is a stabilizer state then \cref{algo:qubits} accepts with certainty, $p_\text{accept}=1$.
  On the other hand, if $\max_S \lvert\braket{S|\psi}\rvert^2\leq1-\eps^2$ then $p_\text{accept}\leq1-\eps^2/4$.
\end{restatable*}
\begin{proof}[Proof sketch]
We want to show that if the success probability, $p_{\text{accept}}$, is close to one then $\psi$ has high overlap with a stabilizer state.
The proof proceeds in two steps.
First, we analyze the success probability and show that if $p_{\text{accept}} \approx 1$ then $p_\psi(\vec a)$ is typically close to its maximum possible value $2^{-n}$.
Next, we use Markov's inequality to find a large set of $\vec a$ where $p_\psi(\vec a)>\frac12 2^{-n}$.
Using a version of uncertainty principle (see~\cref{fig:bloch plane}), we show that the corresponding Weyl operators $W_{\vec a}$ necessarily commute, and therefore form a stabilizer subgroup.
This finally means that our initial state must have a large overlap with a corresponding stabilizer state.
\end{proof}

\begin{figure}
\begin{center}
\includegraphics[height=4cm]{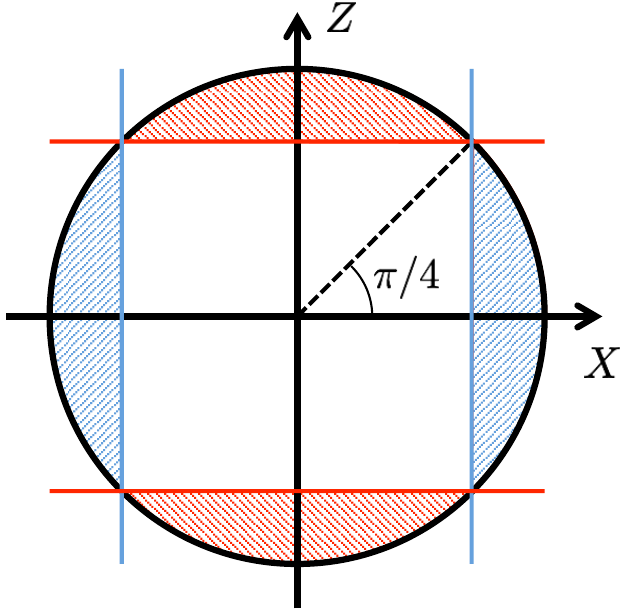}
\caption{The $X$-$Z$-plane of the Bloch sphere.
The area shaded in red indicates the projection of those states $\rho$ for which $\lvert\tr Z \rho\rvert>\sin\frac\pi4=\frac1{\sqrt 2}$.
Likewise, the blue area correspond to the states with $\lvert\tr X \rho\rvert>\frac1{\sqrt 2}$.
As the two areas do not intersect, these two conditions cannot be simultaneously satisfied.
This is a manifestation of the uncertainty principle.}
% In terms of the fourth power of the projection, the condition becomes $|\tr Z \rho|^4>\frac14=1-\frac34$.}
\label{fig:bloch plane}
\end{center}
\end{figure}

\cref{thm:main qubits} solves the stabilizer testing conjecture for qubits.
It also implies a number of interesting corollaries. % which might be of independent interest.
E.g., it directly follows that one can also test Cliffordness of a unitary efficiently, without given black-box access to the inverse as in~\cite{low2009learning,wang2011property}; this resolves another open problem from~\cite{montanaro2013survey}.
From a structural point of view, it shows that the Clifford group is the solution, within $U(2^n)$, of a set of polynomial equations of order $6$.
Our result is optimal in the sense that there exist no perfectly complete tests for fewer than six copies that achieve statistical power independent of the number of qubits (see \cref{sec:moments} and~\cite{rajamsc}).

\begin{figure}
\begin{center}
\includegraphics[width=9cm]{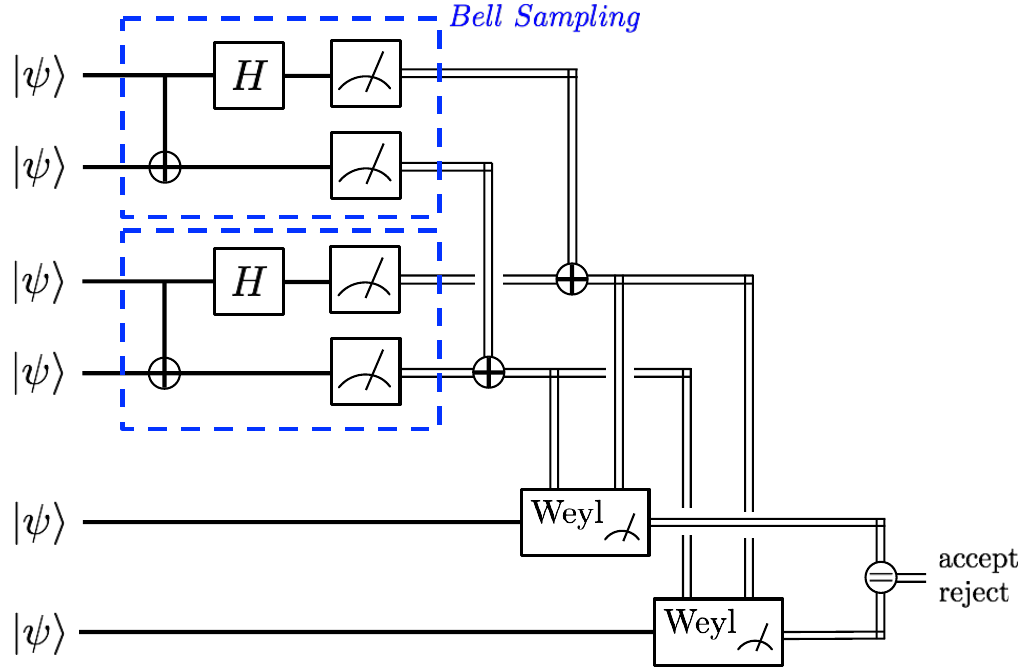}
\caption{
Quantum circuit implementing \cref{algo:qubits} for qubit stabilizer testing.
Inside the blue blocks: The quantum gates denote the controlled-NOT and the Hadamard gate, respectively; the measurements are in the $n$-qubit computational basis.
Outside the blue blocks:
Double lines represent classical information.
The ``$\oplus$''-operation is addition modulo two.
The boxes labeled ``Weyl'' perform a two-outcome measurement with respect to the eigenspaces of $W_{\vec a}$, where $\vec a$ is determined by classical inputs.
For $n$ qubits, the circuit is fully transversal in the sense that all operations
are required to be coherent only across two copies, and factorize with respect to the $n$ qubits.}\label{fig:circuit qubits}
\end{center}
\end{figure}

\subsubsection*{Qudits}
A careful analysis of the measurement of \cref{algo:qubits} shows that it is equivalent to a projective measurement of the form $\Pi_{\text{accept}}=\frac12 \left( I + V \right)$, where $V$ is the following Hermitian unitary operator:
\begin{equation}\label{eq:V for qubits intro}
  V  =2^{-n} \sum_{\vec x} W_{\vec x}^{\ot 6}. %=\left(\frac12\left( \sum_{\vec x\in\mathcal V_1} W_{\vec x}\right)^{\ot 6}\right)^{\ot n} .
\end{equation}
It can be readily seen that the operator~\cref{eq:V for qubits intro} commutes with tensor powers of Clifford unitaries.

In fact, as discussed earlier, it is natural to approach the stabilizer testing problem by measuring operators in the commutant of the Clifford group.
Since the stabilizer states are a single orbit under the Clifford group, any such measurement by design will have the same level of significance on all stabilizer states.

\Cref{eq:V for qubits intro} and corresponding measurement have a clear generalization to arbitrary qudits.
Let $d\geq2$ and consider the operators
\begin{align}\label{eq:V for qudits intro}
 \Pi_{s,\text{accept}} = \frac12(I+V_s)\quad \text{where}  \quad   V_s = d^{-n} \sum_{\vec x} (W_{\vec x} \ot W_{\vec x}^\dagger)^{\ot s}.
\end{align}
One can see that if $(d,s)=1$, $V_s$ is a Hermitian unitary and so $\Pi_{s,\text{accept}}$ is a projector.
We now state our general stabilizer testing result:

\begin{restatable*}[Stabilizer testing for qu$d$its]{thm}{restateThmMainQudits}\label{thm:main qudits}
Let $d\geq2$ and choose $s\geq2$ such that $(d,s)=1$.
Let $\psi$ be a pure state of $n$ qu$d$its and denote by $p_\text{accept}=\tr[\psi^{\ot2s}\Pi_{s,\text{accept}}]$ the probability that the POVM element~$\Pi_{s,\text{accept}}$ accepts given $2s$ copies of $\psi$.
If $\psi$ is a stabilizer state then it accepts with certainty, $p_\text{accept}=1$.
On the other hand, if $\max_S\lvert\braket{S|\psi}\rvert^2\leq1-\eps^2$ then $p_\text{accept}\leq1-C_{d,s}\eps^2$, where $C_{d,s} = (1 - (1 - 1/4d^2)^{s-1})/2$.
\end{restatable*}

The proof proceeds similarly to the one of~\cref{thm:main qubits}.
Again, an uncertainty relation for Weyl operators plays an important role.
We record it since it may be of independent interest:

\begin{restatable*}[Uncertainty relation]{lem}{restateLemQuditUncertainty}\label{lem:qudit uncertainty}
  Let $\delta = 1/2d$ and $\psi$ a pure state such that $\lvert\tr[\psi W_{\vec x}]\rvert^2 > 1-\delta^2$ and $\lvert\tr[\psi W_{\vec y}]\rvert^2 > 1-\delta^2$.
  Then $W_{\vec x}$ and $W_{\vec y}$ must commute.
\end{restatable*}

We also study the \emph{minimal} number of copies required to distinguish stabilizer states from non-stabilizer states in such a way that the power of the statistical test does not decrease with the number of qubits.
Since the stabilizer states share the same second moments with uniformly random states (see \cref{subsec:higher intro} below for more detail), one can see that any such test requires at least three copies.
Our next result shows that this is sufficient at least when $d\equiv1,5\pmod6$.
For this, consider the POVM element
\begin{align*}
  \Pi_{\text{accept}} = \frac12 (I+V) \quad\text{where} \quad V\coloneqq d^{-n} \sum_{\vec x} A_{\vec x}^{\ot 3}.
\end{align*}

\begin{restatable*}[Stabilizer testing from three copies]{thm}{restateThmMainThreeCopies}\label{thm:main three copies}
Let $d\equiv1,5\pmod6$ and $\psi$ a pure state of $n$ qu$d$its.
Denote by $p_\text{accept}=\tr[\psi^{\ot3}\Pi_\text{accept}]$ the probability that the POVM element~$\Pi_\text{accept}$ accepts given three copies of $\psi$.
If $\psi$ is a stabilizer state then it accepts with certainty, $p_\text{accept}=1$.
On the other hand, if $\max_S\lvert\braket{S|\psi}\rvert^2\leq1-\eps^2$ then $p_\text{accept}\leq1-\eps^2/16d^2$.
\end{restatable*}

The operators $A_{\vec x}$ are known as phase-space point operators~\cite{gross2006hudson}, which are defined by a (symplectic) Fourier transform of the Weyl operator basis $W_{\vec a}$ (with respect to the index $\vec a$).
Again, the test corresponds to a particular element of the commutant, and to establish \cref{thm:main three copies} we also need another uncertainty relation, this time for phase-space point operators. %For this one we need a an

\begin{restatable*}{lem}{restateLemPointOpUncertainty}\label{lem:point op uncertainty}
Let $d$ be an odd integer and $\psi$ a pure state of $n$ qu$d$its.
Suppose that $\tr[\psi A_{\vec x}],\tr[\psi A_{\vec y}]$, $\tr[\psi A_{\vec z}] > \sqrt{1-1/2d^2}$.
Then $[\vec z-\vec x,\vec y-\vec x]=0$, i.e., $W_{\vec z-\vec x}$ and $W_{\vec y-\vec x}$ must commute.
\end{restatable*}

Lastly, we derive an explicit prescription for the minimal test that is perfectly complete, i.e., detects all stabilizer states with certainty.
Here we use the full power of the algebraic theory.
We assume that $d$ is a prime.

\begin{restatable*}[$O_t$]{dfn}{restateDfnOt}\label{dfn:O_t}
Consider the quadratic form $q\colon \ZZ_d^t \to \ZZ_D$ defined by $q(\vec x) \coloneqq \vec x\cdot\vec x$.%
\footnote{Recall that for $x\in\ZZ_d$, $x^2$ is well-defined modulo~$D$.}
We define~$O_t(d)$ as the group of $t\times t$-matrices $O$ with entries in $\ZZ_d$ that satisfy the following properties:
\begin{enumerate}
  \item $O$ is a \emph{$q$-isometry}: i.e., $O\vec x\cdot O\vec x=\vec x\cdot\vec x\pmod D$ for all $\vec x\in\ZZ_d^t$.
  \item $O$ is \emph{stochastic}: $O \vec 1_t = \vec 1_t \pmod d$.
\end{enumerate}
We will refer to $O_t(d)$ as the \emph{stochastic orthogonal group}; its elements will be called \emph{stochastic isometries}.
\end{restatable*}

\noindent
Equivalently, $O_t(d)$ is the group of $t\times t$-matrices~$O$ that are orthogonal in the ordinary sense (i.e., $O^T O = I \bmod d$) and such that the sum of elements in each row is equal to 1 (mod $D$). See \cref{rem:O_t} for more details.

Note that the subspace $T_O\coloneqq \{(O\vec y,\vec y)\,:\, \vec y \in \ZZ_d^t\}$ is a stochastic Lagrangian subspace in~$\Sigma_{t,t}(d)$ (as defined above in \cref{dfn:lagrangian stochastic}), and so we obtain a corresponding operator in the commutant, which we abbreviate by $R(O)=R(T_O)$.
It is easy to see that the operators $R(O)$ define a representation of the group $O_t(d)$, so
\[ \Pi^{\min}_t \coloneqq \frac1 {\lvert O_t(d) \rvert }\sum_{O\in O_t(d)} R(O) \]
is the projector onto the invariant subspace for this action.
Remarkably, not only do the $R(O)$ stabilize all stabilizer tensor powers~$\ket S^{\ot t}$ (\cref{eq:eigenvectors}), but $\Pi^{\min}_t$ is in fact the minimal perfectly complete test for stabilizer states:

\begin{restatable*}[Minimal stabilizer test with perfect completeness]{thm}{restateThmMinimalTest}\label{thm:minimal test}
Let $d$ be a prime and $n,t\geq1$.
Then the projector $\Pi^{\min}_t$ is the orthogonal projector onto $\Span~\{\ket S^{\ot t}\, : \, \ket S\bra S \in \Stab(n,d) \}$.
\end{restatable*}

Are there any other tensor power states in the support of~$\Pi^{\min}_t$?
For every~$d\geq2$, we have proved above there exists some~$t\geq3$ such that stabilizer testing is possible using~$t$ copies.
Since the accepting POVM element is in each case the projector onto the invariant subspace of an element in~$O_t(d)$ (e.g., the anti-identity for $d=2$ and $t=6$), it follows that in this case the only tensor power states contained in the support of~$\Pi^{\min}_t$ are tensor powers of stabilizer states!

%-----------------------------------------------------------------------------
\subsection{De Finetti theorems for stabilizer symmetries}\label{subsec:de finetti intro}
%-----------------------------------------------------------------------------
Quantum de Finetti theorems provide versatile tools for the study of correlations in quantum states with permutation symmetry.
They have found many important applications, from quantifying the monogamy of entanglement to proving security for quantum key distribution protocols, where de Finetti theorems allow to reduce general attacks to collective attacks%, where the adversary is restricted to applying the same operation on the pieces of the message
~\cite{renner2008security}.
By now, several variants and generalizations are known~\cite{stormer1969symmetric,hudson1976locally,raggio1989quantum,petz1990finetti,caves2002unknown,konig2005finetti,dcruz2007finite,christandl2007one,renner2007symmetry,navascues2009power,koenig2009most, brandao2011quasipolynomial,brandao2013quantum,brandao2017quantum,brandao2016mathematics}.
Generally speaking, de Finetti theorems state that when~$\rho$ is a quantum state on~$(\CC^\ell)^{\ot t}$ that commutes with all permutations
(i.e., $[r_\pi, \rho]=0$ for all $\pi \in S_t$, where $r_\pi$ are the permutation operators defined in \cref{eq:perm_factors})
then its reduced density operators~$\rho_{1\dots{}s}=\tr_{s+1\dots{}t}[\rho]$ are well-approximated by convex mixtures of i.i.d.\ states in some suitable sense if $s\ll t$.
E.g., for any such $\rho$ there exists a probability measure~$d\mu$ on the space of mixed states on $\CC^\ell$ such that~\cite{christandl2007one},
\begin{align}\tag{\ref{eq:christandl mixed de finetti}}
  \frac12 \left\lVert \rho_{1\dots{}s} - \int d\mu(\sigma) \sigma^{\ot s} \right\rVert_1 \leq 2\frac{\ell^2s}{t}.
\end{align}

Using the techniques developed for stabilizer testing, we prove two new versions of the quantum de Finetti theorem \emph{adapted to the symmetries inherent in stabilizer states}.
A key insight from the preceding section was that any stabilizer tensor power $\ket S^{\ot t}$ is stabilized not only by the permutations, but by the larger group~$O_t(d)$.
This group contains includes in general many more elements, for example, the anti-identity~\eqref{eq:anti identity} for the case of qubits.
Our de Finetti theorems for stabilizer states show that if we consider \emph{arbitrary} states $\rho$ on $((\CC^d)^{\ot n})^{\ot t}$ that show symmetries of this kind, then the conclusions of the de Finetti theorem can be strengthened.
In this case, the reduced states can be well-approximated by convex mixtures of tensor powers of stabilizer states in $\Stab(n,d)$ (rather than of general pure states in $(\CC^d)^{\ot n}$).

Our first de Finetti theorem shows that the enlarged symmetry provided by the stochastic orthogonal group ensures that the approximation is \emph{exponentially} good in the number of traced-out subsystems.
This is remarkable, since the ordinary permutation symmetry-based de~Finetti theorem achieves exponential convergence only if the form of allowed states is relaxed to include ``almost product states''~\cite{renner2007symmetry} or ``high weight vectors'' (as opposed to \emph{highest} weight vectors)~\cite{koenig2009most}.
Such a relaxation is, in fact, already necessary for classical distributions~\cite{diaconis1980finite}.
% We also remark that the cardinality of $O_t(d)$ is independent of the number~$n$ of qudits.
In detail:

\begin{restatable*}[Exponential stabilizer de Finetti theorem]{thm}{restateThmExpDeFinetti}\label{thm:exp de finetti}
Let $d$ be a prime and~$\rho$ a quantum state on~$((\CC^d)^{\ot n})^{\ot t}$ that commutes with the action of~$O_t(d) \supseteq S_t$.
Let $1\leq s\leq t$.
Then there exists a probability distribution~$p$ on the (finite) set of mixed stabilizer states%
\footnote{A mixed stabilizer state is a maximally mixed state on a stabilizer code (see \cref{subsec:stabs}).}
of $n$~qudits, such that
\begin{align*}
\frac12\left\lVert \rho_{1\dots{}s} - \sum_{\sigma_S} p(\sigma_S) \sigma_S^{\ot s} \right\rVert_1
% \leq 2 d^{-\frac12((t-s)-(2n+2)^2)}.
\leq 2 d^{\frac12(2n+2)^2} d^{-\frac12(t-s)}.
\end{align*}
\end{restatable*}

Our theorem can be understood as a stabilizer version of the Gaussian Finetti theorems established in~\cite{leverrier2009quantum,leverrier2016supq}; cf.~\cite{diaconis1992finite}.
The latter have been successfully used to establish security of continuous-variable quantum key distribution (QKD) protocols which admit the required symmetries~\cite{leverrier2013security,leverrier2017security}.
Since the input states of entanglement-based QKD schemes~\cite{ekert1991quantum}, are usually taken to be powers of stabilizer states, they show the enlarged symmetry identified here---a fact that seems to have been overlooked so far.
It is this natural to study applications of our de~Finetti theorems to QKD security proofs---we will report results on this elsewhere.

We can also ask to which extent the conclusions of \cref{thm:exp de finetti} hold if we only slightly enlarge the symmetry group.
The following theorem shows that if we consider quantum states that commute with permutations as well as the anti-identity (but not necessarily other elements of $O_t(d)$) then we still get an approximation by mixtures of stabilizer tensor powers -- but now with a polynomially rather than exponentially small error:

\begin{restatable*}[Stabilizer de Finetti theorem for the anti-identity]{thm}{restateThmAntiDeFinetti}\label{thm:anti de finetti}
Let $\rho$ be a quantum state on~$((\CC^2)^{\ot n})^{\ot t}$ that commutes with all permutations as well as with the action of the anti-identity~\eqref{eq:anti identity} on some (and hence every) subsystem consisting of six $n$-qubit blocks.
Let $s<t$ be a multiple of six.
Then there exists a probability distribution~$p$ on the (finite) set of mixed stabilizer states of $n$~qubits, such that
\begin{align*}
\frac12\left\lVert \rho_{1\dots{}s} - \sum_{\sigma_S} p(\sigma_S) \sigma_S^{\ot s} \right\rVert_1
\leq 6 \sqrt 2 \cdot 2^n \sqrt{\frac st}.
\end{align*}
\end{restatable*}

While \cref{thm:anti de finetti} is stated here only for qubits, we believe that a similar result can be established in any prime dimension.

%-----------------------------------------------------------------------------
\subsection{Robust Hudson theorem}
%-----------------------------------------------------------------------------

Similar techniques can also be applied to pure states with a small amount of negativity in their phase space representation.
More precisely, recall that for odd~$d$ the \emph{Wigner function} of a quantum state $\psi$ is defined by $w_\psi(\vec x) = d^{-n} \tr[A_{\vec x} \psi]$, where the operators $A_{\vec x}$ are the phase-space point operators mentioned above.
The Wigner function is a quasi-probability distribution, i.e., $\sum_{\vec x} w_\psi(\vec x)=1$, but it can be negative.
This negativity plays an important role -- e.g., it is an obstruction to efficient classical simulability~\cite{veitch2012negative,mari2012positive} and witnesses the onset of contextuality~\cite{howard2014contextuality}.

In fact, pure stabilizer states are characterized by having a nonnegative Wigner function---this is the discrete Hudson theorem~\cite{gross2006hudson}.
Our next result shows that this characterization is robust, and that the robustness is independent of the system size (number of qudits).
The relevant quantity is the \emph{Wigner} or \emph{sum-negativity} $\sn(\psi) = \sum_{w_\psi(\vec x)<0} \lvert w_\psi(\vec x)\rvert$, i.e., the absolute sum of negative entries of the Wigner function.

\begin{restatable*}[Robust finite-dimensional Hudson theorem]{thm}{restateThmRobustHudson}\label{thm:robust hudson}
  Let $d$ be odd and $\psi$ a pure quantum state of $n$ qu$d$its.
  Then there exist a stabilizer state $\ket S$ such that
  $\lvert\braket{S|\psi}\rvert^2 \geq 1 - 9 d^2 \sn(\psi)$.
\end{restatable*}

Our theorem gives a new quantitative meaning to the sum-negativity, and thereby to the related \emph{mana}, a monotone from the resource theory of stabilizer states~\cite{veitch2014resource} that has attracted increasing attention in the theory of fault-tolerant quantum computation.

%-----------------------------------------------------------------------------
\subsection{Random stabilizers, higher moments and designs}\label{subsec:higher intro}
%-----------------------------------------------------------------------------

As mentioned to above, randomized constructions based on the Haar measure are often near-optimal, yet have the drawback that generic quantum states cannot be efficiently prepared.
In contrast, random stabilizer states can be efficiently implemented, and early on, it had been discovered that they reproduce the same second moments as the Haar measure!
More recently, there had been significant progress on the third and fourth moments~\cite{zhu2016clifford,helsen2016representations,nezami2016multipartite}, opening up several many applications where random Clifford unitaries and stabilizer states have successfully replaced the Haar measure~\cite{magesan2011scalable,helsen2017multiqubit,kueng2016low,hayden2016holographic,nezami2016multipartite}.
To go beyond, however, a general understanding of the statistical properties of random stabilizer states is required.

The theory presented in this paper implies general formulas for the $t$-th moments of stabilizer states.
For qudits,
\begin{equation}\tag{\ref{eq:avg stab}}
\EE_{S \text{ stabilizer}} \left[\ket S \bra S^{\ot t} \right] = \frac1{Z_{n,d,t}} \sum_{T\in \Sigma_{t,t}(d)} R(T),
\end{equation}
where $T$ ranges precisely over the maximal isotropic stochastic subspaces from \cref{thm:commutant}!

Recall that a \emph{(complex projective) $t$-design} is an ensemble of states $\{p_i,\ket{\psi_i}\}$ such that the average of any polynomial of degree $t$ identically matches the average of the same polynomial with respect to the Haar measure. In other words, a $t$-design satisfies
\begin{align}\label{eq:t-th moment intro haar}
\EE_{i \sim p} \left[ (\ket{\psi_i}\bra{\psi_i})^{\otimes t} \right]
= \EE_{\psi \text{ Haar}} \left[ (\ket\psi\bra\psi)^{\otimes t} \right]
= \frac 1 {N'_{d,t}} \sum_{\pi\in S_t} r_\pi,
\end{align}
where the right-hand side is the familiar formula for the maximally mixed state on the symmetric subspace in terms of the symmetrizer---an easy consequence of Schur-Weyl duality.
When the stabilizer states form a $t$-design ($t\leq3$ for qubits, $t\leq2$ otherwise), \cref{eq:avg stab} reduces to \cref{eq:t-th moment intro haar}.
\Cref{eq:avg stab} unifies and generalizes all previously known results~\cite{zhu2015multiqubit,kueng2015qubit,webb2016clifford,zhu2016clifford,helsen2016representations,nezami2016multipartite}.

Importantly, however, our formula allows us to compute an arbitrary $t$-th moment even when the stabilizer states deviate significantly from being a $t$-design.
In fact, we demonstrate the power of the formula by using it to establish that following remarkable fact:
Even when the stabilizer states (a single Clifford orbit) fail to be a $t$-design, we can obtain $t$-designs by taking a finitely many Clifford orbits with appropriately chosen weights:

\begin{restatable*}{thm}{restateThmOrbitDesign}\label{thm:orbit design}
Let $d$ be a prime and $n\geq t-1$.
Then there exists an ensemble $\{p_i,\Psi_i\}_{i=1}^{M_{t,d}}$ of fiducial states in $(\CC^d)^{\ot n}$ such that:
\begin{align*}
\EE_{i \sim p} \EE_{U \text{ Clifford}}\!\!\left[\left(U\ket{\Psi_i}\bra{\Psi_i}U^\dagger\right)^{\ot t}\right]
% = \sum_{i=1}^{M_{t,d}} \frac{p_i}{\lvert \Cliff(n,d) \rvert} \sum_{ U \in \Cliff(n,d)} \!\!\!\!\left( U\ket{\Psi_i}\bra{\Psi_i}U^\dagger\right)^{\ot t}
= \EE_{\Psi \text{ Haar}}\!\!\left[\ket\Psi\bra\Psi^{\ot t}\right]
\end{align*}
\end{restatable*}

That is, the corresponding ensemble of Clifford orbits is a complex projective $t$-design.
Importantly, the number of fiducial states does \emph{not} depend on the number of qudits~$n$.

%=============================================================================
\section{Preliminaries}\label{sec:preliminaries}
%=============================================================================

%-----------------------------------------------------------------------------
\subsection{Pauli and Clifford group}
%-----------------------------------------------------------------------------
Let $d\geq2$ be an arbitrary integer.
We first consider a single qu$d$it with computational basis vectors $\ket q$, where $q\in\{0,\dots,d-1\}$ or $q\in\ZZ_d=\ZZ/d\ZZ$.
We define unitary shift and boost operators
\[ X \ket q = \ket{q+1}, \quad Z \ket q = \omega^q \ket q, \]
where $\omega=e^{2\pi i q/d}$.

The algebra of shift and boost operators differs slightly depending on whether $d$ is even or odd.
For uniform treatment, one introduces $\tau=(-1)^d e^{i\pi/d}=e^{i\pi(d^2+1)/d}$.
Note that $\tau^2=\omega$.
Let $D$ denote the order of $\tau$. Then $D=2d$ if $d$ is even, but $D=d$ if $d$ is odd (indeed, in this case $\tau=\omega^{2^{-1}}$, where $2^{-1}$ denotes the multiplicative inverse of $2$ mod $d$).
Then $Y \coloneqq \tau X^\dagger Z^\dagger$ is such that $XYZ=\tau I$, generalizing the commutation relation of the usual Pauli operators for qubits (where $\tau=i$).
For a single qudit, the Pauli group is generated by $X,Y,Z$ or, equivalently, by $\tau I,X,Z$.

For $n$ qudits, the Hilbert space is $\mathcal H_n = (\CC^d)^{\ot n}$, with computational basis vectors $\ket{\vec q}=\ket{q_1,\dots,q_n}$, and the \emph{Pauli group} $\mathcal P_n$ is the group generated by the tensor product of $I,X,Y,Z$ acting on each of the $n$~qu$d$its.

The  Clifford group $\Cliff(n,d)$ is defined as normalizer of the Pauli group in the unitary group, modulo phases.
That is, it consists of all unitary operators $U$ that $U \mathcal P_n U^\dagger \subseteq \mathcal P_n$, up to phases.
For qubits, the Clifford group is generated by the phase gate~$P=\begin{psmallmatrix}1&0\\0&i\end{psmallmatrix}$, the Hadamard gate~$H=\frac1{\sqrt2}\begin{psmallmatrix}1 & 1\\1 & -1\end{psmallmatrix}$, and the controlled-NOT gate.

%-----------------------------------------------------------------------------
\subsection{Weyl operators and characteristic function}\label{subsec:weyl ops}
%-----------------------------------------------------------------------------
At this point it is useful to recall the phase space picture of finite-dimensional quantum mechanics developed in~\cite{wootters1987wigner,appleby2005sic,gross2006hudson,gross2008quantum,beaudrap2013linearized}, which is analogous to the phase space formalism for continuous-variable systems used, e.g., in quantum optics~\cite{schleich2011quantum}.
For $\vec x=(\vec p,\vec q)\in\ZZ^{2n}$, define the \emph{Weyl operator}
\begin{equation}\label{eq:weyl operator}
  W_{\vec x} = W_{\vec p, \vec q} = \tau^{-\vec p\cdot \vec q} (Z^{p_1} X^{q_1})\otimes\cdots\otimes (Z^{p_n} X^{q_n}).
\end{equation}
Clearly, each Weyl operator is an element of the Pauli group.
Conversely, each element of the Pauli group is equal to a Weyl operator up to a phase that is a power of~$\tau$.
It is not hard to see that the Weyl operators themselves only depend on $\vec x$ modulo $D$ (which we recall is $2d$ or $d$, depending on whether $d$ is even or odd).
Indeed,
\begin{equation}\label{eq:weyl operator mod d}
  W_{\vec x+d\vec z} = (-1)^{(d+1)[\vec x,\vec z]} W_{\vec x},
\end{equation}
where we have introduced the $\ZZ$-valued \emph{symplectic form} on $\ZZ^{2n}$
\begin{align}\label{eq:symplectic form}
[\vec x,\vec y]
= [(\vec p,\vec q),(\vec p',\vec q')]
= \vec p \cdot \vec q' - \vec q \cdot \vec p'.
\end{align}
We will often use the symplectic form in situations where $\vec x$, $\vec y$ are elements of $\ZZ_d^{2n}$ or $\ZZ_D^{2n}$, and interpret $[\vec x,\vec y]$ accordingly.
For example,
\begin{equation}\label{eq:projective group law}
W_{\vec x} W_{\vec y} = \tau^{[\vec x,\vec y]} W_{\vec x+\vec y}
\end{equation}
for all $\vec x$, $\vec y\in\ZZ^{2n}_D$.
This implies that in particular
\begin{equation}\label{eq:commutation law}
W_{\vec x} W_{\vec y} = \omega^{[\vec x,\vec y]} W_{\vec y} W_{\vec x}.
\end{equation}
Thus the commutation relations between Weyl operators only depend on $\vec x, \vec y$ mod $d$.
In this sense, $\mathcal V_n=\{0,\dots,d-1\}^{2n}$ is the natural classical phase space associated with the Hilbert space~$\mathcal H_n=(\CC^d)^{\ot n}$.
We will often write $W_{\vec x}$ for $\vec x\in\ZZ_d^{2n}$, identifying $\ZZ_d^{2n} \cong \mathcal V_n$ in the standard way.

Note that $\tr[W_{\vec x}]\neq0$ if and only if $W_{\vec x}$ is a scalar multiple of the identity (necessarily $\pm I$), that is, if and only if $\vec x\equiv \vec 0 \pmod d$.
Together with \cref{eq:projective group law}, it follows that the re-scaled Weyl operators $\{d^{-n/2} W_{\vec x}\}$ for $\vec x\in V_n$ form an orthonormal basis with respect to the Hilbert-Schmidt inner product $\braket{A,B} = \tr[A^\dagger B]$.
In particular, any operator $B$ on $\mathcal H_n$ can be expanded in the form
$B = d^{-n/2} \sum_{\vec x\in \mathcal V_n} c_B(\vec x) W_{\vec x}$.
The expansion coefficients~$c_B(\vec x)$ together define the \emph{characteristic function}~$c_B: \mathcal V_n\to\CC$ of the operator~$B$,
\begin{equation}\label{eq:characteristic function}
c_B(\vec x) = d^{-n/2} \tr[W_{\vec x}^\dagger B],
\end{equation}
and we have Parseval's identity
\begin{equation}\label{eq:parseval}
\tr[A^\dagger B] = \sum_{\vec x\in\mathcal V_n} \overline c_A(x) c_B(x).
\end{equation}

By definition, if $U$ is a Clifford unitary then, for every~$\vec x\in\mathcal V_n$, $U W_{\vec x} U^\dagger$ is proportional to a Weyl operator $W_{\vec x'}$, where we can take $\vec x'\in\mathcal V_n$ in view of \cref{eq:weyl operator mod d}.
Since conjugation preserves the commutation relations, this action has substantially more structure.
In particular, the mapping $\vec x \mapsto \vec x'$ is implemented by an element of the symplectic group~$\Sp(2n,d)$, i.e., a linear transformation of $\ZZ_d^{2n}$ that preserves the symplectic form~\eqref{eq:symplectic form}.
The following facts are well-known in the literature (e.g.~\cite{appleby2005sic,gross2006hudson,beaudrap2013linearized,zhu2015multiqubit}).
%Since the details of the definitions and presentations vary considerably between authors, we provide a short self-contained proof.

\begin{lem}\label{lem:cliff basics}
For any prime $d$ and any $n\in\NN$, the following holds:
\begin{enumerate}
\item
For each $U\in\Cliff(n,d)$, there is a $\Gamma\in\Sp(2n,d)$ and a function $f\colon \ZZ_d^{2n} \to \ZZ_d$ such that
\begin{align}\label{eq:cliff action prp}
  U W_{\vec x} U^\dagger = \omega^{f(\vec x)} W_{\Gamma \vec x}
  \qquad \forall\,\vec x \in\ZZ_d^{2n}.
\end{align}

\item
Conversely, for each $\Gamma\in\Sp(2n,d)$, there is a $U\in\Cliff(n,d)$ and a phase function $f\colon \ZZ_d^{2n} \to \ZZ_d$ such that \cref{eq:cliff action prp} holds.
If $d$ is odd, one can choose $U$ such that $f\equiv0$.
% If $d$ is even, for any specific $\vec x$, one can choose $U$ sucht that $f(\vec x)=0$.

\item
The quotient of the Clifford group by Weyl operators and phases is isomorphic to $\Sp(2n,d)$.
\end{enumerate}
\end{lem}

\noindent
Below, we will frequently assume that a correspondence~$\Gamma \mapsto U_\Gamma$ has been fixed.

%-----------------------------------------------------------------------------
\subsection{Wigner function and phase space point operators}\label{subsec:phase space}
%-----------------------------------------------------------------------------
It is also useful to consider the symplectic Fourier transform, which for any function $f\colon\mathcal V_n\to\CC$ is defined as
\begin{equation}\label{eq:symplectic fourier transform}
  \hat f(\vec x) = d^{-n} \sum_{\vec y} \omega^{-[\vec x,\vec y]} f(\vec y).
\end{equation}
The transformation $f \mapsto \hat f$ is unitary, i.e., we have Parseval's identity:
$\sum_{\vec x} \overline{\hat f(\vec x)} \hat g(\vec x) = \sum_{\vec y} \overline{f(\vec y)} g(\vec y)$.
% \begin{align*}
% \quad &\quad \sum_{\vec x} \overline{\hat f(\vec x)} \hat g(\vec x) \\
% &= \quad d^{-2n} \sum_{\vec y,\vec y'} \sum_{\vec x} \omega^{[\vec x,\vec y-\vec y']} \overline{f(\vec y)} g(\vec y') \\
% &= \quad d^{-2n} \sum_{\vec y,\vec y'} \delta_{\vec y,\vec y'} d^{2n} \overline{f(\vec y)} g(\vec y') \\
% &= \sum_{\vec y} \overline{f(\vec y)} g(\vec y).
% \end{align*}

The Fourier transform of the characteristic function is (up to normalization) known as the \emph{Wigner function}~\cite{wootters1987wigner} $w_B\colon \mathcal V_n\to\CC$, defined by
\begin{equation}\label{eq:wigner function}
w_B(\vec x)
= d^{-n/2} \hat c_B(\vec x)
= d^{-3n/2} \sum_{\vec y} \omega^{-[\vec x,\vec y]} c_B(\vec y)
= d^{-2n} \sum_{\vec y} \omega^{-[\vec x,\vec y]}\tr[W_{\vec y}^\dagger B]
= d^{-n} \tr[A_{\vec x} B],
\end{equation}
where we have introduced the \emph{phase-space point operators}
\begin{align}\label{eq:phase space point op definition}
A_{\vec x} = d^{-n} \sum_{\vec y} \omega^{-[\vec x,\vec y]} W_{\vec y}^\dagger.
\end{align}
The operators $\{A_{\vec x}\}$ form an orthogonal basis of the space of all operators,
$\tr[A_{\vec x}^\dagger A_{\vec y}] = d^n \delta_{\vec x,\vec y}$,
% \begin{align*}
% &\quad \tr[A_{\vec x}^\dagger A_{\vec y}] \\
% &= d^{-2n} \sum_{\vec z,\vec w} \omega^{[\vec x,\vec z]-[\vec y,\vec w]} \tr[W_\vec{z} W_\vec{w}^\dagger] \\
% &= d^{-2n} \sum_{\vec z,\vec w} \omega^{[\vec x,\vec z]-[\vec y,\vec w]} d^n \delta_{\vec z,\vec w} \\
% &= d^{-n} \sum_{\vec z} \omega^{[\vec x-\vec y,\vec z]} \\
% &= d^{-n} d^{2n} \delta_{\vec x,\vec y}
% &= d^n \delta_{\vec x,\vec y}
% \end{align*}
so the Wigner function can be seen as the set of coefficients of an operator as expanded in this basis, $B = \sum_{\vec x} w_B (\vec x) A_{\vec x}^\dagger$.
% We note that $\tr[A_{\vec x}]=1$ as well as $\sum_{\vec x} w_B(\vec x)=\tr[B]$.
Moreover, the Wigner function of a quantum state is a quasiprobability distribution in the sense that $\sum_{\vec x} w_\rho(\vec x)=1$.

For odd $d$ the Wigner function is particularly well-behaved.
For one, the phase-space point operators are Hermitian (this is also true for qubits) and they square to the identity (so the eigenvalues are~$\pm1$ and in particular $\lVert A_{\vec x}\rVert=1$).
This means that the Wigner function of a quantum state is real and $-d^{-n} \leq w_\psi(\vec x) \leq d^{-n}$.
The phase-space point operators satisfy the following important identity:
\begin{align}\label{eq:three points}
  A_{\vec x} A_{\vec y} A_{\vec z} = \omega^{2 [\vec z-\vec x,\vec y-\vec x]} A_{\vec x-\vec y +\vec z}
\end{align}
Moreover, (only) for odd $d$ does the Wigner function transforms covariantly with respect to the Clifford group.
Here, the Clifford operators can (up to overall phase) be parametrized by an affine symplectic transformation, i.e., by a symplectic matrix $\Gamma\in\Sp(2n,d)$ and a vector $\vec b\in\ZZ_d^{2n}$.
Then $U=W_{\vec v}\,\mu_\Gamma$ is in $\Cliff(n,d)$, where $\mu_\Gamma$ is the so-called \emph{metaplectic} representation of the symplectic group (see, e.g.,~\cite{gross2006hudson}), and the conjugation action of $U$ on phase-space point operators is given by
\begin{equation}\label{eq:clifford phase space}
U A_{\vec x} U^\dagger = A_{\Gamma \vec x + \vec b}.
\end{equation}
In particular, the Weyl operators induce a translation in phase space.

%-----------------------------------------------------------------------------
\subsection{Stabilizer groups, codes, and states}\label{subsec:stabs}
%-----------------------------------------------------------------------------
We now give uniform account of the stabilizer formalism~\cite{gottesman1997stabilizer,gottesman1999heisenberg} for qu$d$its.
Stabilizer states are commonly defined in terms of the Pauli group in the following way:
Consider a subgroup of the Pauli group $S\subseteq \mathcal P_n$ that does not contain any (nontrivial) multiple of the identity operator.
Then the operator
\begin{equation}\label{eq:projection onto stabilizer code}
P_S = \frac 1 {\lvert S\rvert} \sum_{P \in S} P
\end{equation}
is an orthogonal projection onto a subspace $V_S \subseteq \mathcal H_n$ of dimension $d^n / \lvert S\rvert$.
We say that $V_S$ is the \emph{stabilizer code} associated with the \emph{stabilizer group} $S$.
If $\lvert S\rvert = d^n$ then this code is spanned by a single state, called a \emph{(pure) stabilizer state} and denoted by $\ket S\bra S$.
It is given precisely by \cref{eq:projection onto stabilizer code}.
In other words, a stabilizer state $\ket S$ is the unique $+1$ eigenvector (up to scalars) of all the Pauli operators in $S$,
\[ P \ket S = \ket S \quad (\forall P\in S). \]
In the following we will mostly be talking about stabilizer groups that determine a pure state.
We denote the (finite) set of pure stabilizer states in $(\CC^d)^{\ot n}$ by $\Stab(n,d)$.

In order to connect the stabilizer formalism to the phase space picture, we observe that the stabilizer group can be written in the form
\begin{equation}\label{eq:normal form stabilizer group}
S = \{ \omega^{f(\vec x)} W_{\vec x} : \vec x \in M \},
\end{equation}
for some subset $M\subseteq \mathcal V_n$ and some function $f\colon M\to\ZZ_d$.
The two pieces of data determine the stabilizer state uniquely. Indeed, $\ket S$ can be characterized by demanding that
\[ W_{\vec x} \ket S = \omega^{-f(\vec x)} \ket S \quad (\forall P\in S). \]
Moreover, it is not hard to verify that $M$ is closed under addition (because $S$ is a group) and that $[\vec x,\vec y]=0$ for any two elements $\vec x,\vec y\in M$.
Thus, $M$ is a totally isotropic submodule of the phase space $\mathcal V_n$ (which itself can be thought of as a $\ZZ_d$-module).
For simplicity, we will usually say \emph{subspace} instead of submodule, although the latter terminology is more appropriate for non-prime $d$.
Moreover, $\lvert M\rvert=d^n$, which is the maximal possible cardinality of any such subspace -- one often says that $M$ is a Lagrangian subspace and it holds that $M=M^\perp$, where $M^\perp=\{\vec y \in \mathcal V_n | [\vec x,\vec y]=0\;\forall \vec x\in M\}$.
See, e.g.,~\cite{gross2006hudson,gross2013stabilizer} for further detail on this symplectic point of view.

Conversely, suppose that $M$ is a Lagrangian subspace of $\mathcal V_n$.
Then there always exist functions $f$ such that $\{ \omega^{f(\vec x)} W_{\vec x} \}_{\vec x\in M}$ is a stabilizer group; we will denote the corresponding stabilizer states by $\ket{M,f}$.
Any other such function $f$ can be obtained by replacing~$f$ by $g=f+\delta$, where $\delta \colon M\to\ZZ_d$ is a $\ZZ_d$-linear function.
We can always write $\delta(\vec x)=[\vec z,\vec x]$; then $\ket{M,g} = W_{\vec z} \ket{M,f}$.
In this way, $M$ parametrizes an orthonormal basis of $\mathcal H_n$ worth of stabilizer states.
In particular, any state that is a simultaneous eigenvector of the $\{W_{\vec x}\}_{\vec x\in M}$ is necessarily a stabilizer state.
It is not hard to verify that the quantum channel that implements the projective measurement in this \emph{stabilizer basis} $\{\ket{M,f}\}_f$ is given by
\begin{equation}\label{eq:measurement in stabilizer basis}
  \Lambda_M[\rho]
= \sum_f \ket{M,f}\braket{M,f|\rho|M,f}\bra{M,f}
= d^{-n} \sum_{\vec x \in M} W_{\vec x} \rho W_{\vec x}^\dagger.
\end{equation}
% Indeed,
% \begin{align*}
% &\quad \Lambda_M[\ket{M,f}\bra{M,g}]
% = d^{-n} \sum_{\vec x \in M} W_{\vec x} \ket{M,f}\bra{M,g} W_{\vec x}^\dagger \\
% &= d^{-n} \sum_{\vec x \in M} \omega^{g(x)-f(x)} \ket{M,f}\bra{M,g}
% = \delta_{f,g} \ket{M,f}\bra{M,f},
% \end{align*}
% since $g$ and $f$ differ by a $\ZZ_d$-linear function.

The fact that any stabilizer state can be parametrized as $\ket S=\ket{M,f}$ will be of fundamental importance to our investigations.
As a first consequence, we note that \cref{eq:projection onto stabilizer code,eq:normal form stabilizer group} imply that
$\ket S\bra S
= d^{-n} \sum_{\vec x\in M} \omega^{f(\vec x)} W_{\vec x}$.
This shows that the characteristic function is given by
\begin{equation}\label{eq:characteristic function stabilizer}
  c_S(\vec x) = \begin{cases}
  d^{-n/2} \omega^{f(\vec x)} & \text{ if } \vec x \in M, \\
  0 & \text{ otherwise,}
  \end{cases}
\end{equation}
i.e., it is supported precisely on the set $M$.

For odd $d$ the phase is a linear function, so it can be written as $f(\vec x)=[\vec a,\vec x]$ for some suitable vector~$\vec a$ (e.g.,~\cite[App.~C]{gross2006hudson}).
This means that the Wigner functions of stabilizer states have the following form~\cite{gross2013stabilizer}:
\begin{equation}\label{eq:wigner function odd stabilizer}
  w_S(\vec x)
% = d^{-3n/2} \sum_{\vec y} \omega^{-[\vec x,\vec y]} c_B(\vec y) \\
= d^{-3n/2} \sum_{\vec y \in M} \omega^{-[\vec x,\vec y]} d^{-n/2} \omega^{[\vec a,\vec y]}
% = d^{-2n} \sum_{\vec y \in M} \omega^{[\vec a-\vec x,\vec y]} \\
% = d^{-n} \delta_{M^\perp}(\vec a-\vec x) \\
% = d^{-n} \delta_{M^\perp}(\vec x-\vec a) \\
% = d^{-n} \delta_M(\vec x-\vec a) \\
% = d^{-n} \delta_{\vec a + M}(\vec x) \\
= \begin{cases}
  d^{-n} & \text{ if } \vec x \in \vec a + M, \\
  0 & \text{ otherwise}
  \end{cases}
\end{equation}
(using that $M=M^\perp$ for a pure stabilizer state).
In particular, the Wigner function is non-negative.
% The converse is also true; if $w_\psi$ is the uniform distribution on a Lagrangian subspace then $\psi$ is a pure stabilizer state.
The finite-dimensional Hudson theorem asserts that, for pure states, the converse is also true~\cite{gross2006hudson}.
In \cref{sec:hudson} we will prove a robust version of this result.

%=============================================================================
\section{Testing stabilizer states}\label{sec:stabilizer testing}
%=============================================================================
Given two copies of an unknown pure state $\psi=\ket\psi\bra\psi$ on $\mathcal H_n$, it is easy to verify using phase estimation whether $\ket\psi$ is an eigenvector of a given Weyl operator $W_{\vec x}$.
In particular, if $W_{\vec x}$ is Hermitian then we simply measure twice and compare the result.
The probability of obtaining the same outcome is
\begin{align}\label{eq:same outcome}
 \tr\left[\psi^{\ot 2} \frac {I + W_{\vec x} \ot W_{\vec x}^\dagger} 2\right]
% = \frac12 \left( 1 + \tr\left[ \psi^{\ot 2} W_{\vec x} \ot W_{\vec x}^\dagger\right] \right) \\
% = \frac12 \left( 1 + \tr\left[ \psi W_{\vec x} \right] \tr\left[\psi W_{\vec x}^\dagger\right] \right)
= \frac12 \left( 1 + d^n \lvert c_\psi(\vec x)\rvert^2 \right),
\end{align}
where we recall that $c_\psi$ denotes the characteristic function defined in \cref{eq:characteristic function}.

To turn this idea into an algorithm for testing whether $\psi$ is a stabilizer state we need a way of generating good candidate Weyl operators.
For this we note that, since $\psi$ is a \emph{pure} quantum state,
\begin{equation}\label{eq:p distribution}
p_\psi(\vec x)
= \lvert c_\psi(\vec x)\rvert^2
= d^{-n} \lvert\braket{\psi|W_{\vec x}|\psi}\rvert^2
= d^{-n} \tr[\psi W_{\vec x} \psi W_{\vec x}^\dagger]
\end{equation}
is a probability distribution on the phase space $\mathcal V_n$.
This follows directly from \cref{eq:parseval}.
We call $p_\psi$ the \emph{characteristic distribution} of $\psi$.
% Clearly, $0\leq p_\psi(x) \leq d^{-n}$.

Now, if $\ket\psi=\ket S=\ket{M,f}$ is a stabilizer state then \cref{eq:characteristic function stabilizer} implies that $p_\psi$ is simply the uniform distribution on the subset $M\subseteq \mathcal V_n$:
\begin{equation}\label{eq:characteristic distribution stabilizer}
  p_S(\vec x) = \begin{cases}
  d^{-n} & \text{ if } \vec x \in M, \\
  0 & \text{ otherwise.}
  \end{cases}
\end{equation}
Note that $p_S$ is maximally sparse in the case of pure stabilizer states, since it always holds true that $0\leq p_\psi(x)\leq d^{-n}$.
Therefore, if we sample from the characteristic distribution of a stabilizer state then \cref{eq:characteristic distribution stabilizer} shows that we would with certainty obtain the label of a Weyl operator for which $\ket\psi$ is an eigenvector.

Importantly, the converse of this statement is also true:
Suppose that $\ket\psi$ is an eigenvector of all Weyl operators $W_{\vec x}$ for $\vec x$ in the support of the characteristic distribution (i.e., $p_\psi(\vec x)>0$).
Since $p_\psi(x)\leq d^{-n}$, the support of $p_\psi$ contains at least $d^n$ points.
Thus if $\ket\psi$ is an eigenvector of all these Weyl operators then the support must be exactly of cardinality $d^n$ and so $\ket\psi$ is a stabilizer state.
This suggests the following algorithm:
\begin{enumerate}
\item Sample from the characteristic distribution of $\psi$. Denote the result $\vec x$.
\item Measure the corresponding Weyl operator $W_{\vec x}$ twice and accept if the result is the same.
\end{enumerate}
By the preceding discussion, this test will accept if and only if the state is a stabilizer state.
But how do we go about sampling from the characteristic distribution?

%-----------------------------------------------------------------------------
\subsection{Qubit stabilizer testing and Bell difference sampling}\label{subsec:bell diff}
%-----------------------------------------------------------------------------
When the wave function $\ket\psi$ is real in the computational basis then sampling from the characteristic distribution can be achieved by \emph{Bell sampling}, introduced for qubits in~\cite{montanaro2017learning} (cf.~\cite{zhao2016fast}).
% Does this ever happen for odd $d$?
Bell sampling amounts to performing a basis measurement in the basis obtained by applying the Weyl operators to a fixed maximally entangled state, $\ket{W_{\vec x}} = (W_{\vec x} \ot I) \ket{\Phi^+}$.
Since the Weyl operators are orthogonal, $\ket{W_{\vec x}}$ is an orthonormal basis of the doubled Hilbert space $\mathcal H_n \ot \mathcal H_n$.
Using the transpose trick,
\begin{align}\label{eq:bell sampling}
  \left\lvert\braket{W_{\vec x} | \psi^{\ot 2}}\right\rvert^2
% = \left\lvert\braket{W_{\vec x} | \left(\ket\psi \ot \ket\psi\right)}\right\rvert^2
= d^{-n} \lvert\braket{\psi|W_{\vec x}|\bar\psi}\rvert^2
\end{align}
In case the wave function is real, \cref{eq:bell sampling} is exactly equal to $p_\psi(\vec x)$; Bell sampling therefore allows us to implement step (1) above given two copies of the unknown state $\psi$.

In general, however, the transformation $\psi \mapsto \overline\psi=\psi^T$ \emph{cannot} be implemented by a physical process, since the transpose map is well-known not to be completely positive.
% (Not even if we restrict to stabilizer states.)
Thus we need a new idea to treat the general case where the wave function can be complex.

We start with the observation that if $\psi$ is a stabilizer state then so is $\overline\psi$.
Indeed, $\overline{W_{\vec p,\vec q}} = (-1)^{(d+1)(\vec p\cdot \vec q)} W_{J(p,q)}$, where $J$ is the involution~\cite{appleby2005sic}
\begin{equation*}
  J \colon \mathcal V_n \to \mathcal V_n, \quad (p,q) \mapsto ((-p) \bmod d,q)
\end{equation*}
on phase space (note that the phase is trivial when $d$ is odd and so always well-defined mod $d$).
On the other hand, $\overline{\omega^{f(\vec x)}} = \omega^{-f(\vec x)}$.
It follows that if $\ket\psi=\ket{M,f}$ then $\ket{\overline\psi}=\ket{J(M),g}$, where $\omega^{g(\vec x)}=\omega^{-f(\vec x)} (-1)^{(d+1)(\vec p \cdot \vec q)}$ (again, this is well-defined for any $d$).

For qubits ($d=2$), the involution $J$ is trivial.
This means that if $\psi$ is a stabilizer state then $\psi$ and $\bar\psi$ are characterized by the same subspace $M$, but possibly different phases.
We saw above that (only) in this case there exists a Weyl operator $W_{\vec z}$ such that $\ket{\bar\psi} = W_{\vec z} \ket{\psi}$.
As a consequence, if we perform Bell sampling on $\ket\psi \ot \ket\psi$ then, from \cref{eq:bell sampling},
\begin{align*}
  \left\lvert\braket{W_{\vec x} | \psi^{\ot 2}}\right\rvert^2
% = d^{-n} \lvert\braket{\psi|W_{\vec x}|\bar\psi}\rvert^2
% = d^{-n} \lvert\braket{\psi|W_{\vec x}W_{\vec z}|\psi}\rvert^2
= d^{-n} \lvert\braket{\psi|W_{\vec x + \vec z}|\psi}\rvert^2
= p_\psi(\vec x+\vec z).
\end{align*}
Of course, $\vec z$ is an unknown vector that depends on the stabilizer state $\psi$.
% (so there is no contradiction to our earlier assertion that $\psi\mapsto\bar\psi$ cannot be implemented by a CPTP map).
But since $\vec z$ depends \emph{only} on the stabilizer state $\psi$, it is clear that we may Bell sample \emph{twice} and take the difference of the result in order to obtain a uniform sample $\vec a$ from the subspace $M$.
Formally:

\restateDfnBellDifferenceSampling

It is not obvious that Bell difference sampling should be meaningful for non-stabilizer quantum states $\psi$.
The following theorem shows that it has a natural interpretation for general states:

\restateThmBellDifferenceSamplingQubits

The proof of \cref{thm:bell difference sampling qubits} uses the symplectic Fourier transform defined in \cref{eq:symplectic fourier transform}.
 Remarkably, the characteristic \emph{distribution} of any pure state is left invariant by the Fourier transform:
\begin{equation}\label{eq:pure state fourier invariance}
\begin{aligned}
  \widehat p_\psi(\vec a)
&= 2^{-n} \sum_{\vec x} (-1)^{[\vec a,\vec x]} c_\psi(\vec x) c_\psi(\vec x)
= 2^{-n} \sum_{\vec x} c_\psi(\vec x) c_{W_{\vec a} \psi W_{\vec a}}(\vec x) \\
&= 2^{-n} \tr[\psi W_{\vec a} \psi W_{\vec a}]
= p_\psi(\vec a),
\end{aligned}
\end{equation}
where the third step is \cref{eq:parseval} (note that for qubits the characteristic function is real).

We now give the proof of \cref{thm:bell difference sampling qubits}:

\begin{proof}[Proof of \cref{thm:bell difference sampling qubits}]
We start with the observation that
$\Pi_{\vec a} = (I \ot I \ot I \ot W_{\vec a}) \Pi_{\vec 0} (I \ot I \ot I \ot W_{\vec a})$.
On the other hand, it is easy to verify that
\begin{equation}\label{eq:projection Pi_0}
  \Pi_{\vec 0} = \frac1{2^{2n}} \sum_{\vec x} W_{\vec x}^{\ot 4}
\end{equation}
(i.e., it is the projection onto a stabilizer code of dimension $2^{2n}$, which played an important role in~\cite{zhu2016clifford}, and Bell difference sampling achieves precisely the syndrome measurement for this code).
It follows that
\begin{equation}\label{eq:Pi_a}
\Pi_{\vec a} = \frac1{2^{2n}} \sum_{\vec x} (-1)^{[\vec a,\vec x]} W_{\vec x}^{\ot 4}
\end{equation}
and so
\begin{align*}
\tr\left[\Pi_{\vec a} \psi^{\ot 4}\right]
&= \frac1{2^{2n}} \sum_{\vec x} (-1)^{[\vec a,\vec x]} \tr\left[ W_{\vec x}^{\ot 4} \psi^{\ot 4} \right]
= \sum_{\vec x} (-1)^{[\vec a,\vec x]} p_\psi(\vec x) p_\psi(\vec x) \\
&= \sum_{\vec x} \hat p_\psi(\vec x) \hat p_\psi(\vec x+\vec a)
= \sum_{\vec x} p_\psi(\vec x) p_\psi(\vec x+\vec a);
\end{align*}
the third equality is the unitarity of the Fourier transform, which also maps modulations to translations, and in the last step we used \cref{eq:pure state fourier invariance}, namely that the characteristic distribution of a pure state is left invariant by the Fourier transform.
\end{proof}

\Cref{thm:bell difference sampling qubits} motivates \cref{algo:qubits} as a natural algorithm for testing whether a multi-qubit state is a stabilizer state.
The following theorem shows that stabilizer states are the only states that are accepted with certainty, and it quantifies this observation in a dimension-independent way:

\restateThmMainQubits

\noindent
The converse bound of \cref{thm:main qubits} can be stated equivalently as
\begin{align}\label{eq:saner quantitative bound}
  \max_S \lvert\braket{S|\psi}\rvert^2 \geq 4p_{\text{accept}}-3.
\end{align}

\begin{proof}
According to \cref{thm:bell difference sampling qubits}, step~1 of the algorithm samples elements $\vec a$ with probability $q(\vec a)=\sum_{\vec x} p_\psi(\vec x) p_\psi(\vec x+\vec a)$.

Let us first discuss the case that $\psi$ is a stabilizer state, say $\ket\psi=\ket{M,f}$.
Since $p_\psi(\vec x)$ is the uniform distribution over $M$, which is a \emph{subspace}, it holds that $q(\vec a) = p_\psi(\vec a)$, since, for $x\in M$, $x+a\in M$ if and only if $a\in M$.
But this means that $\vec a\in M$ with certainty.
Thus, $\ket\psi$ is an eigenvector of the corresponding Weyl operator $W_{\vec a}$ and step~2 of the test always accepts.

We now consider the case that $\psi$ is a general pure state.
Our goal will be to show that if \cref{algo:qubits} succeeds with high probability then there must exist a stabilizer state with high overlap with $\psi$.
According to \cref{eq:same outcome}, the probability of acceptance is given by
\begin{align*}
  p_\text{accept}
% = \sum_{\vec a} q(\vec a) \frac12 \left( 1 + 2^n \lvert c_\psi(\vec a)\rvert^2 \right)
= \frac12 \sum_{\vec a} q(\vec a) \left( 1 + 2^n p_\psi(\vec a) \right)
\end{align*}
where we recall that $q(\vec a)=\sum_{\vec x} p_\psi(\vec x) p_\psi(\vec x+\vec a)$.
Thus, by the Cauchy-Schwarz inequality,
\begin{equation}\label{eq:qubit acceptance}
\begin{aligned}
p_\text{accept}
&= \frac12 \sum_{\vec x} p_\psi(\vec x) \left( 1 + 2^n \sum_{\vec a} p_\psi(\vec x+\vec a) p_\psi(\vec a) \right)
\leq \frac12 \sum_{\vec x} p_\psi(\vec x) \left( 1 + 2^n \sum_{\vec a} p_\psi(\vec a)^2 \right) \\
&= \frac12 \left( 1 + 2^n \sum_{\vec a} p_\psi(\vec a)^2 \right)
= \frac12 \sum_{\vec a} p_\psi(\vec a) \left( 1 + 2^n p_\psi(\vec a) \right),
\end{aligned}
\end{equation}
where we have also used the fact that $p_\psi$ is a probability distribution.
Intuitively, this bound shows that if our test accepts with high probability then $p_\psi(\vec a) \approx 2^{-n}$ with high probability.
Indeed, let us consider
\[ M_0 \coloneqq \{ \vec a\in \mathcal V_n : 2^n p_\psi(\vec a) > 1/2 \}. \]
Then Markov's inequality (which can be applied since it is always true that $p_\psi\leq2^{-n}$) asserts that
\begin{equation}\label{eq:qubit markov}
  \sum_{\vec a \in M_0} p_\psi(\vec a)
% = \Pr(2^n p_\psi(\vec a) > 1/2) \\
% = 1 - \Pr(2^n p_\psi(\vec a) \leq 1/2) \\
% = 1 - \Pr(1 - 2^n p_\psi(\vec a) \geq 1/2) \\
\geq 1 - 2 \sum_{\vec a} p_\psi(\vec a) \left( 1 - 2^n p_\psi(\vec a) \right)
% &= 1 - 2 \left(2 - \sum_{\vec a} p_\psi(\vec a) \left( 1 + 2^n p_\psi(\vec a) \right) \right) \\
= 1 - 4\left(1 - p_\text{accept}\right).
\end{equation}
The choice of threshold $1/2$ in the definition of $M_0$ ensures that the Weyl operators corresponding to any two points $\vec a,\vec b\in M_0$ commute.
To see, we use that any pair of \emph{anti}commuting $W_{\vec a},W_{\vec b}$ can by a base change be mapped onto the Pauli operators $X$,$Z$; it can then verified on the Bloch sphere that there exists no qubit state $\rho$ such that both $\tr[\rho X]^2>1/2$ and $\tr[\rho Z]^2>1/2$ (see \cref{fig:bloch plane} for a graphical proof).

Let us now extend the set $M_0$ to some maximal set $M$ such that the corresponding Weyl operators commute.
Then $M$ is automatically a Lagrangian subspace, of dimension~$n$.%
\footnote{It is natural to ask whether the subspace~$M$ is uniquely determined by $M_0$.
This is the case when, e.g., $p_\text{accept}>7/8$.
Indeed, in this case, \cref{eq:qubit markov} implies that
$2^{-n} \lvert M_0\rvert \geq \sum_{\vec a \in M_0} p_\psi(\vec a) > 1 - 4\left(1 - 7/8\right) = 1/2$, so $\lvert M_0\rvert > 2^{n-1}$.
It follows that~$M_0$ spans an $n$-dimensional subspace which is necessarily contained in, and hence equal to, $M$.}
As discussed in \cref{subsec:stabs}, it determines a whole basis of stabilizer states, $\{\ket{M,f}\}_f$.
Thus:
\begin{align*}
\max_S \lvert\braket{S|\psi}\rvert^2
&\geq \max_f \braket{M,f|\psi|M,f}
\geq \sum_f \braket{M,f|\psi|M,f}^2
= \tr\left[\Lambda_M[\psi]^2\right] \\
% &= 2^{-2n} \sum_{\vec x,\vec y\in M} \tr\left[W_{\vec x} \psi W_{\vec x}^\dagger W_{\vec y} \psi W_{\vec y}^\dagger\right] \\
&= 2^{-2n} \sum_{\vec x,\vec y\in M} \tr\left[\psi W_{\vec x}^\dagger W_{\vec y} \psi (W_{\vec x}^\dagger W_{\vec y})^\dagger \right]
= 2^{-n} \sum_{\vec x\in M} \tr\left[\psi W_{\vec x} \psi W_{\vec x}^\dagger \right] \\
&= \sum_{\vec x\in M} p_\psi(\vec x)
\geq \sum_{\vec x\in M_0} p_\psi(\vec x)
\geq 1 - 4\left(1 - p_\text{accept}\right)
\end{align*}
where we used \cref{eq:measurement in stabilizer basis} for the measurement $\Lambda_M$ in the stabilizer basis; the last bound is \cref{eq:qubit markov}.
In particular, if $\max_S \lvert\braket{S|\psi}\rvert^2 \leq 1-\eps^2$ then $p_\text{accept} \leq 1-\eps^2/4$.
\end{proof}

Our theorem has the following consequence for quantum property testing, resolving an open question first raised by Montanaro and de Wolf~\cite[Question 7]{montanaro2013survey}.

\begin{cor}
  Let $\psi$ be a pure state of $n$ qubits and let $\eps>0$.
  Then there exists a quantum algorithm that, given $O(1/\eps^2)$ copies of $\psi$, accepts any stabilizer state (it is \emph{perfectly complete}), while it rejects states such that $\max_S \lvert\braket{S|\psi}\rvert^2\leq1-\eps^2$ with probability at least $2/3$.
\end{cor}

Before our result, the best known algorithms required a number of copies that scaled linearly with $n$, the number of qubits.
Indeed, these algorithms proceeded by attempting to \emph{identify} the stabilizer state, which requires $\Omega(n)$ copies by the Holevo bound~\cite{aaronson2008identifying,montanaro2017learning,zhao2016fast}.
Moreover, our algorithm is manifestly efficient (see the circuit in \cref{fig:circuit qubits}).

\begin{rem}
For multi-qubit states $\psi$ that are \emph{real} in the computational basis, we can replace step~1 of the algorithm by a single Bell sampling, which in this case directly samples from the characteristic distribution $p_\psi$ (see \cref{eq:bell sampling}).
The resulting algorithm operators on \emph{four} copies of $\psi$ and achieves the same guarantees as \cref{thm:bell difference sampling qubits}.
% Indeed, \cref{eq:qubit acceptance} is then an equality and the remainder of the proof proceeds as before.
\end{rem}

\begin{rem}
The scaling in \cref{thm:bell difference sampling qubits} is optimal.
Indeed, it is known that distinguishing any fixed pair of states $\ket\psi,\ket\phi$ with $\lvert\braket{\psi|\phi}\rvert^2=1-\eps^2$ requires $\Omega(1/\eps^2)$ copies~\cite{montanaro2013survey}.
In particular, this lower bound holds if we choose $\ket\psi$ to be a stabilizer state and $\ket\phi$ a state that is $\eps$-far away from being a stabilizer state, in which case our \cref{algo:qubits} is applicable.
\end{rem}

\begin{rem}[Clifford testing]
It follows from \cref{thm:main qubits} that we can also test whether a given unitary $U$ is in the Clifford group or not (without given access to $U^\dagger$).
This resolves another open question in the survey of Montanaro and de Wolf~\cite[Question 9]{montanaro2013survey}.

Indeed, given black-box access to $U$ alone we can create the Choi state $\ket{U} \coloneqq (U \ot I) \ket{\Phi^+}$, which is a stabilizer state if and only if $U$ is a Clifford unitary.
Moreover, the ``average case'' distance measure used in the literature for quantum property testing of unitaries is precisely equal to trace distance between the corresponding Choi states~\cite[Section~5.1.1]{montanaro2013survey}.
% \begin{align*}
%   \braket{U|V} = d^{-n} \tr[U^\dagger V] = \braket{A,B}
% \end{align*}
Thus, by first creating the Choi state and then running our \cref{algo:qubits} we can efficiently test whether a given unitary $U$ is a Clifford unitary.
\end{rem}

It is instructive to write down the accepting POVM element for \cref{algo:qubits}.
From \cref{eq:Pi_a,eq:same outcome}, we find that it is given by
\begin{align}\label{eq:accepting povm}
\Pi_\text{accept}
= \sum_{\vec a} \Pi_{\vec a} \ot \frac {I + W_{\vec a} \ot W_{\vec a}^\dagger} 2
= \frac12 \left( I + U \right),
\end{align}
where we have introduced the unitary
\begin{align*}
U
= \frac1{2^{2n}} \sum_{\vec x,\vec a\in\mathcal V_n} (-1)^{[\vec a,\vec x]} W_{\vec x}^{\ot 4} \ot W_{\vec a}^{\ot 2}
= \biggl( \underbrace{\frac14 \sum_{\vec x,\vec a\in\mathcal V_1} (-1)^{[\vec a,\vec x]} W_{\vec x}^{\ot 4} \ot W_{\vec a}^{\ot 2}}_{=: u} \biggr)^{\ot n}.
\end{align*}
It is easy to verify that $U=u^{\ot n}$ is a Clifford unitary acting on the space $\mathcal H_n^{\ot 6} \cong \mathcal H_{6n}$ of $6n$ qubits.

For any pure state $\psi$, $\psi^{\ot n}$ is in the symmetric subspace, and so invariant under left and right-multiplication by permutations.
In particular, we obtain a test of the same goodness as \cref{thm:main qubits} if we replace $U$ by $V=U (I^{\ot 4} \ot \FF)$, where $\FF=R((1 2))$ denotes the operator that swaps (or flips) two blocks of $n$~qubits.
% Since $\FF=(I \ot I + X \ot X + Y \ot Y + Z \ot Z)/2=2^{-1} \sum_{\vec b\in\mathcal V_1} W_{\vec b}^{\ot2}$, we obtain
Since $\FF=2^{-n} \sum_{\vec b} W_{\vec b}^{\ot2}$,%
% (for qubits; in general there is a minus sign)
we obtain the formula
\begin{equation}\label{eq:V for qubits}
\begin{aligned}
  V
% &= U (I^{\ot 4} \ot \FF_1)
&= 2^{-3n} \sum_{\vec x,\vec a,\vec b} (-1)^{[\vec a,\vec x]} W_{\vec x}^{\ot 4} \ot (W_{\vec a} W_{\vec b})^{\ot 2}
% &= 2^{-3n} \sum_{\vec x,\vec a,\vec b} (-1)^{[\vec a,\vec x+\vec b]} W_{\vec x}^{\ot 4} \ot W_{\vec a+\vec b}^{\ot 2} \\
= 2^{-3n} \sum_{\vec x,\vec a,\vec b} (-1)^{[\vec a,\vec x+\vec b]} W_{\vec x}^{\ot 4} \ot W_{\vec a+\vec b \bmod 2}^{\ot 2} \\
% &= 2^{-3n} \sum_{\vec x,\vec a,\vec b} (-1)^{[\vec a,\vec x+\vec a+\vec b]} W_{\vec x}^{\ot 4} \ot W_{\vec a+\vec b \bmod 2}^{\ot 2} \\
&= 2^{-3n} \sum_{\vec x,\vec a,\vec b} (-1)^{[\vec a,\vec x+\vec b]} W_{\vec x}^{\ot 4} \ot W_{\vec b}^{\ot 2}
% &= 2^{-3n} \sum_{\vec x,\vec b} 2^{2n} \delta_{\vec x,\vec b} W_{\vec x}^{\ot 4} \ot W_{\vec b}^{\ot 2} \\
= 2^{-n} \sum_{\vec x\in\mathcal V_n} W_{\vec x}^{\ot 6}
= \biggl( \underbrace{\frac12 \sum_{\vec x\in\mathcal V_1} W_{\vec x}^{\ot 6}}_{=: v} \biggr)^{\ot n}.
\end{aligned}
\end{equation}
Thus, we recognize that the unitary~$V=v^{\ot n}$ is precisely the action of the \emph{anti-identity}~\eqref{eq:anti identity} described in the introduction (for $t=6$):
\begin{align}\label{eq:anti identity six}
  V = R(\bar\id) = 2^{-n} \left( I^{\ot6} + X^{\ot6} + Y^{\ot6} + Z^{\ot6} \right)^{\otimes n}
\end{align}
See also \cref{rem:qubits four vs six}.
We discuss anti-permutations in more detail in \cref{dfn:anti-permutation}.

\Cref{eq:V for qubits} allows us to express the acceptance probability of \cref{algo:qubits} in an interesting way:
\begin{align}
\nonumber
p_\text{accept}
&= \tr\left[\psi^{\ot 6} \Pi_\text{accept}\right]
= \frac12\left( 1 + \tr\left[\psi^{\ot6} U\right] \right)
= \frac12\left( 1 + \tr\left[\psi^{\ot6} V\right] \right) \\
\nonumber
&= \frac12\left( 1 + 2^{-n} \sum_{\vec x} \tr\left[\psi^{\ot6} W_{\vec x}^{\ot6} \right] \right)
= \frac12\left( 1 + 2^{2n} \sum_{\vec x} c_\psi(\vec x)^6 \right)\\
\nonumber
&= \frac12\left( 1 + 2^{2n} \lVert c_\psi \rVert_{\ell_6}^6 \right)
= \frac12\left( 1 + 2^{2n} \lVert p_\psi \rVert_{\ell_3}^3 \right) \\
\label{eq:better qubit acceptance}
&= \frac12\left( 1 + 2^{2n} \sum_{\vec x} p_\psi(\vec x)^3 \right)
= \sum_{\vec x} p_\psi(\vec x) \frac12\left( 1 + 2^{2n} p_\psi(\vec x)^2 \right).
\end{align}
It is intuitive that the $\ell_p$-norms should appear, since stabilizer states can be characterized by having a maximally peaked characteristic function and distribution (\cref{eq:characteristic function stabilizer,eq:characteristic distribution stabilizer}).

In fact, the result of this calculation is plainly a strengthening of \cref{eq:qubit acceptance}, since $2^n p_\psi(\vec x)\leq1$.
If we follow the rest of the proof of \cref{thm:main qubits} then we obtain $p_\text{accept}\leq1-3\eps^2/8$, a slight improvement.
More importantly, though, this argument completely avoids the analysis of Bell difference sampling in \cref{thm:bell difference sampling qubits}.
This leads us towards an approach for testing general qu$d$it stabilizer states.

%-----------------------------------------------------------------------------
\subsection{Qudit stabilizer testing}\label{subsec:qudit stabilizer testing}
%-----------------------------------------------------------------------------

While Bell sampling can only be used for qubit systems, \cref{eq:V for qubits} has a clear generalization to arbitrary qu$d$its.
Let $d\geq2$ and consider the operator
\begin{align}\label{eq:V for qudits}
  V_s = d^{-n} \sum_{\vec x} (W_{\vec x} \ot W_{\vec x}^\dagger)^{\ot s}.
\end{align}
(For qubits, the Weyl operators are Hermitian and so $V_3$ is precisely \cref{eq:V for qubits}.)
Suppose we choose $s$ such that $V_s$ is a Hermitian unitary (we will momentarily see that this can always be done).
Then
\begin{align*}%\label{eq:Pi_accept qudits}
\Pi_{s,\text{accept}} = \frac12(I+V_s)
\end{align*}
is a projection.
If we think of it as the accepting element of a binary POVM then
\begin{equation}\label{eq:qudit acceptance}
\begin{aligned}
p_\text{accept}
&= \tr[\psi^{\ot 2s} \Pi_{s,\text{accept}}]
= \frac12 \left( 1 + \tr[\psi^{\ot s} V_s] \right)
= \frac12 \left( 1 + d^{-n} \sum_{\vec x} \lvert\tr[\psi W_{\vec x}]\rvert^{2s} \right) \\
&= \frac12 \left( 1 + d^{(s-1)n} \sum_{\vec x} p^s_\psi(\vec x) \right)
= \sum_{\vec x} p_\psi(\vec x) \frac12 \left( 1 + d^{(s-1)n} p^{s-1}_\psi(\vec x) \right),
\end{aligned}
\end{equation}
which generalizes \cref{eq:better qubit acceptance}.

When is $V_s$ Hermitian and unitary?
It is always Hermitian, since $W_{\vec x} \ot W_{\vec x}^\dagger$ only depends on $\vec x$ modulo~$d$.
For unitarity we use \cref{eq:projective group law} and calculate
\begin{align*}
  V_s^2
&= d^{-2n} \sum_{\vec x,\vec y} (W_{\vec x} W_{\vec y} \ot W_{\vec x}^\dagger W_{\vec y}^\dagger)^{\ot s}
= d^{-2n} \sum_{\vec x,\vec y} \omega^{s[\vec x,\vec y]} (W_{\vec x + \vec y} \ot W_{-(\vec x+\vec y)})^{\ot s} \\
&= d^{-2n} \sum_{\vec x,\vec y} \omega^{s[\vec x,\vec y]} (W_{\vec x + \vec y \bmod d} \ot W_{-(\vec x+\vec y \bmod d)})^{\ot s}
= d^{-2n} \sum_{\vec z} \left( \sum_{\vec x} \omega^{s[\vec x,\vec z]} \right) (W_{\vec z} \ot W_{-\vec z})^{\ot s}.
\end{align*}
If $s$ is invertible modulo $d$ then $\omega^{s[-,\vec z]}$ is a nontrivial character for all $\vec z$, and so the inner sum simplifies to $d^{2n} \delta_{\vec z,\vec0}$.
It follows that $V_s^2 = I$, as desired.
We summarize:

\begin{lem}\label{lem:hermitian unitary}
Let $d\geq2$ and $s$ an integer that is invertible modulo $d$ (i.e., $(s,d)=1$).
Then $V_s$ is a Hermitian unitary.
\end{lem}

\begin{rem}[Qubits]\label{rem:qubits four vs six}
For qubits, the operator $V_s$ is a Hermitian unitary if and only if $s$ is odd.
E.g., for $s=1$ it is the unitary swap operator $\FF$ and for $s=3$ it is precisely \cref{eq:V for qubits} (the anti-identity), while for $s=2$ it is \emph{not} unitary but in fact \emph{proportional} to one of the POVM elements from Bell difference sampling.
Indeed, $V_2=2^n \Pi_0$ where $\Pi_0$ is the projection from \cref{eq:projection Pi_0}.
Thus $\lVert V_2\rVert=2^n$ and so we \emph{cannot} interpret the associated $\Pi_2$ as a POVM element.
This already partly explains why we had to resort to six copies to test stabilizerness.
\end{rem}

The second ingredient used to establish \cref{thm:main qubits} was an uncertainty principle for Weyl operators.
The following lemma supplies this for general $d$:

\restateLemQuditUncertainty
\begin{proof}
  Note that
  \[ \lVert W_{\vec x} \ket \psi - \ket\psi \braket{\psi | W_{\vec x} | \psi} \rVert < \delta \]
  and likewise for $W_{\vec y}$.
  By the triangle inequality,
  \begin{align*}
  &\quad \lVert W_{\vec x} W_{\vec y} \ket \psi - \ket\psi \braket{\psi | W_{\vec x} | \psi} \braket{\psi | W_{\vec y} | \psi} \rVert \\
  &\leq \lVert W_{\vec x} W_{\vec y} \ket \psi - W_{\vec x}\ket\psi \braket{\psi | W_{\vec y} | \psi} \rVert
  + \lVert W_{\vec x}\ket\psi \braket{\psi | W_{\vec y} | \psi} - \ket\psi \braket{\psi | W_{\vec x} | \psi} \braket{\psi | W_{\vec y} | \psi} \rVert \\
  &\leq \lVert W_{\vec x} \rVert  \lVert W_{\vec y} \ket \psi - \ket\psi \braket{\psi | W_{\vec y} | \psi} \rVert
  + \lVert W_{\vec x}\ket\psi - \ket\psi \braket{\psi | W_{\vec x} | \psi} \rVert  \braket{\psi | W_{\vec y} | \psi} \\
  &\leq \lVert W_{\vec y} \ket \psi - \ket\psi \braket{\psi | W_{\vec y} | \psi} \rVert
  + \lVert W_{\vec x}\ket\psi - \ket\psi \braket{\psi | W_{\vec x} | \psi} \rVert
  < 2\delta,
  \end{align*}
  but also
  \begin{align*}
  &\quad \lVert W_{\vec x} W_{\vec y} \ket \psi - \omega^{[x,y]} \ket\psi \braket{\psi | W_{\vec x} | \psi} \braket{\psi | W_{\vec y} | \psi} \rVert \\
  &= \lVert \omega^{[x,y]} W_{\vec y} W_{\vec x} \ket \psi - \omega^{[x,y]} \ket\psi \braket{\psi | W_{\vec y} | \psi} \braket{\psi | W_{\vec x} | \psi} \rVert
  < 2\delta.
  \end{align*}
  If we combine this with another triangle inequality, we obtain that
  \begin{align*}
  \lvert 1 - \omega^{[x,y]} \rvert = \lVert \omega^{[x,y]} \ket\psi - \ket\psi\rVert
  < \frac {4\delta}{\braket{\psi | W_{\vec x} | \psi} \braket{\psi | W_{\vec y} | \psi}}
  < \frac {4\delta}{1 - \delta^2}.
  \end{align*}
  Now suppose that $W_{\vec x}$ and $W_{\vec y}$ do not commute.
  Then $[x,y]\neq0$ and so
  \begin{align*}
    \lvert 1 - \omega^{[x,y]} \rvert
  \geq \lvert 1 - \omega \rvert
  = 2\sin(\pi/d)
  % \geq 2 \frac 2 \pi \frac \pi d
  \geq \frac 4 d.
  \end{align*}
  Thus, $4/d < 4\delta/(1-\delta^2)$, which plainly contradicts our choice of $\delta$.
  This is the desired contradiction and we conclude that $W_{\vec x}$ and $W_{\vec y}$ commute.
  % \sqrt 3 (\sqrt 7 - 2) / d works for d>=3
\end{proof}

We now show that stabilizer testing can be done in arbitrary local dimension:

\restateThmMainQudits
\begin{proof}
If $\psi$ is a stabilizer state, say $\ket\psi=\ket{M,f}$, then $p_\psi(\vec x)$ is the uniform distribution on $M$, which has $d^n$ elements.
In view of \cref{eq:qudit acceptance},
\begin{align*}
  p_\text{accept}
= \sum_{\vec x} p_\psi(\vec x) \frac12 \left( 1 + d^{(s-1)n} p^{s-1}_\psi(\vec x) \right) = 1,
\end{align*}
so the test accepts with certainty.

Now suppose that $\psi$ is a general state.
Define
\begin{align*}
  M_0 \coloneqq \{ \vec x \in \mathcal V_n : d^n p_\psi(\vec x) > 1 - 1/4d^2 \}.
\end{align*}
By \cref{lem:qudit uncertainty}, the Weyl operators $W_{\vec x}$ for $\vec x\in M_0$ all commute.
We can thus extend $M_0$ to a maximal set $M$ with this property.
As in the proof of \cref{thm:main qubits}, we can bound
\begin{align*}
  \max_S \lvert\braket{S|\psi}\rvert^2 \geq \sum_{\vec x\in M_0} p_\psi(\vec x).
\end{align*}
But this probability can be bounded as before using the Markov inequality (but now for a $(s-1)$st moment):
\begin{align*}
\sum_{\vec x\in M_0} p_\psi(\vec x)
&= 1 - \smashoperator{\sum_{d^n p_\psi(\vec x) \leq 1 - 1/4d^2}} p_\psi(\vec x)
% &= 1 - \smashoperator{\sum_{1 - d^{(s-1)n} p_\psi^{s-1}(\vec x) \geq 1 - (1 - 1/4d^2)^{s-1}}} p_\psi(\vec x) \\
\geq 1 - \frac {\sum p_\psi(\vec x) \left( 1 - d^{(s-1)n} p_\psi^{s-1}(\vec x) \right)} {1 - (1 - 1/4d^2)^{s-1}} \\
&= 1 - \frac 2 {1 - (1 - 1/4d^2)^{s-1}} \left(1 - p_\text{accept}\right).
\end{align*}
The last equality is \cref{eq:qudit acceptance}.
This yields the desired bound.
\end{proof}

\begin{rem}
It is clear that $s=d+1$ is always a valid choice in \cref{thm:main qudits}.
This leads to $C_{d,s} \approx 1/8d$ for large $d$, but the resulting test involves gates that act on $2d+2$ qu$d$its at a time.
However, this choice of $s$ is in general rather pessimistic.
E.g., if $d$ is odd then we may always choose $s=2$, meaning that our test acts on four copies at a time.
\end{rem}

\begin{cor}
  Let $d\geq2$ and fix $s$ as in \cref{thm:main qudits}.
  Let $\psi$ be a pure state of $n$ qu$d$its and let $\eps>0$.
  Then there exists an quantum algorithm that, given $O(1/C_{d,s}\eps^2)$ copies of $\psi$, accepts any stabilizer state (it is \emph{perfectly complete}), while it rejects states such that $\max_S \lvert\braket{S|\psi}\rvert^2\leq1-\eps^2$ with probability at least $2/3$.
\end{cor}

It is clear that the POVM measurement $\{\Pi_{s,\text{accept}},I-\Pi_{s,\text{accept}}\}$ can be implemented efficiently.
Using phase estimation, it suffices to argue that the controlled version of $V_s$ can be implemented efficiently.
But $V_s = v_s^{\ot n}$, so its controlled version is equal to a composition of $n$ controlled versions of $v_s$, each of which acts only on a constant (with respect to $n$) number of qu$d$its.
It follows that our stabilizer test for qu$d$its is efficient.

\medskip

It is instructive to compute the action of the unitary $V_s=v_s^{\ot n}$ more explicitly:
Let $\ket{\vec x}=\ket{\vec x_1,\dots,\vec x_{2s}}$ denote a computational basis vector of $\mathcal H_n^{\ot 2s}$.
Then, using \cref{eq:weyl operator},
\begin{align*}
V_s \ket{\vec x}
&= d^{-n} \sum_{\vec a\in\mathcal V_n} (W_{\vec a} \ot W_{\vec a}^\dagger)^{\otimes s} \ket{\vec x}
% &= d^{-n} \quad\mathclap{\sum_{\vec p,\vec q\in\ZZ_d^n}}\qquad \omega^{-s\vec p\cdot\vec q} (\vec Z^{\vec p} \vec X^{\vec q} \ot \vec Z^{-\vec p} \vec X^{-\vec q})^{\otimes t} \ket{\vec x} \\
% &= d^{-n} \quad\mathclap{\sum_{\vec p,\vec q\in\ZZ_d^n}}\qquad \omega^{-s\vec p\cdot\vec q} (\vec Z^{\vec p} \ot \vec Z^{-\vec p})^{\otimes t} \ket{\vec x_1 + \vec q, \vec x_2 - \vec q,\dots} \\
% &= d^{-n} \quad\mathclap{\sum_{\vec p,\vec q\in\ZZ_d^n}}\qquad \omega^{-s\vec p\cdot\vec q} \omega^{\vec p \cdot (\vec x_1 + \vec q)} \omega^{-\vec p \cdot (\vec x_2 - \vec q)} \cdots \ket{\vec x_1 + \vec q, \vec x_2 - \vec q,\dots} \\
= d^{-n} \smashoperator{\sum_{\vec p,\vec q\in\ZZ_d^n}} \omega^{\vec p\cdot(s \vec q + \vec x_1 - \vec x_2 + \dots - \vec x_{2s})} \ket{\vec x_1 + \vec q, \vec x_2 - \vec q,\dots} \\
&= d^{-n} \smashoperator{\sum_{\vec p,\vec q\in\ZZ_d^n}} \omega^{s \vec p\cdot(\vec q + \bar{\vec x}_\text{odd} - \bar{\vec x}_\text{even})} \ket{\vec x_1 + \vec q, \vec x_2 - \vec q,\dots}
= \ket{\vec x_1 - \bar{\vec x}_\text{odd} + \bar{\vec x}_\text{even}, \vec x_2 + \bar{\vec x}_\text{odd} - \bar{\vec x}_\text{even},\dots},
\end{align*}
where $\bar{\vec x}_\text{even} = s^{-1} \sum_{k\text{ even}} \vec x_k$ and $\bar x_\text{odd}$ is defined analogously.
If we re-order the tensor factors so that the odd systems come first, followed by the even ones, we find that a basis vector $\ket{\vec x_\text{odd},\vec x_\text{even}}$ is mapped to $\ket{\vec x_\text{odd} - \bar{\vec x}_\text{odd} + \bar{\vec x}_\text{even},\vec x_\text{even} + \bar{\vec x}_\text{odd} - \bar{\vec x}_\text{even}}$.
Thus, $V_s$ is a unitary that permutes the computational basis vectors by ``swapping the mean'' of the even and the odd sites of the $2s$ many blocks of $n$ qu$d$its.

Here is one last reformulation that will be useful to connect to our algebraic results.
Let $\vec p_{2s} = (-1,1,\dots,-1,1) \in \ZZ_d^{2s}$ denote the `parity vector' that is $\pm1$ on even/odd sites, and consider the following $2s \times 2s$ matrix with entries in $\ZZ_d$:
\begin{equation}\label{eq:orthogonal and stochastic}
  \tilde\id =\id - s^{-1} \vec p_{2s} \vec p_{2s}^T
\end{equation}
Then we can write the action of $V_s$ as
\begin{align}\label{eq:V_s as R(T)}
  V_s \ket{\vec x} = \ket{\tilde\id (\vec x_1,\dots,\vec x_{2s})} = \ket{(\tilde\id \ot I_n)\vec x}.
\end{align}
It is easy to verify that $\tilde\id$ is a stochastic isometry (cf.~\cref{eq:antiperm 2} in \cref{subsec:examples}).
For qubits and $s=3$, $\tilde\id$ is just the matrix obtained by taking the $6\times6$ identity matrix and inverting each bit (the `anti-identity').
This gives a pleasant and insightful interpretation of \cref{eq:V for qubits}, as we will see in \cref{subsec:stabilizer testing revisited}.
Interestingly, the anti-identity has previously appeared in the classification of Clifford gates in~\cite{grier2016classification} (their $T_6$).

%=============================================================================
\section{Algebraic theory of Clifford tensor powers}\label{sec:algebra}
%=============================================================================
In this section, we present a general framework for studying the algebraic structure of stabilizer states and Clifford operators.
We start by describing the commutant of the tensor powers of the Clifford group, where we obtain results similar in flavor to the Schur-Weyl duality between the unitary group and the symmetric group.
Next, we apply this machinery to compute arbitrary moments of qu$d$it stabilizer states, and we describe how to construct $t$-designs of arbitrary order from weighted Clifford orbits.
Lastly, we return to the stabilizer testing problem and explain how our solution from \cref{sec:stabilizer testing} can be understood more systematically and generalized.
In particular, we find an optimal projection that characterizes the tensor powers of stabilizer states precisely.

Throughout this section we assume that $d$ is prime.

%-----------------------------------------------------------------------------
\subsection{Commutant of Clifford tensor powers}\label{sec:commcliff}
%-----------------------------------------------------------------------------
Schur-Weyl duality in its most fundamental form asserts that any operator on $(\CC^D)^{\ot t}$ that commutes with $U^{\ot t}$ for all unitaries $U\in U(D)$ is necessarily a linear combination of permutation operators.
Using the double commutant theorem, this implies at once that $(\CC^d)^{\ot t} = \bigoplus_\lambda V_{U(D),\lambda} \ot V_{S_t,\lambda}$, where the $V_{U(D),\lambda}$ and $V_{S_t,\lambda}$ are pairwise inequivalent irreducible representations of the unitary group $U(D)$ and of the symmetric group $S_t$, respectively.

The main result of this section is that the commutant of the tensor powers of the Clifford group can be completely described in terms of a natural generalization of permutation operators (see \cref{thm:commutant} below).
Mathematically, this generalization involves Lagrangian subspaces of a space equipped with a quadratic form.
Since stabilizer states can be described in terms of Lagrangian subspaces with respect to a symplectic form (\cref{sec:preliminaries}), this is reminiscent of Howe's classical duality between sympletic and orthogonal group actions.

To describe the result more precisely, let $T$ denote a subspace of $\ZZ_d^t \op \ZZ_d^t$.
We define a corresponding operator
\begin{align*}
  r(T) = \sum_{(\vec x,\vec y)\in T} \ket{\vec x}\bra{\vec y}
\end{align*}
on $(\CC^d)^{\ot t}$, where $\ket{\vec x} = \ket{x_1,x_2,\dots,x_t} \in (\CC^d)^{\ot t}$ denotes the computational basis vector associated with some $\vec x\in\ZZ_d^t$.
We also consider the $n$-fold tensor power
\begin{align*}
  R(T)\coloneqq r(T)^{\ot n},
\end{align*}
 which is an operator on
$((\CC^d)^{\ot t})^{\ot n} \cong (\CC^d)^{\ot t n} \cong ((\CC^d)^{\ot n})^{\ot t}.$
Both $r(T)$ and $R(T)$ are represented by real matrices in the computational basis.

\restateDfnLagrangianStochastic

\noindent See~\cite[App.~C]{nezami2016multipartite} for a complete list of the subspaces $\Sigma_{t,t}(d)$ for $t=3$, and \cref{subsec:examples} for examples.

In \cref{lem:commutant}, we will show that the operators~$R(T)$ are indeed in the commutant of $\Cliff(n,d)^{\ot t}$.
The proof is straightforward and elucidates the role of the three conditions in \cref{dfn:lagrangian stochastic} as well as the difference between even and odd~$d$.

\begin{rem}\label{rem:bilinear form}
Recall that a subspace $T$ is called \emph{totally isotropic} with respect to a quadratic form~$\mathfrak q$ if $\mathfrak q(\vec v) = 0$ for every $\vec v\in T$.
This explain our terminology in \cref{dfn:lagrangian stochastic}.

We can also consider the $\ZZ_d$-valued bilinear form $\mathfrak b((\vec x, \vec y),(\vec x',\vec y')) \coloneqq \vec x\cdot\vec x' - \vec y\cdot\vec y' \in \ZZ_d$.
By a straightforward calculation,
\begin{align}\label{eq:bilinear vs quadratic form}
  \mathfrak q(\vec v + \vec w) = \mathfrak q(\vec v) + \mathfrak q(\vec w) + 2 \mathfrak b(\vec v, \vec w) \pmod D
\end{align}
for all $\vec v$, $\vec w\in \ZZ_d^{2t}$.
Thus, $\mathfrak q$ is a $\ZZ_D$-valued quadratic form associated to the $\ZZ_d$-bilinear form~$\mathfrak b$ in the sense of~\cite{wood1993witt}.
Note that if $T$ is totally isotropic with respect to $\mathfrak q$ then \cref{eq:bilinear vs quadratic form} shows that $T$ is self-orthogonal, i.e., $T \subseteq T^\perp$, where
\begin{align*}
  T^\perp \coloneqq \{ \vec v \in \ZZ_d^{2t} \;:\; \mathfrak b(\vec v,\vec w) = 0 \quad \forall \vec w\in T \}.
\end{align*}
If $d$ is odd then $\mathfrak q(\vec v)=\mathfrak b(\vec v,\vec v)$, so any self-orthogonal subspace is automatically totally isotropic with respect to~$\mathfrak q$.

If $d=2$ then \cref{eq:bilinear vs quadratic form} implies that, for a self-orthogonal subspace, the set of isotropic vectors forms a subspace -- so we can check total isotropicity on a basis.
Moreover, for $d=2$, if $T$ is Lagrangian then it is automatically stochastic;
indeed, $\mathfrak b(\vec v, \vec 1_{2t}) = \mathfrak q(\vec v) \pmod 2$, so $\vec 1_{2t}$ is contained in any \emph{maximal} totally isotropic subspace.
\end{rem}

Our goal of this section it to prove the following theorem:

\restateThmCommutant

It is instructive to discuss a few key features of \cref{thm:commutant}.
First, we know that the permutation group on $t$ elements, $S_t$, is in the commutant of the Clifford group $\Cliff(n,d)$, because it is even in the commutant of the larger unitary group~$U(d^n)$.
Indeed, let $\pi \cdot \vec y = (y_{\pi^{-1}(1)},\dots,y_{\pi^{-1}(t)})$ denote the permutation action of~$S_t$ on $\ZZ_d^t$.
The one can see that, for any permutation $\pi\in S_t$, the subspace~$T_\pi = \{(\pi\cdot\vec y,\vec y) : \vec y\in\ZZ_d^t\}$ is Lagrangian and stochastic.
The corresponding operator $R(T_\pi)=r(T_\pi)^{\ot n}$ agrees precisely with the usual permutation action of $S_t$ on $((\CC^d)^{\ot n})^{\ot t}$. %, confirming our expectation. \leftarrow it is in the
Accordingly, we may identify $S_t$ with a subset of $\Sigma_{t,t}(d)$.
We will see below in \cref{dfn:O_t} that the set of subspaces $T$ for which $R(T)$ is invertible forms a (in general, proper) subgroup that is (in general, strictly) larger than $S_t$.

Remarkably, \cref{thm:commutant} shows that the size of the commutant \emph{stabilizes} as soon as $n\geq t-1$.
That is, just like for the symmetric group in Schur-Weyl duality of $D$, the set $\Sigma_{t,t}(d)$ that parametrizes the commutant of the Clifford tensor powers is independent of~$n$, the number of qudits, provided that $n\geq t-1$.
This stabilization, along with the fact that the operators $R(T)=r(T)^{\ot n}$ are tensor powers, are highly useful properties in applications (e.g.,~\cite{nezami2016multipartite} and \cref{sec:moments,subsec:stabilizer testing revisited} below).

\begin{rem}
We believe that the results of Nebe et al~\cite{nebe2006self} show that the operators $R(T)$ span the commutant of $\Cliff(n,d)^{\ot t}$ for any value of $n$.
But we caution that if $n<t-1$ then the $R(T)$ are in general no longer linearly independent (e.g.,~\cite[eqs.~(9) and (10)]{zhu2015multiqubit}).
\end{rem}

\Cref{thm:commutant} will be established by combining a number of intermediate results of independent interest.
We first show that the operators~$R(T)$ are indeed in the commutant of $\Cliff(n,d)^{\ot t}$.

\begin{lem}\label{lem:commutant}
For every $T\in\Sigma_{t,t}(d)$ and for every $U\in\Cliff(n,d)$, we have that
$[R(T), U^{\ot t}] = [r(T)^{\ot n}, U^{\ot t}] = 0$.
\end{lem}
\begin{proof}
Up to global phases, the Clifford group is generated by the following three operators, which are allowed to act on arbitrary qudits or pairs of qudits~\cite{gottesman1999heisenberg,farinholt2014ideal,nielsen2002universal}:
The Fourier transform (also known as the Hadamard gate for $d=2$),
\begin{align*}
  H=\frac1{\sqrt d}\sum_{a,b\in\ZZ_d} \omega^{ab} \ket a\bra b,
\end{align*}
the phase gate, which is defined as
\begin{align*}
  P=\sum_{a\in\ZZ_2} i^{a^2} \ket a\bra a \text{ for $d=2$,} \qquad
  P=\sum_{a\in\ZZ_d} \omega^{2^{-1}a(a-1)} \ket a\bra a \text{ for $d\neq2$,}
\end{align*}
(here we use that for $d=2$, $a^2$ is well-defined modulo four, while for odd $d$, $2$ has a multiplicative inverse, denoted $2^{-1}$),
and the controlled addition (also known as the CNOT gate for $d=2$)
\begin{align*}
  \operatorname{CADD}=\sum_{a,b\in\ZZ_d} \ket{a,a+b}\bra{a,b}.
\end{align*}
To establish the lemma we will prove the claim for each generator (cf.~\cite{nebe2006self}).

The Fourier transform~$H$ is a one-qudit gate, so it suffices to show that $[H^{\ot t},r(T)]=0$ for every~$T\in\Sigma_{t,t}(d)$.
Indeed:
\begin{align*}
  H^{\ot t} r(T) H^{\dagger, \ot t}
&= d^{-t} \sum_{\vec a, \vec b\in\ZZ_d^t} \sum_{(\vec x,\vec y)\in T} \omega^{\vec a\cdot\vec x-\vec b\cdot\vec y} \ket{\vec a}\bra{\vec b}
= d^{-t} \sum_{\vec a, \vec b\in\ZZ_d^t} \sum_{(\vec x,\vec y)\in T} \omega^{\mathfrak b((\vec a,\vec b),(\vec x,\vec y))} \ket{\vec a}\bra{\vec b} \\
&= \sum_{(\vec a, \vec b)\in T^\perp} \ket{\vec a}\bra{\vec b} = r(T).
\end{align*}
In the second step and third steps, we used the notation~$\mathfrak b$ and $T^\perp$ from \cref{rem:bilinear form}, respectively, as well as that $\dim T=t$.
The last step holds since $T=T^\perp$, as $T$ is a Lagrangian subspace.

Next, we consider the phase gate, which is likewise a single-qudit gate.
For $d=2$, we have that
\begin{align*}
  P^{\ot t} r(T) P^{\dagger,\ot t}
= \sum_{(\vec x, \vec y)\in T} i^{\vec x\cdot\vec x-\vec y\cdot\vec y} \ket{\vec x}\bra{\vec y}
= r(T)
\end{align*}
since $T$ is totally isotropic.
For odd $d$, we instead compute
\begin{align*}
  P^{\ot t} r(T) P^{\dagger,\ot t}
= \sum_{(\vec x, \vec y)\in T} \omega^{2^{-1} \sum_j x_j(x_j-1) - y_j(y_j-1)} \ket{\vec x}\bra{\vec y}
= \sum_{\vec v=(\vec x, \vec y)\in T} \omega^{2^{-1} \mathfrak b(\vec v,\vec v-\vec 1_{2t})} \ket{\vec x}\bra{\vec y}
= r(T),
\end{align*}
since $T$ is totally isotropic and stochastic (so $\vec w=\vec v-\vec 1_{2t}\in T$ and $\mathfrak b(\vec v,\vec w)=0$ for every $\vec v\in T$).

Lastly, we consider the controlled addition gate, which is a two-qudit gate:
\begin{align*}
&\quad \operatorname{CADD}^{\ot t} r(T)^{\ot 2} \operatorname{CADD}^{\dagger,\ot t}
= \sum_{(\vec x,\vec y)\in T} \sum_{(\vec x',\vec y')\in T} \operatorname{CADD}^{\ot t} \ket{\vec x,\vec x'}\bra{\vec y,\vec y'} \operatorname{CADD}^{\dagger,\ot t} \\
&= \sum_{(\vec x,\vec y)\in T} \sum_{(\vec x',\vec y')\in T} \ket{\vec x,\vec x+\vec x'}\bra{\vec y,\vec y+\vec y'}
= \sum_{(\vec x,\vec y)\in T} \sum_{(\vec x',\vec y')\in T} \ket{\vec x,\vec x'}\bra{\vec y,\vec y'}
\end{align*}
where we only used that $T$ is a subspace.
\end{proof}

We now show that the operators $R(T)$ are linearly independent as soon as $n\geq t-1$.
For this, we introduce the following useful notation:

\begin{dfn}[Vectorization]\label{dfn:vectorization}
The \emph{vectorization} operator $\vecmap$ is defined by its action in the computational basis via
\begin{equation*}
  \vecmap (\ket{\vec x} \bra{\vec y}) = \ket{\vec x} \otimes \ket{\vec y} = \ket{\vec x,\vec y}.
\end{equation*}
\end{dfn}

\begin{lem}\label{lem:indep}
  If $n\geq t-1$ then operators $R(T)$ are linearly independent.
\end{lem}
\begin{proof}
  For each $T\in\Sigma_{t,t}(d)$, consider the vectorization of $r(T)$, which we denote by $\ket{T}\coloneqq \vecmap{\left(r(T)\right)} = \sum_{\vec v \in T} \ket{\vec  v} \in (\CC^d)^{\ot 2t}$.
  Note that $\braket{\vec v | T} = \delta_{\vec v \in T}$.
  Clearly, $\vecmap{\left( R(T)\right)} = \vecmap{\left(r(T) \right)}^{\ot n} = \ket T ^{\ot n}$.
  Therefore, we want to show that the vectors $\ket T^{\ot n}$ are linearly independent as soon as $n\geq t-1$.
  But each $T$ is $t$-dimensional and contains the vector $\vec 1_{2t}$.
  Extend it by $\vec v_1,\dots,\vec v_{t-1}$ to a basis of $T$.
  Then, if $T'$ is another subspace:
  \begin{align*}
  \bra{\vec v_1}\dots\bra{\vec v_{t-1}}\bra{0}^{\ot n-(t-1)}\ket{T'}^{\ot n}
  = \bra{\vec v_1}\dots\bra{\vec v_{t-1}}\ket{T'}^{\ot (t-1)}
  = \delta_{\vec v_1,\dots,\vec v_{t-1}\in T'}
  = \delta_{T,T'}
  \end{align*}
  This concludes the proof.
\end{proof}

So far, we have accomplished the task of finding a large set of linearly independent operators in the commutant of $\Cliff(n,d)^{\ot t}$, one for each element of~$\Sigma_{t,t}(d)$.
In the remainder of this section we will compute the dimension of the commutant as well as the cardinality of $\Sigma_{t,t}(d)$, and show that the two numbers agree precisely.
We will use the \emph{Gaussian binomial coefficients}, which are defined by
\begin{align*}
  \binom n k_d = \frac{[n]_d [n-1]_d \cdots [n-k+1]_d}{[k]_d [k-1]_d \cdots [1]_d},
  \quad
  \text{ where }
  [k]_d = \sum_{i=0}^{k-1} d^i,
\end{align*}
It is well-known that $\binom n k_d$ equals the number of $k$-dimensional subspaces in $\ZZ_d^n$.
The Gaussian binomial coefficients satisfy the following analogs of Pascal's rule,
\begin{align}\label{eq:q binomial pascal}
  \binom n k_d = d^k \binom{n-1}k_d + \binom{n-1}{k-1}_d,
\end{align}
and of the binomial formula,
\begin{align}\label{eq:q binomial formula}
  \sum_{k=0}^n d^{k(k-1)/2} \binom n k_d t^k = \prod_{k=0}^{n-1} \left(d^k t + 1\right).
\end{align}

We now compute the dimension of the commutant.
This has previously been done for $t\leq 4$ by Zhu~\cite{zhu2015multiqubit} and before that for $d=2, n=1$ by van~den~Nest et al~\cite{van2005invariants}.

We start with the following result from~\cite{zhu2015multiqubit}, which reduces the dimension computation to a counting problem.
Zhu arrived at this result by computing the \emph{frame potential} of the Clifford group~--~essentially, the norm squared of the character of the representation $U\mapsto U^{\otimes k}$.
In contrast, we will follow the approach by van~den~Nest et al, who considered the action of the Clifford average (also known as the \emph{twirl operation} or \emph{Reynolds operator}) on the basis of Weyl operators, correcting a glitch in~\cite{van2005invariants} along the way.%
\footnote{In Ref.~\cite{van2005invariants}, the relative phases~$\omega^{f_\Gamma(\vec x)}$ that appear in our \cref{eq:true signs} are all taken to be trivial (for qubits).
The origin seems to lie in their Section~II, where it is stated---in their language---that $\alpha_1\alpha_2\alpha_3=1$.
But this holds only if their~$\pi$ is cyclic.
Clifford operations inducing non-cyclic permutations do, however, exist.
We thank Huangjun Zhu for identifying the root of the apparent contradiction.}

\begin{lem}[\cite{zhu2015multiqubit,van2005invariants}]
The dimension of the commutant of $\Cliff(n,d)^{\ot t}$ is equal to the number of orbits for the diagonal action of the symplectic group~$\Sp(2n,d)$ on $t-1$ copies of the phase space~$\ZZ_d^{2n}$, i.e., for the action
\begin{align}\label{eq:commutant action}
  \Gamma \cdot (\vec x_1,\dots,\vec x_{t-1})
  =
  (\Gamma\vec x_1,\dots,\Gamma\vec x_{t-1}),
\end{align}
where $\Gamma\in\Sp(2n,d)$ and $(\vec x_1,\dots,\vec x_{t-1})\in(\ZZ_d^{2n})^{t-1}$.
\end{lem}
\begin{proof}
We will show that the dimension of the commutant is equal to the number of orbits for the diagonal action of $\Sp(2n,d)$ on
\begin{align*}
  \mathcal W_t \coloneqq \{ (\vec x_1,\dots,\vec x_t) : \sum_{i=1}^t \vec x_i = 0 \},
\end{align*}
which is plainly an equivalent statement.

We start by noting that the Weyl operators $W_{\vec x}$  for $\vec x = (\vec x_1,\dots,\vec x_t)\in(\ZZ_d^{2n})^{t}$ form a basis of the space of operators on~$((\CC^d)^{\ot n})^{\ot t}$.
We can thus obtain a generating set of the commutant by averaging each Weyl operator $W_{\vec x}$ with respect to the tensor power action of the Clifford group.
According to \cref{lem:cliff basics}, we can for each symplectic matrix $\Gamma$ fix a Clifford unitary $U_\Gamma$ such that the set of $\{U_\Gamma W_{\vec b}\}$ equals the Clifford group, up to global phases.
Let us denote by $f_\Gamma$ the phase function corresponding to $U_\Gamma$, as in \cref{eq:cliff action prp}.
Thus, the average of the Weyl operator~$W_{\vec x}$ is, up to overall normalization, given by
\begin{align*}
  \Lambda_{\Cliff}(W_{\vec  x})
&\coloneqq d^{-2n} \sum_{\Gamma\in\Sp(2n,d)} \sum_{\vec b\in\ZZ_d^{2n}} \left( U_\Gamma W_{\vec b} \right)^{\ot t} W_{\vec x} \left( U_\Gamma W_{\vec b} \right)^{\dagger,\ot t}
% \\&= d^{-2n} \sum_{\Gamma\in\Sp(2n,d)} \sum_{\vec b\in\ZZ_d^{2n}} U_\Gamma^{\ot t} W_{\vec b}^{\ot t} W_{\vec x} W_{\vec b}^{\dagger,\ot t} U_\Gamma^{\dagger,\ot t}
\\&= d^{-2n} \sum_{\Gamma\in\Sp(2n,d)} \sum_{\vec b\in\ZZ_d^{2n}} \omega^{[\vec b,\vec x_1+\dots+\vec x_t]} U_\Gamma^{\ot t} W_{\vec x} U_\Gamma^{\dagger,\ot t}
% \\&= \delta_{\vec x_1+\dots+\vec x_t,0} \sum_{\Gamma\in\Sp(2n,d)} \omega^{f_\Gamma(\vec x)} W_{\Gamma\cdot\vec x}
\\&= \delta_{\vec x \in \mathcal W_t} \sum_{\Gamma\in\Sp(2n,d)} \omega^{f_\Gamma(\vec x)} W_{\Gamma\vec x},
\end{align*}
where $f_\Gamma(\vec x) = \sum_{i=1}^t f_\Gamma(\vec x_i)$ and $\Gamma\vec x\coloneqq(\Gamma \vec x_1, \dots, \Gamma \vec x_t)$.

When $d$ is odd, the phase function $f$ can be chosen to vanish (\cref{lem:cliff basics}).
Thus, the averaged operator is equal to the sum of Weyl operators over the $\Sp(2n,d)$-orbit of~$\vec x$, provided $\vec x\in\mathcal W_t$, and zero otherwise.
Since distinct orbits are disjoint, it is clear that we obtain a basis of the commutant by averaging one Weyl operator for each orbit of the diagonal action of~$\Sp(2n,d)$ on~$\mathcal W_t$.

Now consider the case where~$d=2$.
To each $\vec x\in\mathcal W_t$, associate the phase~$\phi_{\vec x}$ (a power of $\tau$) such that
\begin{align*}
  W_{\vec x_1} \cdots W_{\vec x_t} = \phi_{\vec x}\, I.
\end{align*}
Then, for each $\Gamma\in\Sp(2n,d)$,
\begin{align*}
 \phi_{\vec x} \, I
= U_\Gamma W_{\vec x_1} \cdots W_{\vec x_t} U_\Gamma^\dagger
= (U_{\Gamma} W_{\vec x_1} U_{\Gamma}^\dagger) \cdots (U_{\Gamma} W_{\vec x_t} U_{\Gamma}^\dagger)
= \omega^{f_\Gamma(\vec x)} W_{\Gamma \vec x_1} \cdots W_{\Gamma\vec x_t}
= \omega^{f_\Gamma(\vec x)} \phi_{\Gamma \vec x}.
\end{align*}
It follows that the phase function $f_\Gamma(\vec x)$ depends only on $\vec x$ and $\Gamma \vec x$ (rather than directly on $\Gamma$) and is given explicitly by the quotient
\begin{align}\label{eq:true signs}
  \omega^{f_\Gamma(\vec x)} = \frac{\phi_{\vec x}}{\phi_{\Gamma  \vec x}}.
\end{align}
Thus, for~$\vec x \in \mathcal W_t$,
\begin{align*}
  \Lambda_{\Cliff}(W_{\vec x})
  = \sum_{\Gamma\in\Sp(2n,d)} \frac{\phi_{\vec x}}{\phi_{\Gamma \vec x}} W_{\Gamma\vec x}
  = \phi_{\vec x} \sum_{\Gamma\in\Sp(2n,d)} \frac{W_{\Gamma\vec x}}{\phi_{\Gamma \vec x}}.
\end{align*}
In particular, if $\vec y$ is in the same $\Sp(2n,d)$-orbit as $\vec x$ then $\Lambda_{\Cliff}(W_{\vec y}) =
\frac{\phi_{\vec y}}{\phi_{\vec x}} \Lambda_{\Cliff}(W_{\vec x})$, i.e., the two averaged operators only differ by a global phase.
Thus, also for $d=2$ we obtain a basis of the commutant by averaging one Weyl operator for each orbit of the diagonal action of~$\Sp(2n,d)$ on~$\mathcal W_t$.
\end{proof}

We now derive an explicit formula for the dimension of the commutant.

\begin{thm}[Dimension of commutant]\label{thm:commutant size}
  Let $n\geq t-1$.
  Then the dimension of the commutant of $\Cliff(n,d)^{\ot t}$ is equal to $\prod_{k=0}^{t-2} (d^k+1)$.
\end{thm}
\begin{proof}
To count the number of orbits of the action~\eqref{eq:commutant action}, we will associate to any orbit~$O$ an invariant, the \emph{dimension}, defined by $\dim(O) = \dim \Span~\{ \vec x_1,\dots,\vec x_{t-1} \}$,
where $(\vec x_1,\dots,\vec x_{t-1})$ is any point in the orbit.
We write $\Omega_t$ for the set of all orbits and $\Omega_t^\ell$ for the set of orbits with dimension~$\ell$.
We will establish and solve the following recursion relation:
\begin{align}\label{eq:orbit recursion relation}
  |\Omega_t^\ell| = |\Omega_{t-1}^\ell| d^\ell + |\Omega_{t-1}^{\ell-1}| d^{\ell-1}
\end{align}
To see why this is true, suppose $(\vec x_1,\dots,\vec x_{t-1})\in\Omega_t^\ell$.
Then there are two cases:
\begin{enumerate}
\item $x_{t-1} \in \Span~\{\vec x_1,\dots,\vec x_{t-2}\}$:
Then the orbit through $(\vec x_1,\dots,\vec x_{t-2})$ is in $\Omega_{t-1}^\ell$, and there are $d^\ell$ ways to choose $\vec x_{t-1}\in\Span\{\vec x_1,\dots,\vec x_{t-2}\}$.
Together, this contributes $|\Omega_{t-1}^\ell| d^\ell$ many orbits to~$\Omega_t^\ell$.
\item $x_{t-1} \not\in \Span~\{\vec x_1,\dots,\vec x_{t-2}\}$:
Then the orbit through $(\vec x_1,\dots,\vec x_{t-2})$ is in $\Omega_{t-1}^{\ell-1}$, and we have to count the number of ways that we can add a new vector $\vec x_{t-1}$ to $\Span~\{\vec x_1,\dots,\vec x_{t-2}\}$ such that we get different orbits.
By Witt's theorem, which also holds for alternating forms in characteristic two~\cite{wilson2009finite}, the only invariants are the inner products between $\vec x_{t-1}$ and a basis of $\Span~\{\vec x_1,\dots,\vec x_{t-2}\}$.
By assumption, the latter space has dimension $\ell-1\leq t-2<n$, so we have $d^{\ell-1}$ options for $\vec x_{t-1}$.
Together, this contributes $|\Omega_{t-1}^{\ell-1}| d^{\ell-1}$ many orbits to $\Omega_t^\ell$.
\end{enumerate}
We have thus established the recursion relation~\eqref{eq:orbit recursion relation}.
Since $\Omega_{2,0} = \{ \{0\} \}$ and $\Omega_{2,1} = \{ \vec x_1 \neq 0 \}$, we find the initial conditions $|\Omega_{2,0}|=|\Omega_{2,1}|=1$.
The solution to the recursion relation is
\begin{align*}
  |\Omega_t^\ell| = d^{\ell(\ell-1)/2} \binom{t-1}\ell_d,
\end{align*}
as can be verified by using Pascal's rule~\eqref{eq:q binomial pascal}.
% \begin{align*}
%   d^{h(h-1)/2} \binom{t-2}h_d d^h + d^{(h-2)(h-1)/2} \binom{t-2}{h-1}_d d^{h-1},
% = d^{h(h-1)/2} \binom{t-2}h_d d^h + d^{h(h-1)/2} \binom{t-2}{h-1}_d
% = d^{h(h-1)/2} \binom{t-1}h_d.
% \end{align*}
% \begin{align*}
%   d^{0(0-1)/2} \binom{t-1}0_d = 1
% \end{align*}
% \begin{align*}
%   d^{1(1-1)/2} \binom{t-1}1_d = 1
% \end{align*}
Using the binomial formula~\eqref{eq:q binomial formula}, we conclude that
\begin{equation}\label{eq:use q binomial formula}
  |\Omega_t|
= \sum_{\ell=0}^{t-1} |\Omega_t^\ell|
= \sum_{\ell=0}^{t-1} d^{\ell(\ell-1)/2} \binom{t-1}\ell_d
= \prod_{k=0}^{t-2} (d^k+1).
\end{equation}
This establishes the desired formula for the dimension of the commutant.
\end{proof}

Next, we count the number of stochastic Lagrangian subspaces.
To this end, define the ``diagonal subspace''
\begin{align*}
  \Delta = \{ (\vec x, \vec x) \,|\, \vec x\in \ZZ_d^t \}\subset \ZZ_d^{2t}.
\end{align*}

\begin{thm}[Cardinality of $\Sigma_{t,t}$]\label{thm:sigma size}
  We have $\left\lvert\Sigma_{t,t}(d)\right\rvert=\prod_{k=0}^{t-2} (d^k+1)$.
\end{thm}
\begin{proof}
  Let $\Sigma_{t,t}^\ell(d)$ denote the set of subspaces~$T\in\Sigma_{t,t}(d)$ such that $\dim(T\cap\Delta)=t-\ell$.
  We will show that
  \begin{align}\label{eq:sigma ell claim}
    \left\lvert\Sigma_{t,t}^\ell(d)\right\rvert
  = d^{\ell(\ell-1)/2} \binom{t-1}{\ell}_d,
  \end{align}
  which implies the claim by the same calculation as in \cref{eq:use q binomial formula}.
  To start, consider a subspace $T\in\Sigma_{t,t}^\ell(d)$ and consider
  \begin{align*}
    T_\Delta \coloneqq T \cap \Delta = \{ (\vec x,\vec x) : \vec x \in X \},
  \end{align*}
  with $X$ a $(t-\ell)$-dimensional subspace that is uniquely determined by $T$.
  Since $T$ is stochastic, we know that $\vec 1_t\in X$.
  Fix a basis $\vec x_1,\dots,\vec x_{t-\ell}$ of~$X$ and extend it by vectors $\vec z_1,\dots,\vec z_\ell$ to a basis of~$\ZZ_d^t$.
  Denote the dual basis with respect to the ordinary dot product by $\hat{\vec x}_1,\dots,\hat{\vec x}_{t-\ell}$, $\hat{\vec z}_1,\dots,\hat{\vec z}_\ell$.
  Now, any vector in $\ZZ_d^{2t}$, so particularly in~$T$, can be written uniquely in the form $(\vec a+\vec b,\vec b)$.
  The subspace of vectors where $\vec b$ is a linear combination of $\vec z_1,\dots,\vec z_\ell$ forms a complement of $T_\Delta \subseteq T$, which we shall denote by $T_N$.
  Since $T_N \cap \Delta = \{0\}$, we know that $\vec a\neq0$ for any nonzero vector in~$T_N$.
  The condition that $T$ is self-orthogonal implies that $\vec a \cdot \vec x_i = 0$ for all $i=1,\dots,t-\ell$, so that $\vec a$ is a linear combination of $\hat{\vec z}_1,\dots,\hat{\vec z}_\ell$.
  Since also $\dim T_N=\ell$, this implies that $T_N$ has a \emph{unique} basis of the form $(\hat{\vec z}_1+\vec w_1,\vec w_1)$, \dots, $(\hat{\vec z}_\ell+\vec w_\ell,\vec w_\ell)$, where each $\vec w_i$ is of the form $\vec w_i = \sum_{j=1}^\ell A_{ij} \vec z_j$.
  We still need to implement the condition that $T_N$ is self-orthogonal.
  In terms of the matrix $A=(A_{ij})$, this means that
  \begin{align}\label{eq:constraints on A}
    0 = (\hat{\vec z}_i+\vec w_i) \cdot (\hat{\vec z}_j+\vec w_j) - \vec w_i \cdot \vec w_j = \hat{\vec z}_i\cdot\hat{\vec z}_j + A_{ij} + A_{ji} \pmod d
  \end{align}
  for any $i,j$.
  This means that the lower triangular part of~$A$ is uniquely determined by the upper triangular part.

  For $d\neq2$, \eqref{eq:constraints on A} furthermore implies that the diagonal entries of~$A$ are fixed, so there are in total $d^{\ell(\ell-1)/2}$ many options for~$A$.
  We have thus implemented all conditions for~$T$ to be a subspace in~$\Sigma_{t,t}^\ell(d)$ since, according to \cref{rem:bilinear form}, for $d\neq2$, any self-orthogonal $T$ is automatically totally isotropic.
  The set of $(t-\ell)$-dimensional subspaces of~$\ZZ_d^t$ that contain~$\vec 1_t$ are in bijection with the $(t-\ell-1)$-dimensional subspaces in~$\ZZ_d^t/\ZZ_d \vec 1_t$, hence there are $\binom{t-1}{t-\ell-1}_d=\binom{t-1}\ell_d$ many choices for $X$.
  Together, we obtain~\eqref{eq:sigma ell claim}.

  For $d=2$, \eqref{eq:constraints on A} gives no constraint about the diagonal entries of~$A$.
  Instead, it asserts that $\hat{\vec z}_i \cdot \hat{\vec z}_i = 0$ or, equivalently, that $\hat{\vec z}_i \cdot \vec 1_t = 0$ for $i=1,\dots,\ell$, which is automatically satisfied since $\vec 1_t\in X$.
  We will now show that there is a unique choice for the diagonal entries of~$A$ such that $T$ is totally isotropic with respect to the $\ZZ_4$-valued quadratic form~$\mathfrak q$.
  By the discussion in \cref{rem:bilinear form}, since $T$ is self-orthogonal, it suffices to consider $T_\Delta$ and its complement $T_N$ separately.
  But the vectors in $T_\Delta$ are automatically isotropic, while for $T_N$ total isotropy amounts to the condition that
  \begin{align*}
    0 = (\hat{\vec z}_i+\vec w_i) \cdot (\hat{\vec z}_i+\vec w_i) - \vec w_i \cdot \vec w_i = \hat{\vec z}_i\cdot\hat{\vec z}_i + 2 A_{ii} \pmod 4,
  \end{align*}
  which fixes the~$A_{ii}$ uniquely.
  We thus obtain~\eqref{eq:sigma ell claim} by the same counting as above.
\end{proof}

% \begin{rem}
% For odd $d$, we can also appeal to known results in the literature~\cite{wilson2009finite}.
% Since $\vec 1_{2t}$ is self-orthogonal with respect to the form~$\mathfrak b$ defined in \cref{rem:bilinear form}, we may instead count the number of Lagrangian subspaces in the quotient space $(\ZZ_d {\vec 1}_{2t})^\perp / \ZZ_d \vec 1_{2t}$.
% This space is of `plus type' (i.e., has a maximally isotropic subspace with dimension equal to half of the dimension of the space), because the original space was of plus type and so the advertised formula follows from the known cardinality of Lagrangian subspaces for orthogonal forms over finite fields.
% \end{rem}

We finally obtain \cref{thm:commutant} as a consequence of the preceding results.

\begin{proof}[Proof of \cref{thm:commutant}]
By combining \cref{lem:commutant,thm:sigma size,thm:commutant size}, we see that the operators $R(T)$ form a basis of the commutant of $\Cliff(n,d)^{\ot t}$ on $((\CC^d)^{\ot n})^{\ot t}$.
\end{proof}

It is interesting to note that all elements $R(T)$ of our basis of the commutant of $\Cliff(n,d)^{\ot s}$ have the property that $\braket{S^{\ot t} | R(T) | S^{\ot t}}=1$ for every stabilizer state~$\ket S$.
Indeed, if $T\in \Sigma_{t,t}(d)$ and $\ket S=U \ket 0^{\ot n}$ for some Clifford unitary $U$, then
\begin{align}\label{eq:not quite eigenvectors}
  \braket{S^{\ot t} | R(T) | S^{\ot t}}=\braket{S^{\ot t} | R(T) U^{\ot t} | 0^{\ot tn}}
  = \braket{S^{\ot t} | U^{\ot t} R(T) | 0^{\ot tn}}
  = \braket{0^{\ot tn} | R(T) | 0^{\ot tn}} = 1,
\end{align}
where we used that $\vec 0\in T$ (see also \cref{eq:eigenvectors} below).

%-----------------------------------------------------------------------------
\subsection{Structure of the commutant}
%-----------------------------------------------------------------------------
\Cref{thm:commutant} is in the spirit of Schur-Weyl duality in that it establishes a natural basis of the commutant of the tensor power action of the Clifford group (a subgroup of the unitary group), generalizing the permutation operators.
Yet, in contrast to the permutation group, $\Sigma_{t,t}(d)$ is \emph{not} in general a group and the operators $R(T)$ for $T\in\Sigma_{t,t}(d)$ are not always invertible.
In this section we show that $\Sigma_{t,t}(d)$ has a rich algebraic structure.

We first observe that there is a maximal subset of $\Sigma_{t,t}(d)$ that carries a group structure such that the $R(T)$ form a (unitary) representation.
The following definition and lemma identify these elements:

\restateDfnOt

\noindent
To see that $O_t(d)$ forms a group we only need to observe that $O^{-1} = O^T$ is again in $O_t(d)$.
The following remark is completely analogous to \cref{rem:bilinear form}.

\begin{rem}\label{rem:O_t}
Recall that a linear map is an \emph{isometry} with respect to a quadratic form $q$ if $q(O\vec x)=q(\vec x)$ for all $\vec x\in\ZZ_d^t$.
This justifies our terminology in \cref{dfn:O_t}.
As before, we note that $q$ is a $\ZZ_D$-valued quadratic form associated to the $\ZZ_d$-bilinear form $\vec x\cdot\vec y$ in the sense of~\cite{wood1993witt}, namely,
\begin{align}\label{eq:bilinear vs quadratic form O}
  q(\vec x + \vec y) = q(\vec x) + q(\vec y) + 2 \vec x \cdot \vec y \pmod D.
\end{align}
In particular, any $O\in O_t(d)$ is an orthogonal matrix in the ordinary sense that $O^T O = I \pmod d$, i.e., $O\vec x \cdot O\vec y=\vec x\cdot\vec y$ for all $\vec x, \vec y\in\ZZ_d^t$.
If $d$ is odd then $q(\vec x)=\vec x\cdot \vec x$, so any orthogonal matrix is automatically a $q$-isometry.

If $d=2$ then \cref{eq:bilinear vs quadratic form O} implies that an orthogonal matrix~$O$ is a $q$-isometry provided that $q(\vec x)=1\pmod4$ for every column of~$O$ or, equivalently, for every row of~$O$ (since $O^T = O^{-1}$).
In particular, any $q$-isometry is automatically stochastic.
\end{rem}

\noindent
The significance of \cref{dfn:O_t} is the following observation.

\begin{lem}\label{lem:O_t maximal subgroup}
For every $O\in O_t(d)$, the subspace
\begin{align*}
  T_O\coloneqq \{(O\vec x,\vec x)\,:\, \vec x \in \ZZ_d^t\}
\end{align*}
is an element of $\Sigma_{t,t}(d)$ and the operators
\begin{align}\label{eq:O_t action}
  r(O) \coloneqq r(T_O) = \sum_{\vec x} \ket{O\vec x}\bra{\vec x}, \quad
  R(O) \coloneqq r(O)^{\ot n} = R(T_O)
\end{align}
are unitary.
Conversely, if $R(T)$ is invertible then $T=T_O$ for some $O\in O_t(d)$.
Moreover, the operators $R(O)$ define a unitary representation of $O_t(d)$ on $(\CC^d)^{\ot tn}$.
\end{lem}
\begin{proof}
Only the converse needs justification.
Note that in order for $R(T)$ to be invertible, both subspaces $\{ \vec x : (\vec x,\vec y)\in T \}$ and $\{ \vec y : (\vec x,\vec y)\in T\}$ of $\ZZ_d^t$ should be $t$-dimensional (corresponding to $r(T)$ having full row and column rank).
The claim now follows easily.
\end{proof}

\noindent
We will often regard $O_t(d)$ as a subset of $\Sigma_{t,t}(d)$ via the assignment $O \mapsto T_O$.
Note that any permutation matrix satisfies the conditions of \cref{dfn:O_t}, so we can consider $S_t$ as a subgroup of $O_t(d)$ for every value of~$d$, and hence as a subset of $\Sigma_{t,t}(d)$.

\begin{rem}
The Clifford group is a $t$-design (for $n\geq t-1$) if and only if $S_t = O_t(d) = \Sigma_{t,t}(d)$, i.e., if and only if
\begin{align*}
  t! = |S_t| = |\Sigma_{t,t}(d)| = \prod_{k=0}^{t-2} (d^k+1).
\end{align*}
This identity always holds up to $t=2$, and up to $t=3$ precisely in the case of qubits ($d=2$).
Thus the multiqubit Clifford group is a 3-design (but not a 4-design), while in higher dimensions the Clifford group is only a 2-design (but not a 3-design), reproducing prior beautiful results~\cite{zhu2015multiqubit,webb2016clifford}.
\end{rem}

For $O\in O_t(d)$, \cref{eq:not quite eigenvectors} implies that
\begin{align}\label{eq:eigenvectors}
  R(O) \ket S^{\ot t} = \ket S^{\ot t}
\end{align}
for every stabilizer state $\ket S$.
That is, stochastic isometries stabilize any stabilizer tensor power.
We will return to discussing the implications of this important fact in \cref{sec:moments} below.

Next, we note that the group $O_t(d)$ naturally acts on the elements of $\Sigma_{t,t}(d)$ from left and right, suggesting that it is the natural symmetry group of $\Sigma_{t,t}(d)$.

\begin{dfn}[Left and right action on subspaces]\label{dfn:O on Sigma}
Consider a subspace $T\in \Sigma_{t,t}(d)$ and a matrix $O \in O_t(d)$.
We define the \emph{left action} of $O$ on $T$ as follows:
\begin{align}\label{eq:left action}
  OT=\{(O \vec x,\vec y)\, : \, (\vec x,\vec y) \in T\}.
\end{align}
Similarly, the \emph{right action} of $O$ on $T$ is defined as:
\begin{align}\label{eq:right action}
  TO=\{( \vec x,O^T \vec y)\, : \, (\vec x,\vec y) \in T\}.
\end{align}
It is easy to check that $OT, TO \in \Sigma_{t,t}(d)$.
\end{dfn}
Note that this action is consistent with the composition of the operators $R(T)$ and $R(O)$:
For all $T\in\Sigma_{t,t}(d)$ and $O,O'\in O_t(d)$ we have that
\[ R(O) R(T) R(O')=R(OTO'). \]
We can therefore decompose $\Sigma_{t,t}(d)$ into a disjoint union of \emph{double cosets} with respect to the left and right action:
\begin{equation}\label{eq:double cosets}
\Sigma_{t,t}(d) = O_t(d) T_1 O_t(d) \cup \dots \cup O_t(d) T_k O_t(d),
\end{equation}
where $T_1,\dots,T_k$ are choices of subspaces in $\Sigma_{t,t}(d)$ that represent the different cosets.
We note that $O_t(d)$ is always one of the double cosets in \cref{eq:double cosets}, corresponding to, e.g., the choice $T_1=\Delta$.

\bigskip

We will now derive a complete classification of the double cosets.
We start with the central definition.
Recall the quadratic form $q\colon \ZZ_d^t \to \ZZ_D$, $q(\vec x) \coloneqq \vec x\cdot\vec x$ from \cref{dfn:O_t}.

\begin{dfn}[Defect subspaces]\label{dfn:defect subspaces}
A \emph{defect subspace} is a subspace $N\subseteq\ZZ_d^t$ with the following properties:
\begin{enumerate}
\item $N$ is \emph{totally $q$-isotropic}: i.e., $q(\vec x)=0\pmod D$ for all $\vec x\in N$.
\item $N$ is \emph{co-stochastic}: $\vec 1_t\in N^\perp$, i.e., $\vec x\cdot\vec 1_t=0\pmod d$ for every $\vec x\in N$.
\end{enumerate}
The quotient $N^\perp/N$ inherits a $\ZZ_D$-valued quadratic form, which we also denote by~$q([\vec y]) \coloneqq q(\vec y)$.

Given two defect subspaces $N$ and $M$, we write $\Iso(N,M)$ for the set of \emph{defect isomorphisms}, by which we mean invertible linear maps~$J \colon N^\perp/N \to M^\perp/M$ with the following two properties:
\begin{enumerate}
\item $J$ is a \emph{$q$-isometry}: i.e., $q(J[\vec y]) = q([\vec y])$ for all $[\vec y]\in N^\perp/N$.
\item $J$ is \emph{stochastic}: $J[\vec 1_t] = [\vec 1_t]$.
\end{enumerate}
The inverse of $J$ is again a map in $\Iso(M,N)$.
\end{dfn}

This definition is central as it allows us to construct elements in $\Sigma_{t,t}(d)$.
If $N$, $M$ are defect subspaces and $J\in\Iso(M,N)$, then
\begin{equation}\label{eq:T from defect spaces}
\begin{aligned}
  T &= \{ (\vec x + \vec z, \vec y + \vec w) : [\vec y]\in M^\perp/M, \, [\vec x]=J[\vec y], \, \vec z\in N, \, \vec w\in M \} \\
  &= \{ (\vec x, \vec y) : \vec y\in M^\perp, \, \vec x \in J[\vec y] \}
\end{aligned}
\end{equation}
is an element in $\Sigma_{t,t}(d)$.
Note that, necessarily, $\dim N=\dim M$ (since $J$ is invertible) and $\vec 1_t\in N$ if and only if $\vec 1_t\in M$ (since $J$ is also stochastic).
We now show that \emph{all} elements of $\Sigma_{t,t}(d)$ can be obtained in this way.

\begin{prp}\label{prp:defect subspaces}
Let $T\in\Sigma_{t,t}(d)$.
\begin{enumerate}
  \item The subspaces $T_{LD} \coloneqq \{\vec x : (\vec x,0)\in T\}$ and $T_{RD} \coloneqq \{\vec y : (0,\vec y)\in T\}$ are defect subspaces.
  We call them the \emph{left} and \emph{right defect subspaces} of $T$, respectively.
  \item $\dim T_{LD}=\dim T_{RD}$ and $\vec 1_t\in T_{LD}$ if and only if $\vec 1_t\in T_{RD}$.
  \item $T_{LD}^\perp = T_L \coloneqq \{ \vec x : (\vec x,\vec y) \in T\}$ and $T_{RD}^\perp = T_R \coloneqq \{ \vec y : (\vec x,\vec y) \in T\}$,
  \item For every $\vec y\in T_{RD}^\perp$, choose some~$\vec x(\vec y)$ such that $(\vec x(\vec y), \vec y)\in T$.
  Then $T_J\colon [\vec y]\mapsto[\vec x(\vec y)]$ is a well-defined defect isomorphism, i.e., an element in $\Iso(T_{RD}, T_{LD})$.
  \item The data $(T_{LD}, T_{RD}, T_J)$ is uniquely determined by $T$.
  \item $T$ is of the form~\eqref{eq:T from defect spaces}, with $T_{LD}=N$, $T_{RD}=M$, and $T_J=J$.
\end{enumerate}
\end{prp}
\begin{proof}
The first claim is clear from \cref{dfn:lagrangian stochastic}.
Since $T$ is stochastic, so $(\vec 1_t,0)\in T$ if and only if $(0,\vec 1_t) = \vec 1_{2t}-(\vec 1_t,0)\in T$, which proves half of the second claim.
Next, consider the maps
\begin{align*}
  \pi_L\colon T\to\ZZ_d^t, \; (\vec x, \vec y)\mapsto\vec x \quad\text{and}\quad
  \pi_R\colon T\to\ZZ_d^t, \; (\vec x, \vec y)\mapsto\vec y.
\end{align*}
Then $T_{LD}\cong\ker\pi_R$ and $T_{RD}\cong\ker\pi_L$, while $T_L=\ran\pi_L$ and $T_R=\ran\pi_R$.
Note that $T_L\subseteq T_{LD}^\perp$ and $T_R\subseteq T_{RD}^\perp$, since $T$ is totally isotropic.
Using the rank-nullity theorem,
\begin{align*}
  \dim T_{LD}^\perp &= \dim T - \dim T_{LD} = \dim T/\ker\pi_R = \dim\ran\pi_R = \dim T_R \leq \dim T_{RD}^\perp, \\
  \dim T_{RD}^\perp &= \dim T - \dim T_{RD} = \dim T/\ker\pi_L = \dim\ran\pi_L = \dim T_L \leq \dim T_{LD}^\perp.
\end{align*}
Adding the two inequalities we see that they must both be equalities, hence $\dim T_{LD} = \dim T_{RD}$ as well as $T_L=T_{LD}^\perp$, $T_R=T_{RD}^\perp$.
This establishes the second and third claim.

For the fourth claim, first recall from above that $T_{LD}^\perp=T_L$ and $T_{RD}^\perp=T_R$, which means that for any $\vec y\in T_{RD}^\perp$ there exists some $\vec x\in T_{LD}^\perp$ such that $(\vec x,\vec y)\in T$.
Next, suppose that $(\vec x,\vec y),(\vec x',\vec y')\in T$ such that $\vec y-\vec y'\in T_{RD}$.
Then $(0,\vec y-\vec y')\in T$, so
\begin{align*}
  (\vec x-\vec x',0) = (\vec x,\vec y) - (0,\vec y-\vec y') - (\vec x',\vec y') \in T,
\end{align*}
which means that $\vec x-\vec x'\in T_{LD}$.
As a consequence, $[\vec y]\mapsto[\vec x(\vec y)]$ is well-defined as a map from $T_{RD}^\perp/T_{RD}$ to $T_{LD}^\perp/T_{LD}$.
Using \cref{dfn:lagrangian stochastic}, it is not hard to see that it defines an element of~$\Iso(T_{RD}, T_{LD})$.
% Its inverse is given by $[\vec y]\mapsto[\vec x(\vec y)]$, which is likewise well-defined and an element of $\Iso(T_{RD},T_{LD})$.

The fifth claim is clear by construction of $T_{LD}$ and $T_{RD}$ and from the fact that $[\vec y]\mapsto[\vec x(\vec y)]$ is well-defined.
And the last claim can be seen to hold since the right-hand side of~\eqref{eq:T from defect spaces} is clearly a subset of~$T$ for our choice of $N$, $M$, and $J$, but also of dimension~$t$.
\end{proof}

Thus we can cleanly decompose a subspace $T$ into the two defect subspaces $T_{LD}$ and $T_{RD}$ as well as the defect isomorphism $T_J\colon [\vec y]\mapsto[\vec x(\vec y)]$.
When $T$ corresponds to a stochastic isometry $O\in O_t(d)$, i.e., $T = T_O= \{(O\vec y,\vec y)\}$, then both defect subspaces are trivial and the defect isomorphism $T_J$ can be identified with~$O$ itself.

\medskip

According to \cref{prp:defect subspaces}, for every $T\in\Sigma_{t,t}(d)$, the left and right defect subspaces $T_{LD}$ and $T_{RD}$ necessarily have the same dimension and $\vec 1_t\in T_{LD}$ if and only if $\vec 1_t\in T_{RD}$.
Note that if $T\in\Sigma_{t,t}(d)$ and $O, O'\in O_t(d)$, then $T' = O T O'$ has defect subspaces
\begin{align}\label{eq:defect space action}
  T'_{LD} = O T_{LD} \quad\text{and}\quad T'_{RD} = (O')^T T_{RD}
\end{align}
and defect isomorphism
\begin{align}\label{eq:defect iso action}
  T'_J \colon [\vec y] \mapsto [O \vec x(O' \vec y)].
\end{align}
Which elements $T'\in\Sigma_{t,t}(d)$ can be obtained in this way?
Clearly, the left-right action by $O_t(d)$ preserves the common dimension of the defect subspaces and whether the all-ones vector is contained.
We will now show that these are the only two invariants.

\begin{lem}\label{lem:witt first}
Let $N, M\subseteq\FF_d^t$ be two defect subspaces with $\dim N=\dim M$ and $\vec 1_t\in N$ if and only if $\vec 1_t\in M$.
Then there exists $O\in O_t(d)$ such that $ON=M$.
\end{lem}
\begin{proof}
Let $\tilde N \coloneqq N+\ZZ_d\vec 1_t$ and $\tilde M \coloneqq M+\ZZ_d\vec 1_t$.
The assumption implies that $\dim \tilde N=\dim \tilde M$.
Choose any linear isomorphism $\tilde O\colon \tilde N\to \tilde M$ such that $\tilde O \vec 1_t = \vec 1_t$.
Since both $N$ and $M$ are totally isotropic and co-stochastic, $\tilde O$ is an isometry with respect to the symmetric bilinear form $\vec x\cdot\vec y$.

If $d>2$, we can directly apply the usual version of Witt's lemma for symmetric bilinear forms of odd characteristic~\cite{wilson2009finite} to see that $\tilde O$ extends to an isometry map~$O$ which by construction is also stochastic, i.e., $O\in O_t(d)$.

For $d=2$, we appeal to the version of Witt's lemma from~\cite{wood1993witt} for the $\ZZ_4$-valued quadratic form~$q(\vec x)=\vec x\cdot\vec x\pmod4$ from \cref{rem:O_t}.
Here we need to verify two conditions:
(i) $\tilde O$ should be an isometry with respect to~$q$, i.e., $q(\tilde O\vec x) = q(\vec x)$ for every $\vec x\in \tilde N$.
Since we already know that $\tilde O$ is orthogonal, it suffices to check this condition on a generating set.
By construction, $\tilde O\vec 1_t=\vec 1_t$, so the condition is clearly true for $\vec x=\vec 1_t$.
On the other hand, both defect subspaces are totally isotropic, which means that $q(\tilde O\vec x)=0=q(\vec x)\pmod4$ for every $\vec x\in N$.
Thus, the first condition is satisfied.
(ii) We also need to check is that $\tilde N \cap I^\perp = \tilde M \cap I^\perp$, where~$I \coloneqq \{ \vec y \in \ZZ_2^t : \vec y\cdot\vec y=0 \pmod2 \}$.
But $I = \vec 1_t^\perp$ and hence $I^\perp=\ZZ_2 \vec 1_t$.
By construction, $\vec 1_t \in \tilde N$ and $\vec 1_t\in \tilde M$, so the second condition is also satisfied.
We conclude that $\tilde O$ extends to an isometry~$O$ with respect to the quadratic form~$q$, which implies that $O\in O_t(2)$ (\cref{rem:O_t}).
\end{proof}

\begin{cor}
Let $N, M\subseteq\FF_d^t$ be two defect subspaces with $\dim N=\dim M$ and $\vec 1_t\in N$ if and only if $\vec 1_t\in M$.
Then there exists $T\in\Sigma_{t,t}(d)$ such that $T_{LD}=N$ and $T_{RD}=M$.
\end{cor}
\begin{proof}
  Take $O\in O_t(d)$ as in \cref{lem:witt first}.
  Then, $O^T M = N$, $O^T M^\perp=N^\perp$, and hence $J\colon [\vec y] \mapsto [O^T \vec y]$ is a defect isomorphism.
  Then the subspace~\eqref{eq:T from defect spaces} has the desired defect spaces.
\end{proof}

\begin{lem}\label{lem:witt second}
Let $J\colon N^\perp/N \to M^\perp/M$ be a defect isomorphism.
Then there exists an $O\in O_t(d)$ inducing $J$, i.e., $ON=M$ and $[O\vec x]=J[\vec x]$ for every $[\vec x]\in N^\perp$.
\end{lem}
\begin{proof}[Proof (Sketch)]
Note that the existence of $J$ implies that $\dim M=\dim N$ as well as $\vec 1_t\in M$ if and only if $\vec 1_t\in N$.
This means that we can choose a linear isomorphism $\tilde J\colon N^\perp\to M^\perp$ that fixes $\vec 1_t$, sends $N$ to $M$ and which restricts to~$J$.
As in the proof of \cref{lem:witt first}, we can use the appropriate version of Witt's lemma to obtain the existence of an extension $O\in O_t(d)$.
\end{proof}

\begin{cor}[Equivalence of double cosets]\label{cor:cosets}
Let $T, T'\in\Sigma_{t,t}(d)$.
Then, $T' \in O_t(d) T O_t(d)$ if and only if $\dim T_{LD} = \dim T'_{LD}$ and $\vec 1_t\in T_{LD}$ iff $\vec 1_t\in T'_{RD}$.
In particular, $\Sigma_{t,t}(d)$ consists of no more than $t$ double cosets.
\end{cor}
\begin{proof}
It is clear that the two conditions are necessary.
We will now argue that they are sufficient.
First, use \cref{lem:witt first} to find $O$ and $O'$ such that $O T'_{LD} = T_{LD}$ and $O' T_{RD} = T'_{RD}$.
Then $T'' \coloneqq O T' O'$ is such that $T''_{LD} = T_{LD}$ and $T''_{RD} = T_{RD}$ (\cref{eq:defect space action}).
Next, use \cref{lem:witt second} to obtain some $O''$ that induces the defect isomorphism $T_J (T''_J)^{-1}\colon T_{LD}^\perp/T_{LD}\to T_{LD}^\perp/T_{LD}$.
Then $O'' T'' = T$, concluding the proof.
For the last remark, note that the dimension of a totally isotropic subspace is never larger than~$t/2$.
If the dimension is zero, then it cannot contain~$\vec 1_t$, while if the dimension is $t/2$ then it is Lagrangian, hence must contain~$\vec 1_t$.
Hence the number of possible dimensions is at most $1 + (t/2-1)2 + 1 = t$.
\end{proof}

\begin{rem}
We can also restrict to either the left or the right action.
In this case, the resulting cosets are classified by the right and left defect subspace, respectively, which is an arbitrary defect subspace in the sense of \cref{dfn:defect subspaces}.
\end{rem}

Next, we give an explicit description of the operators $R(T) = r(T)^{\ot n}$ in terms of the defect subspaces and the defect isomorphism.
If $N$ is a defect subspace, define the \emph{coset states}
\begin{align*}
  \ket{N,[\vec x]} \coloneqq |N|^{-1/2} \sum_{\vec z\in N} \ket{\vec x + \vec z},
\end{align*}
which form an orthonormal family for $[\vec x]\in N^\perp/N$.

\begin{lem}\label{lem:form of r(T)}
Let $T\in\Sigma_{t,t}(d)$, with defect subspaces $T_{LD}$, $T_{RD}$ and defect isomorphism $T_J\colon T_{LD}^\perp/T_{LD}\to T_{RD}^\perp/T_{RD}$.
Then:
\begin{align}\label{eq:r(T) via defect spaces}
  \frac 1{|T_{LD}|} r(T) = \!\!\!\sum_{[\vec y]\in T_{RD}^\perp/T_{RD}}\!\!\!\!\!\! \ket{T_{LD},T_J [\vec y]} \bra{T_{RD},[\vec y]}.
\end{align}
Thus, $r(T)$ is proportional to a partial isometry, and $\rank r(T)=|T_{LD}^\perp/T_{LD}|=|T_{RD}^\perp/T_{RD}|$.
\end{lem}
\begin{proof}
We obtain \eqref{eq:r(T) via defect spaces} directly from \cref{prp:defect subspaces} and~\eqref{eq:T from defect spaces}.
Since the coset states form two orthonormal families and $T_J$ is a bijection, the formula for the rank follows at once.
\end{proof}

Now consider the case when the left and right defect subspaces coincide and the defect isomorphism is trivial.
That is,
\begin{equation}\label{eq:css T}
\begin{aligned}
  T
  &= \{ (\vec x+\vec z,\vec x+\vec w) \;:\; [\vec x]\in N^\perp/N, \, \vec z, \vec w\in N \} \\
  &= \{ (\vec x, \vec y) \;:\; \vec y \in N^\perp, \vec x\in[\vec y] \}
  = \{ (\vec x, \vec y) \;:\; \vec x \in N^\perp, \vec y\in[\vec x] \},
\end{aligned}
\end{equation}
where $N \coloneqq T_{LD} = T_{RD}$ is an arbitrary defect subspace.
In view of \cref{cor:cosets}, any double coset contains a subspace of this form.
In this case, $r(T)$ and $R(T)$ are related to a well-known family of codes in quantum information theory.
To state the result, define the Weyl operators of, respectively, shift and multiply type:
\begin{align*}
  Z_{\vec p} = W_{\vec p,\vec 0} \qquad\text{ and }\qquad X_{\vec q} = W_{\vec 0,\vec q}.
\end{align*}
Given any totally isotropic subspace $N\subseteq\ZZ_d^t$, the set % (i.e., $N\subseteq N^\perp$), the set
\begin{align*}
  \CSS(N)\coloneqq\{ Z_{\vec p} X_{\vec q} : \vec q, \vec p \in N \}
\end{align*}
forms a stabilizer group of cardinality $\lvert N\rvert^2$ (since $N$ is self-orthogonal, the Weyl operators commute).
Such codes are a simple variant of \emph{Calderbank-Shor-Sloane (CSS) codes}~\cite{steane1996error,calderbank1996good,steane1996multiple}.
The projection onto the code space can be written as
\begin{align}\label{eq:css projector}
  P_{\CSS(N)} = \frac1{|N|^2} \sum_{\vec q,\vec p\in N} Z_{\vec p} X_{\vec q}
\end{align}
By taking the trace of \cref{eq:css projector}, one finds the dimension of the code is given by $d^{t-2\dim N}=|N^\perp/N|$.
One can readily confirm that the coset states $\ket{N,[\vec z]}$ for $[\vec z]\in N^\perp/N$ form an orthonormal basis, so
\begin{align}\label{eq:P css}
  P_{\CSS(N)} = \!\!\!\sum_{[\vec x]\in N^\perp/N}\!\!\!\!\!\! \ket{N,[\vec x]} \bra{N,[\vec x]},
\end{align}
In particular, all this applies in the situation of \cref{eq:css T}.
It follows that $r(T)$ and $R(T)$ are proportional to orthogonal projections onto CSS codes associated with the defect subspaces:

\begin{thm}[CSS codes]\label{thm:css}
  Suppose that $T$ is of the form~\eqref{eq:css T}, i.e., its left and right defect subspaces coincide and that the defect isomorphism is trivial.
  Let $N \coloneqq T_{LD} = T_{RD}$.
  Then,
  \begin{align*}
    r(T)
    = |N| \, P_{\CSS(N)}
    = d^{\dim N} P_{\CSS(N)}.
  \end{align*}
  Conversely, if $T\in\Sigma_{t,t}(d)$ is such that $r(T)$~is an orthogonal projection, then $T$ is of the form~\eqref{eq:css T}.
\end{thm}
\begin{proof}
  The formula for $r(T)$ follows directly by comparing \cref{eq:P css,eq:r(T) via defect spaces}.

  Conversely, suppose that $r(T)$ is an orthogonal projection.
  We see from \cref{eq:r(T) via defect spaces} that the range of $r(T)$ is spanned by the $\ket{T_{LD},[\vec x]}$, so we must have
  \begin{align*}
    \ket{T_{LD},[\vec x]}
  = r(T) \ket{T_{LD},[\vec x]}
  = \!\!\!\!\!\!\sum_{[\vec y]\in T_{RD}^\perp/T_{RD}}\!\!\!\!\!\! \ket{T_{LD},T_J [\vec y]} \braket{T_{RD},[\vec y] \,|\, T_{LD},[\vec x]}.
  \end{align*}
  Since the coset states $\ket{T_{LD},[\vec x]}$ form a basis, it follows that
  \begin{align*}
    \braket{T_{RD},[\vec y] \,|\, T_{LD},[\vec x]} = \delta_{[\vec x],T_J[\vec y]}
  \end{align*}
  for all $\vec x\in T_{LD}$ and $\vec y\in T_{RD}$.
  When $[\vec x]=T_J[\vec y]$, then $\braket{T_{RD},[\vec y] \;|\; T_{LD},[\vec x]} = 1$, which implies that $T_{LD} = T_{RD}$ (the inner product is at most $|T_{LD}\cap T_{RD}|/|T_{LD}|$ in absolute value).
  Denoting the common defect subspace by $N$, it follows that $\delta_{[\vec x],[\vec y]} = \braket{N,[\vec y] | N,[\vec x]} = \delta_{[\vec x],T_J[\vec y]}$, so $T_J$ is trivial.
\end{proof}

Finally, we can equip the set of subspaces $\Sigma_{t,t}(d)$ with a \emph{semigroup structure}, denoted by $\circ$, such that the assignment $T \mapsto R(T)$ becomes a representation, i.e.,
\begin{align}\label{eq:semigroup rule}
  R(T_1) R(T_2)
= |N_1 \cap N_2|^n \, R(T_1\circ T_2)
= d^{n\dim (N_1 \cap N_2)} \, R(T_1\circ T_2).
\end{align}
First, if $T_1$ or $T_2$ are associated to a stochastic isometry in~$O_t(d)$, then we can simply define $T_1 \circ T_2$ as in \cref{dfn:O on Sigma}.
In particular, the diagonal subspace is the identity element.

Next, consider the case that $T_1$ and $T_2$ are of the form~\eqref{eq:css T}, associated to defect subspaces~$N_1$ and $N_2$.
Then we may define $T_1 \circ T_2$ as the Lagrangian stochastic subspace~$T$ with
\begin{equation}\label{eq:simple semigroup}
\begin{aligned}
  T_{LD} &= (N_1 + N_2) \cap N_1^\perp = N_1^\perp\cap N_2 + N_1, \\
  T_{RD} &= (N_1 + N_2) \cap N_2^\perp = N_1\cap N_2^\perp + N_2, \\
  T_J &\colon T_{RD}^\perp/T_{RD} \to T_{LD}^\perp/T_{LD}, \quad [\vec y]\mapsto[\vec x]
\end{aligned}
\end{equation}
where $\vec x$ is such that $\vec x - \vec y \in N_1 + N_2$.

\begin{lem}
The data in~\eqref{eq:simple semigroup} defines a subspace in~$T\in\Sigma_{t,t}(d)$ such that \cref{eq:semigroup rule} holds.
\end{lem}
\begin{proof}
We first verify that $T_{LD}$ is a defect subspace.
Thus, let $\vec n_1+\vec n_2\in N_1^\perp$, where $\vec n_1\in N_1$ and $\vec n_2\in N_2$.
Then,
\begin{align*}
  q(\vec n_1+\vec n_2) = q(\vec n_1) + q(\vec n_2) + 2\vec n_1\cdot \vec n_2 = 2 \vec n_1\cdot \vec n_2 = 0\pmod D,
\end{align*}
The second step holds since $N_1$ and $N_2$ are defect subspaces, so $q(\vec n_1) = q(\vec n_2) = 0\pmod D$, and the third step holds since $\vec n_2 \in N_1 + N_1^\perp = N_1^\perp$, so $\vec n_1\cdot \vec n_2=0\pmod d$.
Moreover, $\vec 1_t \in N_1^\perp \cap N_2^\perp \subseteq T_{LD}^\perp$, so $T_{LD}$ is also co-stochastic.
Similarly, one can check that $T_{RD}$ is defect subspace.

Next, we verify that $T_J$ is a well-defined defect space isomorphism.
Note that
\begin{align*}
  T_{LD}^\perp &= (N_1 + N_2^\perp) \cap N_1^\perp = N_1^\perp \cap N_2^\perp + N_1, \\
  T_{RD}^\perp &= (N_1^\perp + N_2) \cap N_2^\perp = N_1^\perp \cap N_2^\perp + N_2.
\end{align*}
which shows that for every $\vec y\in T_{RD}^\perp$ there exists $\vec x\in T_{LD}^\perp$ such that $\vec x-\vec y\in N_1 + N_2$.
The same holds vice versa, so the map
\begin{align}\label{eq:pre defect space iso}
  T_{RD}^\perp \to T_{LD}^\perp/T_{LD}, \quad \vec y \mapsto [\vec x]
\end{align}
is surjective provided it is well-defined.
Assume that $\vec x,\vec x'\in T_{LD}^\perp$ are two vectors such that $\vec x-\vec y, \vec x'-\vec y\in N_1 + N_2$.
Then, $\vec x - \vec x' \in T_{LD}^\perp \cap (N_1 + N_2) = T_{LD}$, which shows that~\eqref{eq:pre defect space iso} is indeed well-defined.
Note that its kernel is given by $T_{RD}^\perp \cap (N_1 + N_2) = T_{RD}$.
Thus, the induced map, which is precisely $T_J$ from~\eqref{eq:simple semigroup}, is a well-defined invertible linear map.
We still need to verify that $T_J$ is an isometry and stochastic.
The latter is clear, since $\vec 1_t - \vec 1_t = 0 \in N_1 + N_2$.
For the former, consider $[\vec x]\in T_{LD}^\perp/T_{LD}$ and $[\vec y]\in T_{RD}^\perp/T_{RD}$ such that $\vec y - \vec x = \vec n_1 + \vec n_2$, where $\vec n_1\in N_1$ and $\vec n_2\in N_2$.
Without loss of generality, we may assume that $\vec x, \vec y \in N_1^\perp\cap N_2^\perp$, so in particular $\vec x - \vec y \perp \vec y$ and $\vec n_2 \in N_1^\perp$.
Since $N_1$ and $N_2$ are totally isotropic, it follows that
$q(\vec x - \vec y) = 2\vec n_1\cdot\vec n_2 = 0 \pmod D$
and $q(\vec x) = q(\vec y) + q(\vec x - \vec y) + 2\vec y\cdot(\vec x-\vec y)=0\pmod D$.

Finally, we will establish that \cref{eq:semigroup rule} holds with $T_1 \circ T_2 = T$.
It sufices to prove the claim for $n=1$:
\begin{align*}
  r(T_1) r(T_2)
&= \sum_{\vec x\in N_1^\perp} \sum_{\vec y\in N_2^\perp} \bigl|[\vec x] \cap [\vec y]\bigr| \, \ket{\vec x}\bra{\vec y}
= |N_1 \cap N_2| \sum_{\vec x\in N_1^\perp} \sum_{\vec y\in N_2^\perp} \delta_{\vec y-\vec x \in N_1 + N_2} \ket{\vec x}\bra{\vec y}
\\&= |N_1 \cap N_2| \sum_{\vec y\in T_{RD}^\perp} \sum_{\vec x \in T_{LD}^\perp} \delta_{\vec y-\vec x \in N_1 + N_2} \ket{\vec x}\bra{\vec y}
= |N_1 \cap N_2| \sum_{\vec y\in T_{RD}^\perp} \sum_{\vec x\in T_J[\vec y]} \ket{\vec x}\bra{\vec y}
\\&= |N_1 \cap N_2| \, r(T).
\end{align*}
In the third step, we used that for $\vec x\in N_1^\perp$ and $\vec y\in N_2^\perp$, the condition that $\vec y-\vec x\in N_1+N_2$ implies that $\vec x\in T_{LD}^\perp$ and $\vec y \in T_{RD}^\perp$, and in the fourth step we used the definition of~$T_J$.
\end{proof}

Finally, if $T_1$ and $T_2$ are arbitrary subspaces in~$\Sigma_{t,t}(d)$ then we can always left and right multiply $T_1$ and $T_2$ by suitable stochastic isometries, thereby reducing to the preceding two cases (cf.~\cref{eq:defect space action}).
The semigroup structure is highly useful for calculations (see \cref{subsec:examples} below and~\cite{rajamsc}).
We believe that the projections exhibited by \cref{thm:css} and the semigroup structure of \cref{eq:semigroup rule} will be instrumental in understanding the fine-grained decomposition of $\mathcal H_n^{\ot t}$ into irreducible representations of $\Cliff(n,d)$ and $O_t(d)$, generalizing the results discussed below.
First results in this direction will be reported in~\cite{montealegre2019rank}, with a full analysis being a direction for future work.

%-----------------------------------------------------------------------------
\subsection{Examples}\label{subsec:examples}
%-----------------------------------------------------------------------------
It is instructive to compute the commutant for small values of~$t$.
One can verify that every subspace in $\Sigma_{2,2}(d)$, as well as in $\Sigma_{3,3}(2)$ corresponds to a permutation.
That is, in this case, $\Sigma_{t,t}(d) = O_t(d) = S_t$.
This is consistent with the fact that the Clifford group is always a unitary $2$-design, and even a $3$-design in the case of qubits~\cite{zhu2015multiqubit}.

For certain larger values of $d$ and $t$, it is still true that $\Sigma_{t,t}(d) = O_t(d)$, e.g., for $t=3$ and $d\equiv2\pmod3$~\cite{nezami2016multipartite}.
In this case, the double commutant theorem implies that we have a proper duality akin to Schur-Weyl duality:
\begin{align}\label{eq:howe schematic}
  ((\CC^d)^{\ot n})^{\ot t} = \bigoplus_\lambda V_{\Cliff(n,d),\lambda} \ot V_{O_t(d),\lambda},
\end{align}
where the $V_{\Cliff(n,d),\lambda}$ and $V_{O_t(d),\lambda}$ are pairwise inequivalent irreducible representations of $\Cliff(n,d)$ and of~$O_t(d)$, respectively.
It would be interesting to identify these representations further.
In fact, it would be more appropriate to call \cref{eq:howe schematic} a form of a \emph{Howe duality}, of which an example are well-known dualities between metaplectic and orthogonal groups.

In general, however, $O_t(d)$ is a proper subset of $\Sigma_{t,t}(d)$, and it is an open problem to obtain a complete duality theory in positive characteristic~\cite{howe1973invariant,gurevich2016small}.
We now discuss some explicit examples.

\begin{exa}[$d=3$, $t=3$]\label{eq:third moment qutrits}
In this case, $\mathcal O_3(3) = S_3$, and we have that~\cite{nezami2016multipartite}
\begin{align}\label{eq:cosets 3 3}
\Sigma_{3,3}(3) = S_3 \cup S_3
\left[\begin{array}{ccc|ccc}
  1&-1&0 & 1&-1&0 \\
  \hline
  0&0&0 & 1&1&1 \\
  \hline
  1&1&1 & 0&0&0
\end{array}\right]
S_3
\end{align}
where we identify the matrix with its row space, a Lagrangian stochastic subspace~$T$.
The double coset of $T$ contains only two elements, $T$ and $(1 2)T$.
In total, $\Sigma_{3,3}(3)$ contains $6+2=8$ elements, which is in agreement with \cref{thm:sigma size}.

Next, we note that $T$ corresponds to a CSS code as in \cref{eq:css T,thm:css}, with defect subspace~$N$ spanned by the all-ones vectors~$\vec 1_3$.
Thus, $R(T) = 3^n P$, where
\begin{align}\label{eq:projector P third moment qutrits}
  P \coloneqq P_{\CSS}(N) = p^{\ot n}, \quad p \coloneqq \sum_{x=0}^2 \left( \frac1{\sqrt3} \sum_{y=0}^2 \ket{x+y,y-x,y} \right) \left( \frac1{\sqrt3} \sum_{z=0}^2 \bra{x+z,z-x,z} \right)
\end{align}
is a projector of rank $3^n$ (\cref{eq:P css}).

It is now straightforward to derive the decomposition of $((\CC^3)^{\ot n})^{\ot 3}$ into irreducible representations of the Clifford group (for $n\geq2$).
We start with Schur-Weyl duality, which asserts that
\begin{align}\label{eq:schur weyl third moment qutrits}
  ((\CC^3)^{\ot n})^{\ot 3}
&= \bigoplus_{\lambda\vdash3} V_{U(3^n),\lambda} \ot V_{S_3,\lambda} \\
&= \Sym^3((\CC^3)^{\ot n}) \;\op\; U^{3^n}_{(2,1)} \ot V_{U(3^n),(2,1)} \;\op\; \Alt^3((\CC^3)^{\ot n}),
\end{align}
where $\lambda$ runs over all partitions of 3.
By \cref{eq:cosets 3 3}, the commutant is generated by $S_3$ and the projection~$P$.
Since $P$ commutes with all permutations, it follows that
\begin{align*}
  P_+ \coloneqq \Pi^{\text{sym}}_3 P \Pi^{\text{sym}}_3 &= \frac{3^{-n}}2 \left( R(T) + R((1 2) T) \right), \\
  P_- \coloneqq \Pi^{\text{alt}}_3 P \Pi^{\text{alt}}_3 &= \frac{3^{-n}}2 \left( R(T) + R((1 2) T) \right)
\end{align*}
are orthogonal projections onto subrepresentations of the Clifford group.
We can compute their dimensions readily by using the formula $\tr[R(S)] = d^{n \dim (S \cap \Delta)}$:
\begin{align}\label{eq:dims third moment qutrits}
  \dim W_\pm
= \tr[P_\pm]
% = \frac{3^{-n}}2 \left( \tr[R(T)] \pm \tr[R((1 2) T)] \right)
% = \frac{3^{-n}}2 \left( 3^{2n} \pm 3^n \right)
= \frac{3^n \pm 1}2
\end{align}
Thus we can decompose the symmetric and anti-symmetric subspaces further into four subrepresentations:
\begin{align*}
  \Sym^3((\CC^3)^{\ot n}) &\cong W_+ \op W_+^\perp, \\
  \Alt^3((\CC^3)^{\ot n}) &\cong W_- \op W_-^\perp.
\end{align*}
Since the commutant has dimension $\lvert\Sigma_{3,3}(3)\rvert=8$, these four representations along with $V_{U(3^n),(2,1)}$ (which appears twice in~\eqref{eq:schur weyl third moment qutrits}) are necessarily irreducible and pairwise inequivalent.
% there are 6 reps occuring. thus there are at least 6 irreps:   \sum_i m_i >= 6
% on the other hand, the dimension of the commutant is:          \sum_i m_i^2 = 8
% we know that at least one m_i appears twice, say m_1. thus:    m_1 >= 2
% it follows that 8 = \sum_i m_i^2 = m_1^2 + \sum_{i>1} m_i^2 >= m_1^2 + \sum_{i>1} m_i = m_1^2 + (\sum_i m_i - m_1) = m_1 (m_1 - 1) + \sum_i m_i >= 2 + 6 = 8
% these inequalities are tight if (i) m_i=1 for i>1 and (ii) m_1 = 2
We have thus fully decomposed~$((\CC^3)^{\ot n})^{\ot 3}$ into irreducible representations of~$\Cliff(n,3) \times S_3$.
\end{exa}

\noindent Next, we discuss some multi-qubit examples.

\begin{exa}[$d=2$, $t=4$]
As before, we find that $\mathcal O_4(2)=S_4$.
In addition to the $4!=24$ permutation subspaces, there exist $6$ more Lagrangian subspaces in $\Sigma_{4,4}(2)$ -- making a total of $30$, which is known to be the dimension of the commutant of the multi-qubit Clifford group for $n\geq3$~\cite[(10)]{zhu2015multiqubit}.
We can decompose~$\Sigma_{4,4}(2)$ into two double cosets in a form that is completely analogous to \cref{eq:cosets 3 3}:
\begin{align}\label{eq:cosets 2 4}
  \Sigma_{4,4}(2) = S_4 \cup S_4
\left[\begin{array}{cccc|cccc}
1&0&0&1 & 1&0&0&1 \\
0&1&0&1 & 0&1&0&1 \\
\hline
0&0&0&0 & 1&1&1&1 \\
\hline
1&1&1&1 & 0&0&0&0
\end{array}\right]
S_4
\end{align}
The given matrix is the generator matrix of a Lagrangian subspace which we denote by~$T_4$.
Similarly to above, the operator~$R(T_4)$ is proportional to a projector onto a CSS code, with defect subspace spanned by the all-ones vector~$\vec 1_4$.
This projector is given by \cref{eq:projection Pi_0}, and it can be used to decompose $((\CC^2)^{\ot n})^{\ot 4}$ into irreducible representations of the Clifford group, as explained in~\cite{zhu2016clifford}.
\end{exa}

\begin{exa}[$d=2$, $t=5$]
Likewise, for $t=5$, it is not hard to see that (cf.~\cite{rajamsc})
\begin{align}\label{eq:cosets 2 5}
  \Sigma_{5,5}(2) = S_5 \cup S_5
\left[\begin{array}{ccccc|ccccc}
1&0&0&1&0 & 1&0&0&1&0 \\
0&1&0&1&0 & 0&1&0&1&0 \\
0&0&1&1&0 & 0&0&1&1&0 \\
0&0&0&0&1 & 0&0&0&0&1 \\
\hline
0&0&0&0&0 & 1&1&1&1&0 \\
\hline
1&1&1&1&0 & 0&0&0&0&0
\end{array}\right]
S_5.
\end{align}
The displayed matrix corresponds to a subspace~$T_5$ of the form \cref{eq:css T}, with defect subspace spanned by the vector~$(1,1,1,1,0)$, and the operator~$R(T_5)$ is proportional to a projector onto a CSS code.
Indeed, we have
\begin{align*}
  R(T_5) = R(T_4) \ot I_2^{\ot n}.
\end{align*}
\end{exa}

We now discuss some interesting elements in the groups~$O_t(d)$.
For qubits, we have the class of anti-permutations introduced previously in \cref{eq:anti identity}.

\begin{dfn}[Anti-permutation]\label{dfn:anti-permutation}
  Let $\pi\in S_t$.
  We define the \emph{anti-permutation} $\bar\pi$ as the binary complement of the corresponding permutation matrix.
  Formally, it is the $t\times t$-matrix
  \begin{align*}
    \bar\pi = \vec 1_t \vec 1_t^T - \pi
  \end{align*}
  with entries in $\FF_2$, where we identify $\pi$ with the corresponding permutation matrix.
\end{dfn}

\begin{lem}
  Let $\pi\in S_t$.
  If $t\equiv2\pmod4$ then $\bar\pi \in O_t(d)$.
\end{lem}
\begin{proof}
  By \cref{rem:O_t}, it suffices to check that $\bar\pi$ is orthogonal and that $q(\vec x)=1$ for each column.
  The latter holds since each column of $\bar\pi$ contains $t-1\equiv1\pmod4$ ones.
  For the former,
  \begin{align*}
    \bar\pi^T \bar\pi
  = (\vec 1_t \vec 1_t^T - \pi^T) (\vec 1_t \vec 1_t^T - \pi)
  % = \vec 1_t \vec 1_t^T \vec 1_t \vec 1_t^T - \pi^T \vec 1_t \vec 1_t^T - \vec 1_t \vec 1_t^T \pi + \pi^T \pi
  = (t-2) \vec 1_t \vec 1_t^T + I \equiv I \pmod 2,
  \end{align*}
  where we used that ordinary permutation matrices are orthogonal and stochastic, as well as that $t$ is even.
\end{proof}

\begin{rem}
More generally, the entrywise binary complement maps any $A\in O_t(2)$ to an element $\bar A \in O_t(2)$ provided that the rows of~$A$ each have Hamming weight $w$ such that $t\equiv2w\pmod 4$.
\end{rem}

\noindent For $t\geq6$, the anti-permutations are distinct from the permutations, so in particular $O_t(2) \supsetneq S_t$.
(For $t=2$, the two sets coincide.)

\begin{exa}[$d=2$, $t=6$]
  The \emph{anti-identity} $\bar\id\in\mathcal O_6(2)$ and the corresponding subspace $T_{\bar\id} \subseteq \ZZ_2^6\op\ZZ_2^6$ are given by  (cf.\ \cref{eq:anti identity,eq:anti identity six}).
  \begin{align*}
    \bar\id = \begin{pmatrix}
      0&1&1&1&1&1 \\
      1&0&1&1&1&1 \\
      1&1&0&1&1&1 \\
      1&1&1&0&1&1 \\
      1&1&1&1&0&1 \\
      1&1&1&1&1&0
    \end{pmatrix}, \quad
    T_{\bar\id} = \left[\begin{array}{cccccc|cccccc}
      0&1&1&1&1&1 & 1&0&0&0&0&0\\
      1&0&1&1&1&1 & 0&1&0&0&0&0 \\
      1&1&0&1&1&1 & 0&0&1&0&0&0 \\
      1&1&1&0&1&1 & 0&0&0&1&0&0 \\
      1&1&1&1&0&1 & 0&0&0&0&1&0 \\
      1&1&1&1&1&0 & 0&0&0&0&0&1
    \end{array}\right].
  \end{align*}
\end{exa}

The anti-permutations admit several possible generalizations to odd primes~$d$.
One class of generalizations is given as follows.
For $\pi\in S_t$ with $d\nmid t$, define
\begin{align}\label{eq:anti perm qudit}
  \bar\pi = 2 t^{-1} \vec 1_t \vec 1_t^T - \pi,
\end{align}
where $t^{-1}$ denotes the multiplicative inverse of $t$ in $\FF_d$.
It is easy to verify that~$\bar\pi \in O_t(d)$.
Moreover, $\bar\pi$ is the only nontrivial linear combination of $\vec 1_t \vec 1_t^T$ and $\pi$ with this property.

Another class of generalizations is given by the formula in \cref{eq:orthogonal and stochastic}.
Let $\vec p\in\FF_d^t$ be vector with entries in $\{\pm1\}$ that is `balanced', i.e., $\vec p \cdot \vec 1_t = 0$ (this requires that $t$ is even).
If $d \nmid t$ and $\pi\in S_t$ is a permutation that stabilizes $\vec p$ up to a sign, i.e., $\pi\vec p = \pm\vec p$, then
\begin{align}\label{eq:antiperm 2}
  \tilde\pi = \pi \mp 2 t^{-1} \vec p \vec p^T
\end{align}
is an element in $O_t(d)$.
In particular, this yields a large family of `anti-identities' for odd~$d$.

\medskip

Another non-trivial example of a stochastic isometry can be constructed from the adjacency matrix $A$ of the edge-vertex graph of the icosahedron.
The icosahedron has $12$ vertices, so $A$ is a $12 \times 12$ binary matrix.
Any two vertices share either zero or two neighbors, which implies that $A$ is orthogonal.
Moreover, each vertex has $5$ neighbors, which implies $q(\vec x)=1$ for each column~$\vec x$ of $A$.
By \cref{rem:O_t}, it follows that $A\in O_{12}(2)$.
%\begin{align*}
%  A=
%  \left(
%  \begin{array}{cccccccccccc}
%   0 & 0 & 1 & 0 & 1 & 1 & 0 & 0 & 1 & 1 & 0 & 0 \\
%   0 & 0 & 0 & 1 & 0 & 0 & 1 & 1 & 0 & 0 & 1 & 1 \\
%   1 & 0 & 0 & 0 & 0 & 0 & 1 & 1 & 1 & 1 & 0 & 0 \\
%   0 & 1 & 0 & 0 & 1 & 1 & 0 & 0 & 0 & 0 & 1 & 1 \\
%   1 & 0 & 0 & 1 & 0 & 1 & 0 & 0 & 1 & 0 & 1 & 0 \\
%   1 & 0 & 0 & 1 & 1 & 0 & 0 & 0 & 0 & 1 & 0 & 1 \\
%   0 & 1 & 1 & 0 & 0 & 0 & 0 & 1 & 1 & 0 & 1 & 0 \\
%   0 & 1 & 1 & 0 & 0 & 0 & 1 & 0 & 0 & 1 & 0 & 1 \\
%   1 & 0 & 1 & 0 & 1 & 0 & 1 & 0 & 0 & 0 & 1 & 0 \\
%   1 & 0 & 1 & 0 & 0 & 1 & 0 & 1 & 0 & 0 & 0 & 1 \\
%   0 & 1 & 0 & 1 & 1 & 0 & 1 & 0 & 1 & 0 & 0 & 0 \\
%   0 & 1 & 0 & 1 & 0 & 1 & 0 & 1 & 0 & 1 & 0 & 0 \\
%  \end{array}
%  \right).
%\end{align*}
The space $T_{\bar A}$ generated by the element-wise complement of $A$ is just the extended Golay code $\mathcal{G}_{24}$.
The latter plays an important role in the invariant theory of the Clifford group as detailed in~\cite{nebe2006self}.
Note, however, that unlike $A$ and $T_A$, it is \emph{not} the case that $\bar A \in O_{12}(2)$ or  $T_{\bar A}\in\Sigma_{12,12}(2)$.
Working out the precise connection between $T_A$, the extended Golay code, and their respective roles in the representation theory of the Clifford group is an interesting problem we leave open.
Likewise, we leave open the question of whether $R(A)$ can be given a physical interpretation, as was the case for the anti-identity.
%Symmetry group of the icosahedron is $S_5\subset S_{12}$, so unlike the anti-identity, $R(B)$ is not permutation-invariant.

%=============================================================================
\section{Statistical properties of stabilizer states}\label{sec:moments}
%=============================================================================
In this section we discuss the statistical properties of the stabilizer states.
We use the techniques that we developed in the last section to prove an explicit formula for the $t$-th moment of random stabilizer states, which vastly generalizes previous results in the quantum information literature~\cite{zhu2015multiqubit,kueng2015qubit,webb2016clifford,zhu2016clifford,helsen2016representations}.
Throughout this section, $d$ is assumed to be a prime.

%-----------------------------------------------------------------------------
\subsection{Moments of random stabilizer states}\label{subsec:moments}
%-----------------------------------------------------------------------------
We start by studying the operator-valued \emph{$t$-th moment} of the uniform distribution over all stabilizer states in $(\CC^d)^{\ot n}$:
\begin{equation}\label{eq:defn moment}
  \EE \left[\ket S \bra S ^{\ot t} \right] \coloneqq \frac1{\lvert \Stab(n,d) \rvert} \sum_{\ket S\bra S \in \Stab(n,d)} \ket S \bra S^{\ot t}
\end{equation}
Clearly this operator can be used to calculate the average value of any polynomial of degree $t$ in the coefficients of the wavefunction of a random stabilizer state.

Note that the operator $\EE[\ket S \bra S^{\ot t}]$ is invariant under conjugation by Clifford operators.
This is because the set of stabilizer states is a single orbit of the Clifford group.
Thus $\EE[\ket S \bra S ^{\ot t}]$ is in the commutant of the Clifford group and, assuming $n\geq t-1$, can be written in terms of the basis $R(T)$ from \cref{thm:commutant},
\begin{equation}\label{eq:avg stab with coeffs}
\EE \left[\ket S \bra S ^{\ot t} \right]=\sum_{T\in \Sigma_{t,t}(d)}{\gamma_T \, R(T)},
\end{equation}
for certain coefficients $\gamma_T\in \CC$.
In this section, we will show that these coefficients are all equal and establish an explicit formula for the $t$-th moment of a random stabilizer state which holds for all values of $t$ and $n$.
We start with some useful lemmas.

\begin{rem}[Sum of traces of the $R(T)$]\label{rem:sum trace}
Recall that in order to establish \cref{thm:sigma size}, we determined the cardinality of the set~$\Sigma^\ell_{t,t}(d)$, whose elements are the subspaces $T\in\Sigma_{t,t}(d)$ with $\dim(T\cap\Delta)=t-\ell$.
The significance of the parameter~$\ell$ is that
\begin{align*}
  \tr R(T) = \left(\tr r(T)\right)^n = d^{(t-\ell) n}
\end{align*}
for every subspace $T\in\Sigma_{t,t}^\ell(d)$. 
  Thus, we can e.g.\ compute the sum of the traces of all $R(T)$ by using \cref{eq:sigma ell claim} and the Gaussian binomial formula \cref{eq:q binomial formula}:
\begin{align*}
  \sum_{T\in\Sigma_{t,t}(d)} \tr R(T)
  = \sum_{\ell=0}^{t-1} \bigl\lvert \Sigma_{t,t}^\ell(d)\bigr\rvert d^{(t-\ell)n} 
  %= d^{tn} \sum_{\ell=0}^{t-1} \binom{t-1}{\ell}_d d^{\frac{\ell(\ell-1)}2} d^{-\ell n} 
  % Gaussian binomial with k=l, t=d^-n
  = d^{nt} \prod_{k=0}^{t-2} (1+d^{k-n})
  = d^{n} \prod_{k=0}^{t-2} (d^k+d^{n}).
%  = d^{nt} (-d^{-n}; d)_{t-1},
\end{align*}
  This number can be expressed in terms of the \emph{$q$-Pochhammer symbol} as 
  \begin{align*}
    \sum_{T\in\Sigma_{t,t}(d)} \tr R(T)
	=
	d^{nt} (-d^{-n}; d)_{t-1}.
  \end{align*}
  For $n=0$, we recover the cardinality of $\Sigma_{t,t}(d)$, in agreement with \cref{thm:sigma size}.
\end{rem}

Next, we prove a formula that relates moments of stabilizer states for different numbers of qudits.

\begin{lem}\label{lem:bootstrap}
Let $N\geq n>0$. Then:
\begin{align*}
% &\EE_{\Stab(n',d)} \left[(\bra 0 ^{\ot n'-n} \ket S)^{\ot t} (\bra S \ket{0}^{\ot n'-n}) ^{\ot t} \right]  =\\
\left( I^{\ot nt} \ot \bra{0}^{\ot (N-n)t} \right) \EE_{\Stab(N,d)} \left[\ket S \bra S^{\ot t} \right] \left( I^{\ot nt} \ot \ket{0}^{\ot (N-n)t} \right)
\;\propto\;
\EE_{\Stab(n,d)} \left[\ket S \bra S^{\ot t} \right],
\end{align*}
and both operators are nonzero.
\end{lem}
\begin{proof}
If $\ket S\in\Stab(N,d)$ is a stabilizer state then the partial projection $\left(I \ot \bra 0^{\ot (N-n)}\right) \ket S$ is either zero or proportional to a stabilizer state in $\Stab(n,d)$ (see, e.g.,~\cite[App.~G]{hayden2016holographic}).
Thus, there exist coefficients $\alpha_{S'}$ for $S'\in\Stab(n,d)$ such that
\begin{align*}
\left( I^{\ot nt} \ot \bra{0}^{\ot (N-n)t} \right) \EE_{\Stab(N,d)} \left[\ket S \bra S^{\ot t} \right] \left( I^{\ot nt} \ot \ket{0}^{\ot (N-n)t} \right)
= \sum_{S'\in\Stab(n,d)} \alpha_{S'} \ket{S'}\bra{S'}^{\ot t}
\end{align*}
It is clear that the left-hand side operator is nonzero and invariant under conjugation by $U^{\ot t}$ for any Clifford unitary $U\in\Cliff(n,d)$.
Thus, we can replace the right-hand side by its Clifford average, which is plainly proportional to $\EE_{\Stab(n,d)} \left[\ket S \bra S^{\ot t} \right]$.
\end{proof}

\begin{thm}[$t$-th moment]\label{thm:t-th moment}
Let $n,t\geq1$.
Then the $t$-th moment of a random stabilizer state in $(\CC^d)^{\ot n}$ is given by the formula
\begin{equation}\label{eq:avg stab}
\EE_{\Stab(n,d)} \left[\ket S \bra S ^{\ot t} \right]=\frac1{Z_{n,d,t}} \sum_{T\in \Sigma_{t,t}(d)} R(T),
\end{equation}
  where $Z_{n,d,t}=d^n\prod_{k=0}^{t-2} (d^k+d^n)=d^{nt} (-d^{-n}; d)_{t-1}$.
\end{thm}

\begin{proof}
It suffices to argue that the left-hand side and right-hand side are proportional, since the formula for the proportionality constant~$Z_{n,d,t}$ follows immediately from \cref{rem:sum trace} and comparing traces.
Fixing $d$ and $t$, we will proceed in two steps.
First, we will argue that there exist $\alpha_n\in\CC$ and $\beta_T\in\CC$ for each $T\in\Sigma_{t,t}(d)$ such that
\begin{align}\label{eq:decoupled moment}
  \EE_{\Stab(n,d)} \left[\ket S \bra S^{\ot t} \right] = \alpha_n \sum_{T\in\Sigma_{t,t}(d)} \beta_T \, r(T)^{\ot n}.
\end{align}
We will show this first for $n\geq t-1$ and then for all~$n$.
Afterwards, we will find that the $\beta_T$ are necessarily equal, which as just discussed implies the claim.

Let us first assume that $n\geq t-1$.
By \cref{thm:commutant}, there exist coefficients $\gamma_{n,T}\in\CC$ such that
\begin{align}\label{eq:gamma expansion}
\EE_{\Stab(n,d)} \left[\ket S \bra S^{\ot t} \right] = \sum_{T\in\Sigma_{t,t}(d)} \gamma_{n,T} \, r(T)^{\ot n},
\end{align}
since the left-hand side commutes with arbitrary $t$-th tensor powers of Clifford unitaries.
It follows that, for every $N\geq n$,
\begin{equation}\label{eq:gamma projection}
\begin{aligned}
&\left( I^{\ot nt} \ot \bra{0}^{\ot (N-n)t} \right) \EE_{\Stab(N,d)} \left[\ket S \bra S^{\ot t} \right] \left( I^{\ot nt} \ot \ket{0}^{\ot (N-n)t} \right)
\\&\qquad = \sum_{T\in\Sigma_{t,t}(d)} \gamma_{N,T} \, r(T)^{\ot n} \braket{0^{\ot t} | r(T) | 0^{\ot t}}^{N-n}
= \sum_{T\in\Sigma_{t,t}(d)} \gamma_{N,T} \, r(T)^{\ot n},
\end{aligned}
\end{equation}
since $\vec 0\in T$ for every subspace~$T$.
From \cref{lem:bootstrap} we know that \cref{eq:gamma expansion,eq:gamma projection} are proportional and nonzero.
Since the operators~$r(T)^{\ot n}$ are also linearly independent for~$n\geq t-1$, it follows that there exist $\alpha_n$ and $\beta_T$ such that $\gamma_{T,n} = \alpha_n \beta_T$ for all $n\geq t-1$ (e.g., we can choose $\beta_T \coloneqq \gamma_{T,t-1}$).
Thus, we have established \cref{eq:decoupled moment} for $n\geq t-1$.
To extend its validity to all values of~$n$, we observe that \cref{eq:gamma projection} holds also when $N\geq t-1>n$.
Together with \cref{lem:bootstrap}, we find that, indeed,
\begin{align*}
  \EE_{\Stab(n,d)} \left[\ket S \bra S^{\ot t} \right]
\propto \sum_{T\in\Sigma_{t,t}(d)} \gamma_{N,T} \, r(T)^{\ot n}
\propto \sum_{T\in\Sigma_{t,t}(d)} \beta_T \, r(T)^{\ot n},
\end{align*}
which shows that there exist constants $\alpha_n$ and $\beta_T$ such that \cref{eq:decoupled moment} holds for all values of~$n$.

We will now argue that, in this case, the $\beta_T$ are necessarily all equal.
For this, we compute the expectation value of an operator $R(T)^\dagger = r(T)^{\ot n,\dagger}$.
On the one hand, by \cref{eq:not quite eigenvectors},
\begin{align*}
  \tr\left[ R(T)^\dagger \EE_{\Stab(n,d)} \left[\ket S \bra S^{\ot t} \right] \right]
= \EE_{\Stab(n,d)} \braket{S^{\ot t} | R(T)^\dagger | S^{\ot t}} = 1.
\end{align*}
On the other hand,
\begin{align*}
&\tr\left[ R(T)^\dagger \EE_{\Stab(n,d)} \left[\ket S \bra S^{\ot t} \right] \right]
= \alpha_n \sum_{T'\in\Sigma_{t,t}(d)} \beta_{T'} \, \left( \tr r(T)^\dagger r(T') \right)^n
\\&= \alpha_n \sum_{T'\in\Sigma_{t,t}(d)} \beta_{T'} \, d^{n \dim (T \cap T')}
= \alpha_n d^{nt} \left( \beta_T + O(d^{-n}) \right)
\end{align*}
in the limit of large $n$.
Thus, all $\beta_T$ must be equal, and the statement of the theorem follows.
\end{proof}

\begin{rem}
\Cref{thm:t-th moment} is reminiscent of the following average formula for the tensor power of a Haar-random pure state $\psi$ in~$\CC^D$, which follows from Schur's lemma and Schur-Weyl duality:
\begin{equation}\label{eq:avg state haar}
\EE_{\psi\!\text{ Haar}} \left[\ket\psi\bra\psi ^{\ot t} \right]] = \frac1{\prod_{k=0}^{t-1} (k+D)} \sum_{\pi \in S_t} R(\pi).
\end{equation}
\end{rem}

\begin{rem}\label{rem:known design stuff}
When $d$ is an odd prime, $\Sigma_{t,t}(d)=S_t$ for $t\leq2$, but not for $t\geq3$.
Thus \cref{eq:avg stab,eq:avg state haar} match for $t\leq2$ and deviate for $t\geq3$.
Since the operators $R(T)$ are linearly independent for sufficiently large~$n$, this shows that stabilizer states in odd prime dimension are 2-designs, but not 3-designs or higher (provided $n\geq2$).
Similarly, for $d=2$, $\Sigma_{t,t}(2)=S_t$ for $t\leq3$, but not or $t\geq4$, which shows that multiqubit stabilizer states form 3-designs, but not 4-designs or higher~\cite{kueng2015qubit} (provided $n\geq3$).
\end{rem}

Remarkably, the theory developed in this section allows us to design complex projective $t$-designs for any order~$t$ from the Clifford group orbits of a finite number of fiducial states.
We explain this in \cref{sec:designs} below.

%-----------------------------------------------------------------------------
\subsection{Minimal projections for stabilizer testing}\label{subsec:stabilizer testing revisited}
%-----------------------------------------------------------------------------
We now return to the problem of stabilizer testing; we revisit our solution from \cref{sec:stabilizer testing} and characterize minimal stabilizer tests with perfect completeness.

In \cref{sec:stabilizer testing}, we found that perfectly complete stabilizer tests were in any local dimension~$d$ given by the following accepting POVM element on $t=2s$ copies of $(\CC^d)^{\ot n}$,
\begin{align}\label{eq:p accept restated}
  \Pi_{s,\text{accept}} = \frac12 \left( I + V_s \right),
\end{align}
where $V_s$ is the Hermitian unitary defined in \cref{eq:V for qudits} and $(d,s)=1$.
This means that $V_s$ is a unitary operator with the property that $2s$-th tensor powers of every pure stabilizer state $\ket S$ are contained in its $+1$ eigenspace:
\begin{equation}\label{eq:accept stabilizer}
V_s \ket S^{\ot 2s} =\ket S^{\ot 2s}
\end{equation}
for any pure stabilizer states $\ket S$.
Our soundness result implies that, conversely, these are the only tensor power states with this property.

Note that $V_s$ is an operator in the commutant of the Clifford action.
This is immediate by comparing \cref{eq:O_t action,eq:V_s as R(T)}, which also shows that $V_s$ is precisely the operator $R(\tilde\id)$ associated with the `anti-identity'~$\tilde\id$ defined in \cref{eq:orthogonal and stochastic}!

In fact, \emph{any} $R(O)$ stabilizes the $t$-th tensor powers of stabilizer states:
For all $O\in O_t(d)$,
\begin{align}\label{eq:eigenvectors restated}
  R(O) \ket S^{\ot t} = \ket S^{\ot t}.
\end{align}
We proved this in \cref{eq:eigenvectors}.
As we just saw, $V_s$ is such an operator, so \cref{eq:eigenvectors restated} generalizes \cref{eq:accept stabilizer}.

Note that, since $\tilde\id$ squares to the identity, it generates a subgroup of $O_t(d)$ that contains two elements: $\{\id, \tilde\id\}$.
We can thus interpret the projector~\eqref{eq:p accept restated} as the projector onto the invariant subspace for the action of this subgroup.
This suggest that we look more generally at the invariant subspaces associated with subgroups of $O_t(d)$.
Larger subgroups corresponds to projectors onto smaller invariant subspaces.
In particular, the minimal projector corresponds to the full group $O_t(d)$, i.e.,
\begin{align}\label{eq:minimal projector}
  \Pi^{\min}_t \coloneqq \frac1 {\lvert O_t(d) \rvert }\sum_{O\in O_t(d)} R(O).
\end{align}
By \cref{eq:eigenvectors restated}, $\Pi^{\min}_t$ accepts all stabilizer tensor powers.
Remarkably, it is the \emph{minimal} projector with this property, as follows from the following theorem.

\restateThmMinimalTest
\begin{proof}
Note that the $t$-th moment $\rho\coloneqq\EE \left[\ket S \bra S ^{\ot t} \right]$ defined in \cref{eq:defn moment} is a density operator that is exactly supported on the span of the stabilizer tensor powers.
By the preceding discussion, it remains to prove that the support of $\Pi^{\min}_t$ is contained in the support of $\rho$.
We start with \cref{eq:avg stab}:
\begin{align*}
\EE_{\Stab(n,d)} \left[\ket S \bra S ^{\ot t} \right]=\frac1{Z_{n,d,t}} \sum_{T\in \Sigma_{t,t}(d)} R(T)
\end{align*}
Recall from \cref{eq:double cosets} that we can decompose $\Sigma_{t,t}(d)$ into $k$ double cosets,
\begin{align*}
\Sigma_{t,t}(d) = O_t(d) T_1 O_t(d) \cup \dots \cup O_t(d) T_k O_t(d).
\end{align*}
One of the double cosets is just $O_t(d)$, say the first, corresponding to $T_1=\Delta$ and $R(T_1)=I$.
As a consequence of \cref{cor:cosets}, we can choose each representative~$T_i$ to be of the form~\eqref{eq:css T}.
Then, \cref{thm:css} shows that $R(T_i)$ is proportional to an orthogonal projection (by a positive proportionality constant) so in particular $R(T_i)\geq0$.
On the other hand, we can compute the sum over each double coset by
\begin{align*}
%   \Pi^{\min}_t R(T_k) \Pi^{\min}_t
% \propto \sum_{O,O'\in O_t(d)} R(O) R(T_k) R(O')
% = \sum_{O,O'\in O_t(d)} R(O T_k O')
% = \frac {\lvert O_t(d) \times O_t(d) \rvert} {\lvert O_t(d) T_k O_t(d) \rvert} \sum_{T \in O_t(d) T_k O_t(d)}
% \propto  \sum_{T \in O_t(d) T_k O_t(d)}
  \sum_{T \in O_t(d) T_i O_t(d)} R(T)
= \frac{|O_t(d) T_i O_t(d)|}{|O_t(d)| \times |O_t(d)|} \sum_{O, O' \in O_t(d)} R(O) R(T_i) R(O')
= c_i \, \Pi^{\min}_t R(T_i) \Pi^{\min}_t
\end{align*}
where $c_i \coloneqq |O_t(d) T_i O_t(d)| > 0$.
Together, we obtain that
\begin{align*}
\rho
&= \frac1{Z_{n,d,t}} \sum_{T\in\Sigma_{t,t}(d)} R(T)
% = \frac1{Z_{n,d,t}} \sum_{i=1}^k \sum_{T \in O_t(d) T_i O_t(d)} R(T)
= \frac1{Z_{n,d,t}} \sum_{i=1}^k c_i \, \Pi^{\min}_t R(T_i) \Pi^{\min}_t \\
&= \frac{c_1}{Z_{n,d,t}} \left( \Pi^{\min}_t + \sum_{i=2}^k \frac {c_i} {c_1} \, \Pi^{\min}_t R(T_i) \Pi^{\min}_t \right)
\geq \frac{c_1}{Z_{n,d,t}} \Pi^{\min}_t,
\end{align*}
which shows that the support of $\rho$ indeed contains the support of $\Pi^{\min}_t$.
\end{proof}

Note that there is no condition on $t$ in \cref{thm:minimal test}.
Indeed, while the theorem identifies the projector onto the span of stabilizer tensor powers precisely, it makes no assertion about whether this subspace contains other tensor power states than stabilizer states or not.
 % (e.g., for $t=1$ the range of $\Pi^{\min}_1$ is trivially the entire Hilbert space).
It therefore complements our results on stabilizer testing, from which we can read off values of $t$ such that the projector $\Pi_t$ and hence $\Pi^{\min}_t \leq \Pi_t$ contains only stabilizer tensor powers.

It is also interesting to ask about the minimal number of copies necessary for there to exist a perfectly complete stabilizer test that is dimension-independent.
For $d=2$, it is possible to show that $t=4,5$ copies of a random stabilizer state become on average indistinguishable from a Haar-random pure state as $n\to\infty$.
This can be done by an explicit calculation of the 4th and 5th moments using our \cref{thm:t-th moment,eq:cosets 2 4,eq:cosets 2 5} (for $t=4$, this has been carried out in~\cite{rajamsc}).
Thus, $t=6$ copies as in our \cref{thm:main qubits} are indeed optimal for multiqubit stabilizer testing.
For odd~$d$ (prime or not), we know from \cref{thm:main qudits} that~$t=4$ copies always suffice. % (while $t=2$ are certainly insufficient).
For $d\equiv1,5\pmod6$ it follows from \cref{thm:main three copies} below that even~$t=3$ copies suffice (and are optimal).
For $d\equiv3\pmod6$, we leave the question of minimal~$t$ open.
%% do we?

%=============================================================================
\section{Construction of designs}\label{sec:designs}
%=============================================================================
Next we describe a construction of projective $t$-designs for arbitrary $t$ based on weighted Clifford orbits.
As in \cref{sec:algebra,sec:moments}, we assume that $d$ is prime.

First, we derive expressions for the average tensor powers of the Clifford orbits of arbitrary states.
For any pure state $\ket\Psi$, the average $\EE_{U\in \Cliff(n,d)}[(U\ket\Psi\bra\Psi U^\dagger)^{\ot t}]$ commutes with $\Cliff(n,d)^{\ot t}$.
By \cref{thm:commutant}, and assuming that $n\geq t-1$, it can therefore be expressed as
\begin{align}\label{eq:average clifford orbit}
  \EE_{U\in \Cliff(n,d)}\left[\left(U\ket\Psi\bra\Psi U^\dagger\right)^{\ot t}\right] = \sum_{T\in \Sigma_{t,t}(d)} \alpha'_T R(T)
\end{align}
for some $\alpha_T' \in \CC$.
Not all of the $\alpha'_T$ are independent.
This is because each $(U\ket\Psi\bra\Psi U^\dagger)^{\ot t}$ is invariant under the action of $S_t \times S_t$ (acting from the left and from right), and also under taking the conjugate transpose.
This motivates the following definition:

\begin{dfn}[Equivalence relation $\sim_S$]
We define an \emph{equivalence relation $\sim_S$} on $\Sigma_{t,t}(d)$ in the following way:
$T \sim_S T'$ if and only if there exist $\pi$, $\pi' \in S_t$ such that $T'=\pi T \pi'$ or $T'=\pi T^t \pi'$, where the transposed subspace $T^t$ is defined by~$T^t=\{(\vec y,\vec x)\,:\, (\vec x,\vec y)\in T \}$.
\end{dfn}

\noindent We correspondingly decompose $\Sigma_{t,t}(d)$ into equivalence classes:
\begin{align*}
  \Sigma_{t,t}(d) = \bigcup_{i=1}^{M_{t,d}}\mathcal F_{t,i}(d)
\end{align*}
For convenience, we choose $\mathcal F_{t,1}(d)$ to be the set of subspaces corresponding to the permution group~$S_t$ (these form a single equivalence class).
We also define
\begin{align*}
  \mathcal R_i\coloneqq\sum_{T \in \mathcal F_{t,i}(d)} R(T).
\end{align*}
We note that the operators $\mathcal R_i$ are Hermitian and linearly independent.

Since the $R(T)$ are linearly independent and $R(T)^\dagger = R(T^t)$, it follows that the coefficients $\alpha'_T$ in \cref{eq:average clifford orbit} must be the same for the elements of each equivalence class.
Thus,
\begin{align}\label{eq:t-th moment of clifford orbit}
  \EE_{U\in \Cliff(n,d)}\left[\left(U\ket\Psi\bra\Psi U^\dagger\right)^{\ot t}\right] = \sum_{i=1}^{M_{t,d}} \alpha_i\mathcal R_i
\end{align}
for some coefficients $\alpha_i$.
Note that $\alpha_i \in \RR$ because the $\mathcal R_i$ are Hermitian.

\restateThmOrbitDesign
Importantly, the number~$M_{t,d}$ of Clifford orbit is independent of $n$, the number of qudits.
\begin{proof}
We start the proof by taking an arbitrary finite $t$-design given by an ensemble $\{p_j,\Psi_j\}_{j=1}^K$.
Such designs exist (see for example the early work~\cite{seymour1984averaging}), but $K$ can be very large.
If we replace each $\Psi_j$ by a random element in its Clifford orbit then the resulting ensemble still forms a projective $t$-design.
This means that
\begin{align*}
  \EE_{j \sim p} \EE_{U\in\Cliff(n,d)}\left[ (U \ket{\Psi_j}\bra{\Psi_j} U^\dagger)^{\ot t} \right] \propto \mathcal R_1.
\end{align*}
Thus, if we define $\alpha_i^{(j)}$ as the coefficient of $\mathcal R_i$ in the Clifford average of the fiducial state~$\Psi_j$,
\begin{align*}
  \EE_{U\in \Cliff(n,d)}\left[\left(U\ket{\Psi_j}\bra{\Psi_j} U^\dagger\right)^{\ot t}\right] = \sum_{i=1}^{M_{t,d}} \alpha_i^{(j)} \mathcal R_i,
\end{align*}
then
\begin{align}\label{eq:design conditions}
  \EE_{j \sim p}\left[\alpha_i^{(j)}\right] = 0 \quad \text{ for } i=2,\dots,M_{t,d}.
\end{align}
Conversely, if $\{p_j\}$ is an arbitrary probability distribution that satisfies \cref{eq:design conditions} then the ensemble obtained by first choosing a random fiducial state  according to this distribution and then a random state in its Clifford orbit is a projective $t$-design.
We will now explain how to modify the probabilities~$p_j$ step by step, setting more and more probabilities to zero while ensuring that \cref{eq:design conditions} continues to hold -- until all but~$M_{t,d}$ of them are zero.
Without loss of generality, assume that $p_1>0$.

Each step proceeds as follows:
Suppose that there exist indices $2\leq j_1 < j_2 < \dots < j_{M_{t,d}}\leq K$ such that $p_{j_m} > 0$ for all $m=1,\dots,M_{t,d}$.
(If no such indices exist then we are done.)
Consider the linear system
\begin{align*}
  \sum_{m=1}^{M_{t,d}} q_m \alpha_i^{(j_m)} = 0 \quad \text{ for } i=2,\dots,M_{t,d}
\end{align*}
in the indeterminates $\{q_m\}_{m=1}^{M_{t,d}}$.
This system is real, homogeneous, and underconstrained, so there always exists a nontrivial real solution $q \in \RR^{M_{t,d}}$.
We can also assume that some component of $q$ is positive (otherwise replace $q$ by $-q$).
Now consider $p_{j_m} - x q_m$ for $x\in\RR$.
At $x=0$, all $p_{j_m}$ are strictly positive.
At some critical~$x=x_c$, one of the values $p_{j_m} - x_c q_m$ becomes zero, while all other ones are still non-negative.
Thus, if we modify the probabilities $p_{j,m}$ by the rule
\begin{align*}
  p_{j_m} \mapsto p_{j_m} - x_c q_m \quad \text{ for } m=1,\dots,M_{t,d}
\end{align*}
then it still holds true that $\sum_{j=1}^{M_{t,d}} p_j \alpha_i^{(j)} = 0$, but there is now at least one additional zero among the $p_{j,m}$.
This continues to hold if we further normalize the $\{p_j\}$ to be a probability distribution, i.e.,
\begin{align*}
  p_j \mapsto \frac {p_j} {\sum_{j'} p_{j'}} \quad \text{ for } j=1,\dots,K,
\end{align*}
which is always possible since $p_1>0$.
Thus, we obtain a probability distribution $\{p_j\}$ with strictly smaller support satisfying \cref{eq:design conditions} and $p_1>0$.

We can repeat this process until there are at most $M_{t,d}-1$ nonzero probabilities among the $\{p_j\}_{j=2}^K$.
By including $p_1$, we arrive at an ensemble of at most $M_{t,d}$ fiducial vectors.
The corresponding probabilities satisfy \cref{eq:design conditions}, which is necessary and sufficient for the ensemble of Clifford orbits to be a design.
This completes the proof.
\end{proof}

\begin{rem}
A simple upper bound for $M_{t,d}$ is $\lvert\Sigma_{t,t}(d)\rvert=\prod_{k=0}^{t-2} (d^k+1)$ from \cref{thm:sigma size}.
However, in general this is a rather pessimistic estimate.
For example, consider $d=t=3$.
Then, $\lvert\Sigma_{t,t}(d)\rvert=8$, while there are just $M_{t,d}=2$ equivalence classes, as follows from \cref{eq:cosets 3 3}.
One of them is the set of permutations $S_3$, with $6$ elements, and the other one has 2 elements.
% \begin{align*}
% \left[
% \begin{array}{ccc|ccc}
% 1&2&0 & 1&2&0 \\
% \hline
% 0&0&0 & 1&1&1 \\
% \hline
% 1&1&1 & 0&0&0
% \end{array}\right].
% \end{align*}
% This class has $2$ elements.
% In total, there are $6+2=8$ Lagrangian subspaces as expected.

For $d=3$ and $t=4$, $\lvert\Sigma_{t,t}(d)\rvert=80$, while $M_{t,d}=3$.
Again, one of the equivalence classes is the permutation group with $4!=24$ elements.
The second equivalence class is the class of the \emph{anti-permutations} as defined in \cref{eq:anti perm qudit}, which is represented by the row space of the matrix
\begin{align*}
\left[
\begin{array}{cccc|cccc}
1&2&2&2 & 1&0&0&0 \\
2&1&2&2 & 0&1&0&0 \\
2&2&1&2 & 0&0&1&0 \\
2&2&2&1 & 0&0&0&1
\end{array}\right]
\end{align*}
(the Lagrangian subspace corresponding to the qutrit anti-identity $\bar\id$).
This equivalence class again has $24$ elements.
The last equivalence class can be represented by
\begin{align*}
\left[
\begin{array}{cccc|cccc}
1&1&1&1 & 1&1&1&1 \\
0&1&2&0 & 0&2&1&0 \\
\hline
0&0&0&0 & 1&1&1&0 \\
\hline
1&1&1&0 & 0&0&0&0
\end{array}\right],
\end{align*}
and it has $32$ elements.
% In total, we have $24+24+32=80$ elements as expected.
\end{rem}

\begin{rem}
The criterion used in the proof of \cref{thm:orbit design} can also be used to determine fiducial states that generate a projective $t$-design.
For example, the Clifford orbit through a single fiducial state $\Psi$ forms a projective $t$-design if and only if the coefficients $\alpha_i$ in \cref{eq:t-th moment of clifford orbit} vanish for $i\neq1$.

Let us illustrate this strategy by showing that, for any $n\geq2$, there exists a qutrit state $\ket\psi\in\CC^3$ such that the Clifford orbit of $\Psi = \psi^{\ot n}$ forms a projective 3-design.
We note that this state \emph{cannot} be a stabilizer state, since we know that the ensemble of qutrit stabilizer states does not form a $3$-design (\cref{rem:known design stuff})!
Instead of with $\mathcal R_1$ and $\mathcal R_2$, we will work with their multiples $\Pi^{\text{sym}}_3 \propto \mathcal R_1$ and $P_+ \coloneqq \Pi^{\text{sym}}_3 P \Pi^{\text{sym}}_3 \propto \mathcal R_2$, where $P$ is the projector defined in \cref{eq:projector P third moment qutrits}.
Now consider the third moment
\begin{align*}
  \rho_3 \coloneqq \EE_{U\in\Cliff(n,3)}\left[(U \ket\Psi\bra\Psi U^{\dagger})^{\ot3}\right],
\end{align*}
and expand it as
\begin{align*}
  \rho_3 = \alpha(\psi) \, \Pi^{\text{sym}}_3 + \beta(\psi) \, P_+
\end{align*}
for coefficients $\alpha(\psi)$, $\beta(\psi)\in\RR$ which depend on the choice of fiducial state.
We wish to argue that for every~$n$ there exists a single-qutrit state~$\ket\psi$ such that $\beta(\psi) = 0$.
For this, we note that the coefficients can be computed as follows:
\begin{align*}
  1
&= \tr[\rho_3]
= \alpha(\psi) \tr[\Pi^{\text{sym}}_3] + \beta(\psi) \tr[P_+] \\
\braket{\psi^{\ot3}|r(T)|\psi^{\ot3}}^n
&= \tr[R(T) \rho_3]
% = 3^n \tr[P \rho_3]
= 3^n \alpha(\psi) \, \tr[P_+] + 3^n \beta(\psi) \, \tr[P_+].
\end{align*}
It follows that $\beta(\psi)=0$ if and only if
\begin{align*}
  % \alpha(\psi) = \frac{1 - \beta(\psi) \tr[P_+]}{\tr[\Pi^{\text{sym}}_3]} \\
  \braket{\psi^{\ot3}|r(T)|\psi^{\ot3}}^n
= 3^n \frac{\tr[P_+]}{\tr[\Pi^{\text{sym}}_3]}
= \frac{3}{3^n+2},
\end{align*}
where we used \cref{eq:dims third moment qutrits}.
Thus, the Clifford orbit through $\psi^{\ot n}$ forms a projective 3-design if and only if
\begin{align}\label{eq:qutrit three design condition}
  \braket{\psi^{\ot3}|r(T)|\psi^{\ot3}} = \left(\frac{3}{3^n+2}\right)^{1/n} \in \left[\tfrac13,\tfrac35\right]
\end{align}
But the the left-hand side is equal to one if $\psi$ is a stabilizer state, e.g., $\ket\psi=\ket0$ (\cref{eq:not quite eigenvectors}), while is vanishes for, e.g., the non-stabilizer state $\ket\psi = \frac1{\sqrt2}\left( \ket0 - \ket 1 \right)$.
% Using \cref{eq:P css,eq:cosets 3 3}:
% \begin{align*}
%   \left( \bra 0 - \bra 1 \right)^{\ot 3} r(T) \left( \ket 0 - \ket 1 \right)^{\ot 3}
% % = \sum_{x=0}^2 \left( \bra 0 - \bra 1 \right)^{\ot 3} \left( \sum_y \ket{y+x,y-x,y} \right) \left( \sum_z \bra{z+x,z-x,z} \right) \left( \ket 0 - \ket 1 \right)^{\ot 3}
% = \sum_{x=0}^2 \left\lvert \left( \sum_z \bra{z+x,z-x,z} \right) \left( \ket 0 - \ket 1 \right)^{\ot 3} \right\rvert^2
% \end{align*}
% Now:
% \begin{align*}
% &\quad \sum_z \bra{z+x,z-x,z} \left( \ket 0 - \ket 1 \right)^{\ot 3} \\
% &= \sum_z \bra{z+x} \left( \ket0-\ket1 \right) \bra{z-x} \left( \ket0-\ket1 \right) \bra z \left( \ket0-\ket1 \right) \\
% &= \sum_{z=0,1} (-1)^z \delta_{z+x \in \{0,1\}} (-1)^{z+x} \delta_{z-x\in\{0,1\}} (-1)^{z-x} \\
% &= \sum_{z=0,1} (-1)^z \delta_{z+x \in \{0,1\}} \delta_{z-x\in\{0,1\}} \\
% &= \begin{cases}
% \sum_{z=0,1} (-1)^z \delta_{z \in \{0,1\}} \delta_{z\in\{0,1\}} & x=0 \\
% \sum_{z=0,1} (-1)^z \delta_{z+1 \in \{0,1\}} \delta_{z-1\in\{0,1\}} & x=1 \\
% \sum_{z=0,1} (-1)^z \delta_{z-1 \in \{0,1\}} \delta_{z+1\in\{0,1\}} & x=2
% \end{cases} \\
% &= 0
% \end{cases}
% \end{align*}
By continuity it follows that there always exists a single-qutrit state $\ket\psi$ satisfying \cref{eq:qutrit three design condition}.
It is easy to find such an $\ket\psi$ explicitly, e.g., by considering the one-parameter family of states $\ket{\psi(\theta)} = \cos(\theta) \ket0 - \sin(\theta) \ket1$ and solving \cref{eq:qutrit three design condition} for $\theta\in[0,\frac\pi2]$.
\end{rem}

%=============================================================================
\section{De Finetti theorems for stabilizer symmetries}\label{sec:de finetti}
%=============================================================================
In this section we establish a direct connection between our results on stabilizer testing and the celebrated \emph{quantum de Finetti theorems}, which play an important role in characterizing entanglement and correlations in quantum states with permutation symmetry (cf.~discussion in \cref{subsec:de finetti intro}).

We first recall the finite quantum de Finetti theorem from~\cite{christandl2007one}.
Let $\rho$ be a quantum state on~$(\CC^\ell)^{\ot t}$ that commutes with all permutations (i.e., $[r_\pi,\rho]=0$ for all $\pi\in S_t$).
Then there exists a probability measure~$\mu$ on the space of mixed states on $\CC^\ell$ such that
\begin{align}\label{eq:christandl mixed de finetti}
  \frac12 \left\lVert \rho_{1\dots{}s} - \int d\mu(\sigma) \sigma^{\ot s} \right\rVert_1 \leq 2\ell^2\frac{s}{t}.
\end{align}
Since any quantum state that commutes with permutations admits a purification on the symmetric subspace, \cref{eq:christandl mixed de finetti} follows directly from a similar result for the symmetric subspace, namely, that for every $\ket\Psi\in\Sym^t(\CC^\ell)$ there exists a probability measure~$\mu$ on \emph{pure states} on $\CC^\ell$ such that
\begin{align}\label{eq:christandl pure de finetti}
  \frac12 \left\lVert \Psi_{1\dots{}s} - \int d\mu(\phi) \phi^{\ot s} \right\rVert_1 \leq 2\ell\frac{s}{t}.
\end{align}

In this section, we prove de Finetti theorems adapted to stabilizer states.
The key idea is to extend the permutation symmetry to invariance under a larger group:
\begin{enumerate}
\item the stochastic orthogonal group~$O_t(d)$ (for qudits in any prime dimension~$d$), or
\item the group generated by the permutations and the anti-identity~\eqref{eq:anti identity six} (for qubits).
\end{enumerate}
These symmetries are natural since they are carried by the tensor powers of any stabilizer state, as we proved in \cref{eq:eigenvectors}.

In both cases, our theorems show that the reduced density matrices are close to convex combinations of \emph{tensor powers of stabilizer states}.
In the first case, we find that the reduced state is in fact \emph{exponentially} (in the number of traced out systems) close to a state of this form, which is a much stronger guarantee than provided by the finite de Finetti theorems of \cref{eq:christandl mixed de finetti,eq:christandl pure de finetti} (cf.~\cite{renner2007symmetry,koenig2009most}).
In the second case, we obtain power law convergence but the symmetry requirements are drastically reduced.
We establish our results first for pure states (\cref{subsec:de finetti qubits,subsec:de finetti qudits}) and then extend them by a standard purification argument to mixed states (\cref{subsec:de finetti mixed}).

%-----------------------------------------------------------------------------
\subsection{Exponential stabilizer de Finetti theorem}\label{subsec:de finetti qudits}
%-----------------------------------------------------------------------------
Let $d$ be an arbitrary prime.
We start with the observation that, for any two distinct stabilizer states,
\begin{align}\label{eq:stab neq overlap}
  \lvert\braket{S|S'}\rvert^2\leq \frac1d
\end{align}
(this can be seen from, e.g., \cref{eq:characteristic function stabilizer}).
It follows that, for fixed~$d$ and~$n$, the stabilizer tensor powers~$\ket S^{\otimes t}$ approach orthonormality as $t\to\infty$.
The following lemma makes this precise.

\begin{lem}\label{lem:gram matrix}
  Let $d$ be a prime and $n,t\geq1$.
  Consider the Gram matrix $G_{S,S'} = \braket{S|S'}^t$, where $S,S'\in\Stab(n,d)$.
  If
  \begin{align*}
    \eps \coloneqq d^{\frac12((n+2)^2-t)} < \frac12
  \end{align*}
  then the following holds:
  \begin{enumerate}
  \item
  The Gram matrix is $\eps$-close to the identity matrix in operator norm: $\lVert G - I \rVert_\infty \leq \eps$.
  In particular, the stabilizer tensor powers $\ket S^{\ot t}$ are linearly independent.
  \item
  The nonzero eigenvalues of $Q \coloneqq \sum_S \ket S^{\ot t}\bra S^{\ot t}$ and its pseudoinverse $Q^{+}$ lie in the interval $1\pm2\eps$.
  \item
  The vectors $(Q^+)^{1/2} \ket{S}^{\otimes t}$ for $S\in\Stab(n,d)$ are orthonormal.
  \end{enumerate}
\end{lem}
\begin{proof}
  1. The first claim follows directly from the element-wise bound~\eqref{eq:stab neq overlap}:
  \[ 
  \lVert G - I \rVert_\infty 
  \leq
  \lVert G - I \rVert_{\ell_2}
  \leq 
  \lVert G - I \rVert_{\ell_\infty}\,\lvert\Stab(n,d)\rvert
  \leq \left( \max _{S\neq S'} \lvert\braket{S|S'}\rvert^t \right)\, \lvert\Stab(n,d)\rvert
  \leq %d^{-t/2} \, 
  d^{\frac12 ((n+2)^2-t)}=\eps,\]
  where we used the bound
  \begin{align}\label{eq:stab cardinality}
    \lvert\Stab(n,d)\rvert=d^n \prod_{i=1}^n (d^i+1) \leq d^{(n+2)^2/2}.
  \end{align}
  The cardinality of the set of stabilizer states has been computed in \cite[Prop.~2]{aaronson2004improved} for $d=2$ and in~\cite[Cor.~21]{gross2006hudson} for odd~$d$.
  % Proof: n=1 and n=2 by hand, then by induction in n
  Since $\eps<1$, the statement about the Gram matrix implies that the stabilizer tensor powers are linearly independent.

  2. Now define
  \begin{align*}
    H = \sum_{S\in\Stab(n,d)} \ket S^{\ot t}\bra{e_S},
  \end{align*}
  where $\ket{e_S}$ denotes an orthonormal basis labeled by the set of stabilizer states $\Stab(n,d)$.
  Then,
  \begin{align*}
    G = H^\dagger H
    \qquad\text{ and }\qquad
    Q = H H^\dagger,
  \end{align*}
  and thus the nonzero eigenvalues of $G$ and $Q$ are both identical (to the squared singular values of $H$).
  By part 1, the eigenvalues of $G$ lie in the interval $1\pm\eps$, hence the same is true for the nonzero eigenvalues of $Q$.
  Since we assumed that $\eps<1/2$, it follows that the nonzero eigenvalues of the pseudoinverse $Q^+$ are in the interval $1\pm2\eps$.
  This establishes the second claim.

  3. By the first claim, the stabilizer tensor powers are linearly independent.
  On the other hand,
  \begin{align*}
    \ket S^{\otimes t}
    = Q Q^{+}\ket S^{\otimes t}
    =  \sum_{S'} \ket{S'}^{\otimes t} \bra{S'}^{\otimes t} Q^{+} \ket{S}^{\otimes t}.
  \end{align*}
  Thus, the linear independence implies that the vectors $(Q^+)^{1/2} \ket S^{\otimes t}$ are orthonormal.
\end{proof}

\begin{thm}[Pure-state exponential stabilizer de Finetti theorem]\label{thm:expodef}
  Let $d$ be a prime and~$\ket\Psi\in\Sym^t((\CC^d)^{\ot n})$ a pure quantum state that is left invariant by the action of $O_t(d)$.
  Let $1\leq s\leq t$.
  Then there exists a probability distribution~$p$ on $\Stab(n,d)$, the set of pure stabilizer states of $n$~qudits, such that
  \begin{align*}
  \frac12\left\lVert \Psi_{1\dots{}s} - \sum_S p(S) \ket S^{\ot s}\bra S^{\ot s} \right\rVert_1
  % \leq 2 d^{-\frac12((t-s)-(2n+2)^2)}.
  \leq 2 d^{\frac12(n+2)^2} d^{-\frac12(t-s)}.
  \end{align*}
\end{thm}
\begin{proof}
  By assumption, $\Pi^{\min}_t \ket\Psi = \ket\Psi$, where $\Pi^{\min}_t$ is the minimal projector from \cref{eq:minimal projector}.
  \Cref{thm:minimal test} shows that $\ket\Psi$ is contained in the span of stabilizer tensor powers, i.e.,
  \begin{align*}
    \ket\Psi = \sum_{S\in\Stab(n,d)} \alpha_S \ket S^{\ot t}
  \end{align*}
  for certain coefficients $\alpha_S\in\CC$.
  We now use the third and second claim of \cref{lem:gram matrix} to see that
  % \begin{align*}
  %   \braket{S^{\ot t} | Q^+ | \Psi}
  %   = \sum_{S'} \alpha_{S'} \braket{S^{\otimes t} | Q^+ | {S'}^{\otimes t}}
  %   = \alpha_S
  % \end{align*}
  \begin{align}\label{eq:alpha coeff interval}
    \sum_S \lvert\alpha_S\rvert^2 = \lVert (Q^+)^{1/2} \ket\Psi\rVert^2 \in [1-2\eps,1+2\eps]
  \end{align}
  where $\eps \coloneqq d^{\frac12((n+2)^2-t)}$.
  Here we have assumed that $\eps<1/2$, for otherwise the statement of the theorem is vacuous.
  We now compute the partial trace over all but~$s$ subsystems:
  \begin{align*}
    \Psi_{1\dots{}s}
  = \sum_S \lvert\alpha_S\rvert^2 \ket S^{\ot s}\bra S^{\ot s} \;+\; \sum_{S\neq S'} \alpha_S \overline\alpha_{S'} \ket S^{\ot s} \bra{S'}^{\ot s} \braket{S'|S}^{t-s}.
  \end{align*}
  The norm of the cross terms is small:
  \begin{align*}
  &\quad \left\lVert \sum_{S\neq S'} \alpha_s \overline\alpha_{S'} \ket S^{\ot s} \bra{S'}^{\ot s} \braket{S'|S}^{t-s} \right\rVert_1
  \leq \sum_{S\neq S'} \lvert\alpha_S\rvert \lvert\alpha_{S'}\rvert d^{-(t-s)/2}
  \leq \left( \sum_S \lvert\alpha_S\rvert \right)^2 d^{-(t-s)/2} \\
  &\leq \left( \sum_S \lvert\alpha_S\rvert^2 \right) d^{(n+2)^2/2} d^{-(t-s)/2}
  \leq \left( 1 + 2\eps\right) d^{(n+2)^2/2} d^{-(t-s)/2}
  \leq 2 d^{(n+2)^2/2} d^{-(t-s)/2};
  \end{align*}
  the first inequality uses \cref{eq:stab neq overlap},
  the third inequality is \cref{eq:stab cardinality},
  the fourth inequality is the upper bound in \cref{eq:alpha coeff interval},
  and the last step uses that $\eps< 1/2$.
  Finally, define $p(S) \coloneqq \lvert\alpha_S\rvert^2 / \sum_{S'} \lvert\alpha_{S'}\rvert^2$ (the denominator is positive by the lower bound in \cref{eq:alpha coeff interval} and $\eps<1/2$).
  Then:
  \begin{align*}
  &\quad\left\lVert \Psi_{1\dots{}s} - \sum_S p(S) \ket S^{\ot s}\bra S^{\ot s} \right\rVert_1
  \leq \left\lVert \Psi_{1\dots{}s} - \sum_S \lvert\alpha_S\rvert^2 \ket S^{\ot s}\bra S^{\ot s} \right\rVert_1
  + \sum_S \left\lvert \lvert\alpha_S \rvert^2 - p(S) \right\rvert \\
  % &\leq 2 d^{(n+2)^2/2} d^{-(t-s)/2} + \sum_S \left\lvert \lvert\alpha_S \rvert^2 - p(S) \right\rvert
  &\leq 2 d^{(n+2)^2/2} d^{-(t-s)/2} + \left\lvert 1 - \sum_{S'} \lvert\alpha_{S'}\rvert^2 \right\rvert
  \leq 2 d^{(n+2)^2/2} d^{-(t-s)/2} + 2\eps \\
  &\leq 4 d^{(n+2)^2/2} d^{-(t-s)/2}.
  \qedhere
  \end{align*}
\end{proof}
%-----------------------------------------------------------------------------
\subsection{Stabilizer de Finetti theorem for the anti-identity}\label{subsec:de finetti qubits}
%-----------------------------------------------------------------------------

We now prove a stabilizer de Finetti theorem with reduced symmetry requirements.
For concreteness, we restrict to the multi-qubit case ($d=2$) and to tensor powers that are multiples of six.
Neither restriction is essential.
The following theorem shows that the reduced states of an arbitrary permutation-symmetric quantum state that is invariant under the anti-identity operator~$V=R(\bar\id)$ from~\cref{eq:anti identity six}, but not necessarily under other stochastic isometries, are well-approximated by convex mixtures of tensor powers of stabilizer states.

\begin{thm}[Pure-state stabilizer de Finetti theorem for the anti-identity]\label{thm:antidef}
Let $\ket\Psi \in \Sym^t((\CC^2)^{\ot n})$ be a quantum state that is left invariant by the action of the anti-identity~\eqref{eq:anti identity six} on some (and hence every) subsystem consisting of six $n$-qubit blocks.
Let $s<t$ be a multiple of six.
Then there exists a probability distribution~$p$ on $\Stab(n,2)$, the set of pure stabilizer states of $n$~qubits, such that
\begin{align*}
\frac12\left\lVert \Psi_{1\dots{}s} - \sum_S p(S) \ket S^{\ot s}\bra S^{\ot s} \right\rVert_1
\leq 6 \sqrt{2^{n+1}} \sqrt{\frac st}.
\end{align*}
\end{thm}
\begin{proof}
  By the ordinary finite quantum de Finetti theorem~\eqref{eq:christandl pure de finetti}, there exists a probability measure~$d\mu(\phi)$ on the set of pure states such that
  \begin{align}\label{eq:from ordinary de finetti}
    \frac12 \left\lVert \Psi_{1\dots{}s} - \int d\mu(\phi) \phi^{\ot s} \right\rVert_1 \leq 2^{n+1} \frac st.
  \end{align}
  Let $\Pi_n=(I+R(\bar\id))/2$ denote the projector onto the $+1$-eigenspace of the $6\times 6$-anti identity for $n$ qubits.
  By assumption, $(\Pi_n^{\ot(s/6)} \ot I^{\ot(t-s)}) \ket\Psi=\ket\Psi$, and hence $\tr[\Psi_{1\dots{}s} \Pi_n^{\ot s/6}]=1$.
  Since the trace distance satisfies~$\frac12\lVert \rho - \sigma\rVert_1 = \max_{0\leq Q\leq I} \tr[(\rho-\sigma)Q]$, it follows that
  \begin{align*}
  \int d\mu(\phi) \left( 1- \tr\left[\Pi_n^{\otimes(s/6)} \phi^{\ot s}\right] \right)
  \leq 2^{n+1} \frac st.
  \end{align*}
  Now recall from \cref{eq:accepting povm} that the accepting POVM element for qubit stabilizer testing is given by $\Pi_\text{accept} = \frac12 \left( I + U \right)$, where $U = V(I^{\ot4} \ot \FF)$, where $\FF=R((1 2))$ is the operator that swaps two blocks of $n$~qubits (see discussion above \cref{eq:V for qubits}).
  Since tensor powers of pure states are permutation-symmetric,
  \begin{align}\label{eq:anti de finetti from trace dist}
  \int d\mu(\phi) \left( 1- \tr\left[\Pi_\text{accept} \phi^{\ot6}\right]^{s/6} \right)
  % = \int d\mu(\phi) \left( 1- \tr\left[\Pi_\text{accept}^{\otimes(s/6)} \phi^{\ot s}\right] \right)
  \leq 2^{n+1} \frac st
  \end{align}
  According to \cref{thm:main qubits,eq:saner quantitative bound} and using \cref{lem:technical inequality} below, for each pure state~$\phi$ there exists a pure stabilizer state~$S_\phi$ such that
  \begin{align*}
    \lvert\braket{S_\phi|\phi}\rvert^{2(s/6)}
  % \geq \left( 4 \tr\left[\Pi_\text{accept} \phi^{\ot6}\right] - 3 \right)^{s/6}
  \geq 4 \tr\left[\Pi_\text{accept} \phi^{\ot6}\right]^{s/6} - 3.
  \end{align*}
  Using the estimate~$1-p^6 \leq 6(1-p)$, which holds for all~$p \in [0,1]$, we obtain
  \begin{align*}
    1 - \lvert\braket{S_\phi|\phi}\rvert^{2s}
  \leq 6\left( 1 - \lvert\braket{S_\phi|\phi}\rvert^{2(s/6)} \right)
  \leq 24 \left(1 - \tr\left[\Pi_\text{accept} \phi^{\ot 6}\right]^{s/6} \right).
  \end{align*}
  Combining this estimate with \cref{eq:anti de finetti from trace dist}, we get
  \begin{align*}
  \int d\mu(\phi) \left( 1- \lvert\braket{S_\phi|\phi}\rvert^{2s} \right)
  \leq 24 \cdot 2^{n+1} \frac st.
  \end{align*}
  It follows that replacing each pure state~$\phi$ by the nearby stabilizer state~$S_\phi$ incurs only a small error:
  \begin{align*}
  &\quad \frac12\left\lVert \int d\mu(\phi) \phi^{\ot s} - \int d\mu(\phi) S_\phi^{\ot s} \right\rVert_1
  \leq \int d\mu(\phi) \frac12 \left\lVert \phi^{\ot s} -  S_\phi^{\ot s} \right\rVert_1
  \leq \int d\mu(\phi) \sqrt{1 - \lvert\braket{\phi|S_\phi}^{2s}} \\
  &\leq \sqrt{\int d\mu(\phi) \left( 1 - \lvert\braket{\phi|S_\phi}^{2s} \right)}
  \leq \sqrt{24 \cdot 2^{n+1} \frac st},
  \end{align*}
  where we have used the triangle inequality, the relation between the trace distance and the fidelity between pure states, and the concavity of the square root.
  Together with \cref{eq:from ordinary de finetti}, we obtain
  \begin{align*}
    \frac12\left\lVert \Psi_{1\dots{}s} - \int d\mu(\phi) S_\phi^{\ot s} \right\rVert_1
  \leq 2^{n+1} \frac st + \sqrt{24\cdot 2^{n+1} \frac st}
  % \leq 2^{n+1} \frac st + 2 \sqrt{2^{n+1} \frac st}
  \leq 6 \sqrt{2^{n+1}} \sqrt{\frac st}
  \end{align*}
  where we have assumed without loss of generality that~$2^{n+1} \frac st\leq1$ (otherwise, the right-hand side is larger than one so the resulting bound is trivially true).
\end{proof}

\begin{lem}\label{lem:technical inequality}
The following bound holds for all $k\geq 1$ and $p\in[0,1]$ such that the right-hand side is non-negative:
  \begin{equation*}
    (4p-3)^k \geq 4p^k-3.
  \end{equation*}
\end{lem}
\begin{proof}
  We will prove the inequality for all~$k\geq1$ and~$p\in[\frac34,1]$.
  For this, note that the two expressions coincide for~$p=1$ and that the derivative of their difference is negative for all~$p\in[\frac34,1]$.
  Indeed,
  \begin{align*}
    \partial_p \left( (4p-3)^k - (4p^k-3) \right) \leq 0
    \quad\Leftrightarrow\quad
    (4p-3)^{k-1} \leq p^{k-1}.
  \end{align*}
  In the interval that we are considering, $0 \leq 4p-3 \leq p$, so the right-hand side condition holds.
\end{proof}

%-----------------------------------------------------------------------------
\subsection{Extension to mixed states}\label{subsec:de finetti mixed}
%-----------------------------------------------------------------------------
In this section, we extend~\cref{thm:antidef,thm:expodef} to the case of mixed density matrices.
This is done using a standard purification argument as used to derive the ordinary quantum de Finetti theorem for mixed states from the version for pure states (i.e., \cref{eq:christandl mixed de finetti} from \cref{eq:christandl pure de finetti}).
For the next lemma, recall the vectorization operation from \cref{dfn:vectorization}.

\begin{lem}[Purification and symmetries]\label{lem:gspl}
  Let $\rho$ be positive semi-definite, $\ket\Psi = \vecmap(\rho^{1/2})$ its standard purification, and $O$ a unitary with real matrix elements in the computational basis.
  Then the following conditions are equivalent:
  \begin{enumerate}
  \item\label{it:puri 1} $(O\ot O)\ket\Psi=\ket\Psi$.
  \item\label{it:puri 2} $[\rho,O]=0$.
  \end{enumerate}
\end{lem}

\begin{proof}
  We observe that
  \begin{equation*}
    (O \otimes O) \ket\Psi
    =
    \left(O \otimes (O^{-1})^T\right) \ket\Psi
    =
    \vecmap (O \rho^{1/2} O^{-1} ).
  \end{equation*}
  Thus, condition~\ref{it:puri 1} is equivalent to $O\rho^{1/2}O^{-1} = \rho^{1/2}$.
  It follows that condition~\ref{it:puri 1} implies condition~\ref{it:puri 2} by squaring.
  Conversely, assuming condition~\ref{it:puri 2},
  \begin{equation*}
    \rho
    = O \rho O^{-1}
    = \bigl(O \rho^{1/2} O^{-1}\bigr) \bigl( O \rho^{1/2} O^{-1}\bigr).
  \end{equation*}
  Hence $O \rho^{1/2} O^{-1}$ is a positive semi-definite square root of $\rho$.
  Since such square roots are unique, this implies that $O \rho^{-1/2} O^{-1} = \rho^{1/2}$ and hence condition~\ref{it:puri 1}.
\end{proof}

Clearly, the operators~$R(O)$ for $O\in O_t(d)$ have real matrix elements in the computational basis.
In fact, they are given by permutation matrices in the computational basis, as can be seen from the formula given in \cref{eq:O_t action}.
Thus they satisfy the conditions of~\cref{lem:gspl}.
We use this now to extend our de Finetti theorems to mixed states.

\restateThmExpDeFinetti
\begin{proof}
Using \cref{lem:gspl}, we can find a purification $\ket\Psi \in ((\CC^d)^{\ot2n}))^{\ot t} \cong (\CC^d)^{\ot n})^{\ot t} \ot (\CC^d)^{\ot n})^{\ot t}$ of~$\rho$, which is invariant under the action of $O \in O_t(d)$.
(In particular, $\ket\Psi$ is an elemenrt of the symmetric subspace.)
Here we crucially use that the operators $R(O)$ for~$2n$~qudits are just the second tensor powers of the corresponding operators for $n$ qudits, as is clear from \cref{eq:O_t action}.
Thus we can apply \cref{thm:expodef} to the state $\ket\Psi$.
Since the local Hlibert space now contains $2n$ qudits, we obtain that there exists a probability distribution $p$ over pure stabilizer states on $(\CC^d)^{\ot 2n}$ such that
\begin{align*}
  \frac12\left\lVert \Psi_{1\dots{}s} - \sum_S p(S) \, \ket S\bra S^{\ot s} \right\rVert_1
\leq 2 d^{\frac12(2n+2)^2} d^{-\frac12(t-s)}.
\end{align*}
Taking the partial trace over the purifying systems does not increase the trace distance.
Since reduced density matrices of pure stabilizer states are mixed stabilizer states, we obtain the result.
\end{proof}

\noindent
The very same argument yields the following version of \cref{thm:antidef} for mixed states:

\restateThmAntiDeFinetti

%=============================================================================
\section{Robust Hudson theorem}\label{sec:hudson}
%=============================================================================
The methods developed in \cref{sec:stabilizer testing} also allow us to prove a robust version of the finite-dimensional Hudson theorem.
Recall that from \cref{eq:wigner function odd stabilizer} that, for odd $d$, the Wigner function of a pure stabilizer state is necessarily nonnegative.
Hudson theorem states that, for pure states, this condition is also sufficient, i.e., the Wigner function of a pure quantum state is non-negative if and only if the state is a stabilizer state~\cite{gross2006hudson}.
We will show in \cref{thm:robust hudson} that if the \emph{Wigner} or \emph{sum-negativity}
\begin{align*}
  \sn(\psi)
\coloneqq \sum_{\vec x : w_\psi(\vec x)<0} \lvert w_\psi(\vec x) \rvert
= \frac12 \left( \sum_{\vec x} \lvert w_\psi(\vec x) \rvert - 1 \right).
\end{align*}
is small then the state is close to a stabilizer state.

The Wigner negativity is immediately related to the \emph{mana}~$\mathcal M(\psi)=\log (2\sn(\psi)+1)$, a monotone that plays an important role in the resource theory of stabilizer computation~\cite{gottesman1997stabilizer}.
Throughout this section we assume that $d$ is odd, so that the Wigner function is well-behaved (cf. \cref{subsec:phase space}).

%-----------------------------------------------------------------------------
\subsection{Exact Hudson theorem}
%-----------------------------------------------------------------------------
We first give a new and succinct proof of the finite-dimensional Hudson theorem.
For pure states, $1=\tr\psi^2
% =\sum_{\vec x,\vec y} w_\psi (x) w_\psi (y) \tr \left( A_{\vec x} A_{\vec y}\right)
=\sum_{\vec x} d^n w_\psi(\vec x)^2$.
Thus we can define a probability distribution based on the Wigner function,
\begin{align*}
  q_\psi(\vec x) = d^n w_\psi(\vec x)^2,
\end{align*}
similar to the $p_\psi$ distribution that we defined in \cref{eq:p distribution} via the characteristic function.
Note that $0\leq q_\psi(\vec x)\leq d^{-n}$, since $\lvert w_\psi(\vec x)\rvert \leq d^{-n}$.

We now consider the sum of the absolute value of the Wigner function,
\[ \lVert\psi\rVert_W \coloneqq \sum_{\vec x} \lvert w_\psi(\vec x) \rvert= d^{-n/2} \sum_{\vec x} q_\psi(\vec x)^{1/2}. \]
It holds that $\lVert\psi\rVert_W \geq \sum_{\vec x} w_\psi(\vec x) = 1$, with equality if and only if $w_\psi(x)\geq0$ for all $x$.
By the H\"older inequality (with $p_1=p_2=p_3=3$, so $\sum_k 1/p_k=1$):
\begin{align*}
1 = \sum_{\vec x}  q_\psi(\vec x) = \sum_{\vec x} q_\psi(\vec x)^{1/6} q_\psi(\vec x)^{1/6} q_\psi(\vec x)^{2/3}
\leq \left( \sum_{\vec x} q_\psi({\vec x})^{1/2} \right)^{2/3} \left( \sum_{\vec x} q_\psi(\vec x)^2 \right)^{1/3}.
\end{align*}
Thus we obtain the following fundamental bound:
\begin{equation}\label{eq:average vs wigner}
\sum_{\vec x} q_\psi(\vec x)^2 \geq \frac 1 {d^n \lVert \psi \rVert_W^2}.
\end{equation}
Crucially, we can interpret the left-hand side as the average of the function $q_\psi$ with respect to the same probability distribution, $E_{\vec x\sim q_\psi} q_\psi(\vec x)$.
Now suppose that the Wigner function is nowhere negative, so that the bound simplifies to $\sum_{\vec x} q_\psi(\vec x)^2 \geq d^{-n}$.
But $q_\psi(\vec x)\leq d^{-n}$ for all $\vec x$, so we conclude that the function $q_\psi$ must be equal to $d^{-n}$ on its support.
In other words, $q_\psi(\vec x)$ is the uniform distribution on a subset of cardinality~$d^n$.
This gives a rather direct proof of the finite-dimensional Hudson theorem:

\begin{thm}[Finite-dimensional Hudson theorem,~\cite{gross2006hudson}]\label{thm:exact hudson}
Let $d$ be an odd integer and $\psi$ a pure quantum state of $n$ qu$d$its.
Then the Wigner function of $\psi$ is everywhere nonnegative, $w_\psi(\vec x)\geq0$, if and only if $\psi$ is a stabilizer state.
\end{thm}
\begin{proof}
In view of \cref{eq:wigner function odd stabilizer} we only need to show that if $w_\psi(\vec x)\geq0$ for all $\vec x$ then $\psi$ is a stabilizer state.
By the preceding discussion, we know that $w_\psi(\vec x) = d^{-n} \id_T(\vec x)$, where $\id_T$ denotes the indicator function of some subset $T\subseteq\mathcal V_n$ of cardinality $d^n$.
In other words, $\braket{\psi | A_{\vec x} | \psi} = \id_T(\vec x)$ and so $A_{\vec x}\ket\psi=\ket\psi$ for all $\vec x\in T$.

It remains to show that $T$ is of the form $T=\vec a+M$, where $M$ is a maximal isotropic subspace.
For this, consider any three points $\vec x,\vec y,\vec z\in T$ and use \cref{eq:three points}, which asserts that
$A_{\vec x} A_{\vec y} A_{\vec z} = \omega^{2 [\vec z-\vec x,\vec y-\vec x]} A_{\vec x-\vec y +\vec z}$.
Because $\ket\psi$ is an eigenvector of $A_{\vec x},A_{\vec y},A_{\vec z}$, with eigenvalue $+1$, we obtain
\begin{align*}
  1 = \braket{\psi|A_{\vec x} A_{\vec y} A_{\vec z}|\psi} = \omega^{2 [\vec z-\vec x,\vec y-\vec x]} \braket{\psi|A_{\vec x-\vec y +\vec z}|\psi}.
\end{align*}
But $A_{\vec x-\vec y +\vec z}$ is Hermitian, so this is impossible unless $[\vec z-\vec x,\vec y-\vec x]=0$.
Therefore, $T$ is the translate of an totally isotropic set $M$ of cardinality $d^n$.
Since the maximal size of a totally isotropic subspace is also $d^n$~\cite[App.~C]{gross2006hudson}, $T$ is necessarily a maximal isotropic subspace.
We conclude that $\psi$ is a stabilizer state.
\end{proof}

%-----------------------------------------------------------------------------
\subsection{Robust Hudson theorem}
%-----------------------------------------------------------------------------
To obtain a robust version of the Hudson theorem, we will, similarly as in our approach to stabilizer testing, combine \cref{eq:average vs wigner} with an uncertainty relation that generalizes the proof of \cref{thm:exact hudson}.

\restateLemPointOpUncertainty
\begin{proof}
Note that the assumption implies that
\begin{align*}
  \lVert A_{\vec x} \ket\psi - \ket\psi \rVert < \sqrt{2\left(1-\sqrt{1-\frac{1}{2d^2}}\right)} \leq \frac1d,
\end{align*}
and likewise for $\vec y$ and $\vec z$.
Thus we obtain the following inequalities:
% Using $\sqrt{2(1-\sqrt{1-1/2d^2})}\leq 1/d$, we obtain the following inequalities for $\vec x, \vec y, \vec z \in T$,
\begin{align*}
  \lVert A_{\vec x} \ket\psi - \ket\psi \rVert < \frac1d, \quad
  \lVert A_{\vec y} \ket\psi - \ket\psi \rVert < \frac1d, \quad
  \lVert A_{\vec z} \ket\psi - \ket\psi \rVert < \frac1d .
\end{align*}
As a consequence of the triangle inequality, and using $\lVert A_{\vec z} \rVert \leq 1$, along with \cref{eq:three points}, we obtain
\begin{align*}
\quad \lVert A_{\vec x - \vec y+\vec z} \ket\psi - \omega^{-2[\vec z-\vec x,\vec y-\vec x]}\ket\psi \rVert
= \lVert A_{\vec x} A_{\vec y} A_{\vec z} \ket\psi - \ket\psi \rVert < \frac3d.
% \\
% \leq \lVert A(x) A(y) A(0) \ket\psi - A(x) A(y) \ket\psi \rVert
% + \lVert A(x) A(y) \ket\psi - A(x) \ket\psi \rVert
% + \lVert A(x) \ket\psi - \ket\psi \rVert \\
% \leq \lVert A(x) \rVert \lVert A(y) \rVert \lVert A(0) \ket\psi - \ket\psi \rVert
% + \lVert A(x) \rVert \lVert A(y) \ket\psi - \ket\psi \rVert
% + \lVert A(x) \ket\psi - \ket\psi \rVert \\
% \leq \lVert A(0) \ket\psi - \ket\psi \rVert
% + \lVert A(y) \ket\psi - \ket\psi \rVert
% + \lVert A(x) \ket\psi - \ket\psi \rVert
%\leq 3\sqrt{2\delta} = \sqrt{18\delta}
\end{align*}
We can simply expand this relation and see that it is equivalent to \[1-\frac{9}{2d^2} < \bra \psi  A_{\vec x - \vec y+\vec z} \ket \psi \cos( 2[\vec z-\vec x, \vec y-\vec x] \frac{2\pi}d).\]
If $[\vec z-\vec x, \vec y-\vec x]\neq 0$, then
\begin{align*}
1-\frac{9}{2d^2} &< \bra \psi A_{\vec x - \vec y+\vec z} \ket \psi \cos( 2[\vec z-\vec x, \vec y-\vec x] \frac{2\pi}d)\\
&\leq - \cos(\frac{d-1}{2}\cdot\frac{2\pi}{d})= \cos(\frac{\pi}{d}).
\end{align*}
But, one can see that exactly for $d\geq 3$, $1-\frac{9}{2d^2} \geq \cos(\frac{\pi}{d})$, which is contradiction.
This shows that $[\vec z-\vec x, \vec y-\vec x]=0$.
 % and therefore $T$ is a subset of an affine isotropic subspace.
\end{proof}

\begin{cor}\label{cor:point op uncertainty}
Let $d$ be an odd integer and $\psi$ a pure state of $n$ qu$d$its.
Then $T\coloneqq\{ \vec x \in\mathcal V_n : w_\psi(\vec x) > d^{-n} \sqrt{1-1/2d^2}\}$ is a subset of an affine totally isotropic subspace.
\end{cor}

We now prove the main result of this section:

\restateThmRobustHudson
\begin{proof}
Suppose that $\sn(\psi)\leq\eps$.
Then $\lVert\psi\rVert_W\leq1+2\eps$ and we find from \cref{eq:average vs wigner} that
\begin{equation*}
  \sum_{\vec x} q_\psi(\vec x)^2 \geq \frac 1 {d^n (1+2\eps)^2},
\end{equation*}
i.e.,
\begin{equation}\label{eq:q_psi mean}
\sum_{\vec x} q_\psi(\vec x) \left( d^{-n}-q_\psi(\vec x) \right)
\leq d^{-n} \left(1 - \frac 1{(1+2\eps)^2}\right)
% &\leq d^{-n} \left(\frac {4\eps + 4\eps^2}{(1+2\eps)^2}\right) \\
% &\leq 4\eps d^{-n} \left(\frac {1 + \eps}{(1+2\eps)^2}\right) \\
\leq 4\eps d^{-n}.
\end{equation}
We would like to show that the probability of the set $T$ from \cref{cor:point op uncertainty} with respect to the probability distribution $q_\psi$ is close to one.
First, though, let us consider
\begin{align*}
\tilde T
= \left\{ \vec x : \lvert w_\psi(\vec x)\rvert > d^{-n}\sqrt{1-1/{2d^2}} \right\}
= \left\{ \vec x : q_\psi(\vec x) > d^{-n}\left(1-1/{2d^2}\right) \right\}
\end{align*}
which is defined just like $T$ but for the absolute value of the Wigner function!
Then, using Markov's inequality and \cref{eq:q_psi mean}, we have
\begin{align*}
\sum_{\vec x \in \tilde T} q_\psi(\vec x)
\geq 1 - \frac{\sum_{\vec x} q_\psi(\vec x) \left( d^{-n}-q_\psi(\vec x) \right)}{d^{-n}\cdot 1/{2d^2}}
\geq 1-8d^2 \eps.
\end{align*}
But then $T$ is likewise a high-probability subset:
\begin{align*}
\sum_{\vec x \in T} q_\psi(\vec x)
\geq \sum_{\vec x \in \tilde T} q_\psi(\vec x) - \sum_{w_\psi(\vec x) < 0} q_\psi(\vec x)
\geq \sum_{\vec x \in \tilde T} q_\psi(\vec x)- \sum_{w_\psi(\vec x) < 0} \lvert w_\psi(\vec x)\rvert \geq 1-8d^2 \eps-\eps,
\end{align*}
where we used $q_\psi(\vec x)=d^n \rvert w_\psi(\vec x)\lvert^2\leq\lvert w_\psi(\vec x)\rvert$.

As a result of \cref{cor:point op uncertainty}, $T$ is a subset of some affine totally isotropic subspace $\vec a+M$.
If $\ket S$ denotes the corresponding stabilizer state then
\begin{align*}
  \lvert\braket{\psi|S}\rvert^2
&= d^n \sum_{\vec x} w_\psi(\vec x) w_S(\vec x)
= \sum_{\vec x\in\vec a+M} w_\psi(\vec x)
= \sum_{\vec x\in T} w_\psi(\vec x) + \sum_{\vec x\in (\vec a + M) \setminus T} w_\psi(\vec x)
\geq \sum_{\vec x\in T} w_\psi(\vec x) - \eps \\
&\geq \sum_{\vec x\in T} q_\psi(\vec x) - \eps
\geq 1-(8d^2+2)\eps > 1-9d^2 \eps.
\end{align*}
In the fifth step we used that, for $\vec x\in T$, $w_\psi(\vec x)=\lvert w_\psi(\vec x)\rvert\geq d^n \rvert w_\psi(\vec x)\lvert^2= q_\psi(\vec x)$.
\end{proof}

%-----------------------------------------------------------------------------
\subsection{Stabilizer testing revisited: minimal number of copies}\label{sec:three copy test}
%-----------------------------------------------------------------------------
We will now revisit stabilizer testing from the perspective of the Wigner function and show that for $d\equiv1,5\pmod6$ it is in fact possible to perform stabilizer testing with just three copies of the state.
This is clearly optimal, since the set of stabilizer states forms a projective $2$-design.

We start with the phase space point operators $A_{\vec x}$ from~\eqref{eq:phase space point op definition},
%which can be explicitly written as
%\begin{align*}
%A_{\vec x} = A_{\vec p,\vec q}
% &= d^{-n} \sum_{\vec y} \omega^{-[\vec x,\vec y]} W_{\vec y}^\dagger \\
% &= d^{-n} \sum_{\vec{p'}, \vec{q'}} \omega^{-[\vec x,(\vec p',\vec q')]} (\omega^{-2^{-1}\vec p'\cdot \vec q'} Z_1^{p'_1} \cdots Z_n^{p'_n} X_1^{q'_1} \cdots X_n^{q'_n})^\dagger \\
% &= d^{-n} \sum_{\vec{p'}, \vec{q'}} \omega^{-[\vec x,(\vec p',\vec q')]} \omega^{2^{-1}\vec p'\cdot \vec q'} X_n^{-q'_n} \cdots X_1^{-q'_1} Z_n^{-p'_n} \cdots Z_1^{-p'_1} \\
% &= d^{-n} \sum_{\vec{p'}, \vec{q'}} \omega^{\vec p \cdot \vec q' - \vec q \cdot \vec p'} \omega^{2^{-1}\vec p'\cdot \vec q'} X_n^{q'_n} \cdots X_1^{q'_1} Z_n^{p'_n} \cdots Z_1^{p'_1} \\
% &= d^{-n} \sum_{\vec a} \sum_{\vec{p'}, \vec{q'}} \omega^{\vec p \cdot \vec q' + (\vec a - \vec q + 2^{-1} \vec q') \cdot \vec p'} \ket{\vec a+\vec q'}\bra{\vec a} \\
% &= \sum_{\vec a} \omega^{2 \vec p \cdot (\vec q - \vec a)} \ket{\vec a+2(\vec q - \vec a)}\bra{\vec a} \\
%&= \sum_{\vec a} \omega^{2 \vec p \cdot (\vec q - \vec a)} \ket{2\vec q - \vec a}\bra{\vec a}.
%\end{align*}
Consider the operator
\begin{align}\label{eq:V for phase space}
V \coloneqq d^{-n} \sum_{\vec x} A_{\vec x}^{\ot 3} = d^{-2n} \sum_{\vec y_1+\vec y_2+\vec y_3=0} W_{\vec y_1} \ot W_{\vec y_2} \ot W_{\vec y_3}.
\end{align}
We remark that it is clear from \cref{eq:clifford phase space} that~$V$ is an element of the commutant.
Moreover, we have the following analog of \cref{lem:hermitian unitary}:

\begin{lem}\label{lem:hermitian unitary for phase space}
For $d\equiv1,5\pmod6$, the operator $V$ defined in~\cref{eq:V for phase space} is a Hermitian unitary.
\end{lem}
\begin{proof}
Since the operators~$A_{\vec x}$ are Hermitian, $V$ is Hermitian as well.
Thus it remains to prove that $V$ is unitary.
For this we compute:
\begin{align*}
V^2
&= d^{-4n} \sum_{\vec y_1+\vec y_2+\vec y_3=0 ,\,\vec z_1+\vec z_2+\vec z_3=0 }{ W_{\vec y_1} W_{\vec z_1}\ot W_{\vec y_2}  W_{\vec z_2} \ot W_{\vec y_3}W_{\vec z_3}}\\
&=d^{-4n} \sum_{\vec a_1+\vec a_2+\vec a_3=0 ,\,\vec b_1+\vec b_2+\vec b_3=0 }{W_{\vec a_1} \ot W_{\vec a_2} \ot W_{\vec a_3}}\omega^{\frac18 [\vec a_1 +\vec b_1,\vec a_1-\vec b_1] +\frac18 [\vec a_2 +\vec b_2,\vec a_2-\vec b_2] +\frac18 [\vec a_3 +\vec b_3,\vec a_3-\vec b_3] }\\
&=d^{-4n} \sum_{\vec a_1+\vec a_2+\vec a_3=0 ,\,\vec b_1,\vec b_2 }{W_{\vec a_1} \ot W_{\vec a_2} \ot W_{\vec a_3}}\omega^{- \frac14[\vec a_1 -\vec a_3,\vec b_1] -\frac14[\vec a_2 -\vec a_3,\vec b_2] }\\
&= \sum_{\vec a_1+\vec a_2+\vec a_3=0  }{\delta_{\vec a_1,\vec a_3}\delta_{\vec a_2,\vec a_3}\,\,W_{\vec a_1} \ot W_{\vec a_2} \ot W_{\vec a_3}}=I,
\end{align*}
where, for the second equality, we used the change of variables $\vec a_i=\vec y_i+\vec z_i$ and $\vec b_i=\vec y_i-\vec z_i$.
%For this we expand~$V$ in the computational basis:
%\begin{align*}
%V
%&= d^{-n}\sum_{\vec x} A_{\vec x}^{\ot 3}
%= d^{-n}\sum_{\vec p, \vec q,\vec a_1,\vec a_2,\vec a_3} \omega^{2\vec p\cdot(3\vec q - \vec a_1 -\vec a_2 - \vec a_3)} \ket{2\vec q-\vec a_1,2\vec q-\vec a_2,2\vec q-\vec a_3}\bra{\vec a_1,\vec a_2,\vec a_3}\\
%&= d^{-n}\sum_{\vec q,\vec a_1,\vec a_2,\vec a_3} d^{n}\, \delta_{3\vec q, \vec a_1+,\vec a_2 + \vec a_3}  \ket{2\vec q-\vec a_1,2\vec q-\vec a_2,2\vec q-\vec a_3}\bra{\vec a_1,\vec a_2,\vec a_3}\\
%&= \sum_{\vec a_1,\vec a_2,\vec a_3} \ket{3^{-1}\left(2\vec a_2+2\vec a_3-\vec a_1\right),3^{-1}\left(2\vec a_1+2\vec a_3-\vec a_2\right),3^{-1}\left(2\vec a_1+2\vec a_2-\vec a_3\right) }\bra{\vec a_1,\vec a_2,\vec a_3},
%\end{align*}
%where $3^{-1}$ denotes the inverse of $3$ mod $d$.
%Note that the above operator permutes the computational basis states, because the map
%\[
%\left({\vec a_1,\vec a_2,\vec a_3}\right)\rightarrow \left( {3^{-1}\left(2\vec a_2+2\vec a_3-\vec a_1\right),3^{-1}\left(2\vec a_1+2\vec a_3-\vec a_2\right),3^{-1}\left(2\vec a_1+2\vec a_2-\vec a_3\right) } \right),
%\]
%is a linear invertible map mod $d$.
%We conclude that~$V$ is unitary.
\end{proof}

We now consider the binary POVM measurement with accepting projector
\begin{align*}
  \Pi_\text{accept} = \frac12 (I+V).
\end{align*}

\restateThmMainThreeCopies
\begin{proof}
We first note that
\begin{align*}
  p_\text{accept}
=\tr\left[\Pi_\text{accept} \psi^{\ot 3}\right]
=\frac12 \left( 1 + d^{-n} \sum_{\vec x} \tr\left[A_{\vec x}^{\ot 3} \psi^{\ot 3}\right] \right)
=\frac12 \left( 1 + d^{2n} \sum_{\vec x} w^3_\psi(\vec x) \right),
\end{align*}
where $w_\psi$ denotes the Wigner function defined in \cref{eq:wigner function}.
It is clear from \cref{eq:wigner function odd stabilizer} that if $\psi$ is a stabilizer state then $p_\text{accept}=1$.

Now assume that $\psi$ is an arbitrary pure state.
Since $q(\vec x) = d^n w_\psi^2(\vec x)$ is a probability distribution, we can rewrite the above as
\begin{align*}
  \sum_{\vec x} q_\psi(\vec x) \left(1 - d^n w_\psi(\vec x) \right) = 2 \left(1 - p_\text{accept} \right).
\end{align*}
Moreover, $q(\vec x)\leq d^{-n}$ for all $\vec x$, so we can use Markov's probability for the set
\begin{align*}
  T = \left\{ \vec x \in \mathcal V_n : w_\psi(\vec x) > d^{-n} \sqrt{1 - 1/2d^2} \right\}
\end{align*}
to see that
\begin{align}\label{eq:phase space T bound}
  \sum_{\vec x \in T} q_\psi(\vec x)
\geq 1 - \frac{2 \left(1 - p_\text{accept} \right)}{1 - \sqrt{1 - 1/2d^2}}
\geq 1 - 8d^2 \left(1 - p_\text{accept} \right).
\end{align}
We now argue similarly as in the proof of the robust Hudson theorem (\cref{thm:robust hudson}).
From~\cref{cor:point op uncertainty} below we know that there exists an affine Lagrangian subspace $\vec a+M$ that contains $T$.
Let $\ket S$ denote the corresponding stabilizer state.
Then,
\begin{align*}
  \lvert\braket{\psi|S}\rvert^2
= d^n \sum_{\vec x} w_\psi(\vec x) w_S(\vec x)
= \sum_{\vec x\in\vec a+M} w_\psi(\vec x)
= \sum_{\vec x\in T} w_\psi(\vec x) + \sum_{\vec x\in (\vec a+M)\setminus T} w_\psi(\vec x)
\end{align*}
For $x\in T$, $w_\psi(\vec x)\geq0$ and so $w_\psi(\vec x) \geq q_\psi(\vec x)$.
Thus the first sum can be lower-bounded by using \cref{eq:phase space T bound}.
For the second sum, we note that \cref{eq:phase space T bound} also implies that $1 - 8d^2 \left(1 - p_\text{accept} \right) \leq d^{-n} \lvert T\rvert$, since $q_\psi(\vec x)\leq d^{-n}$, and so
\begin{align*}
  \sum_{\vec x\in (\vec a+M)\setminus T} w_\psi(\vec x)
  \geq -d^{-n} \lvert(\vec a+M) \setminus T\rvert
  = -d^{-n} \left(d^n - \lvert T\rvert\right)
  % = -\left(1 - d^{-n} \lvert T\rvert\right)
  \geq -8d^2 \left(1 - p_\text{accept} \right).
\end{align*}
Together, we obtain that $\lvert\braket{\psi|S}\rvert^2 \geq 1 - 16d^2 \left(1 - p_\text{accept} \right)$, or $p_\text{accept} \leq 1 - \eps^2/16d^2$, which is what wanted to show.
\end{proof}

%-----------------------------------------------------------------------------
\addcontentsline{toc}{section}{Acknowledgments}
\section*{Acknowledgments}
We thank Paula Belzig, Raja Damanik, Patrick Hayden, Ashley Montanaro, Felipe Montealegre, Gabriele Nebe, Goetz Pfander, Zeljka Stojanac, Mario Szegedy, and Ronald de Wolf for interesting discussions.
MW acknowledges support by the NWO through Veni grant no.~680-47-459 and grant OCENW.KLEIN.267, and AFOSR grant FA9550-16-1-0082.
DG's work has been supported by the Excellence Initiative of the German Federal and State Governments (Grant ZUK 81), the ARO under contract W911NF-14-1-0098 (Quantum Characterization, Verification, and Validation),
Universities Australia and DAAD's Joint Research Co-operation Scheme (using funds provided by the German Federal Ministry of Education and Research),
the DFG (SPP1798 CoSIP, project B01 of CRC 183), and
Germany's Excellence Strategy -- Cluster of Excellence \emph{Matter and Light for Quantum Computing (ML4Q)}, EXC2004/1.
This research was supported in part by the National Science Foundation under Grant No.~NSF PHY11-25915.
%-----------------------------------------------------------------------------

\addcontentsline{toc}{section}{References}
\bibliographystyle{alphaurl}
\bibliography{stabtest}

\end{document}